\documentclass[10pt]{article}
\usepackage[margin=1in]{geometry}
\usepackage[all=normal,bibliography=tight]{savetrees}

\usepackage{microtype}
\usepackage{wrapfig}

\newcommand{\uselipics}{no}
\usepackage{ifthen}

\newcommand{\iflipics}[2]{\ifthenelse{\equal{\uselipics}{yes}}{#1}{#2}}
\newcommand{\onlylipics}[1]{\iflipics{#1}{}}
\newcommand{\onlyfull}[1]{\iflipics{}{#1}}

\usepackage{amsthm,amssymb,amsmath}
\newtheorem{theorem}{Theorem}[section]
\newtheorem{lemma}[theorem]{Lemma}

\newtheorem{observation}[theorem]{Observation}
\newtheorem{corollary}[theorem]{Corollary}
\usepackage{enumerate}
\usepackage{url}
\title{A deterministic polynomial kernel for Odd Cycle Transversal and Vertex Multiway Cut in planar graphs}

\author{
  Bart M.P. Jansen\thanks{Eindhoven University of Technology, \texttt{b.m.p.jansen@tue.nl}. 
    Supported by NWO Gravitation grant ``Networks''.}
  \and Marcin Pilipczuk\thanks{Institute of Informatics, University of Warsaw, \texttt{malcin@mimuw.edu.pl}. Supported by the ``Recent trends in kernelization: theory and experimental evaluation'' project, carried out within the Homing programme of the Foundation for Polish Science co-financed by the European Union under the European Regional Development Fund.}
  \and Erik Jan van Leeuwen\thanks{Department of Information \& Computing Sciences, Utrecht University, \texttt{e.j.vanleeuwen@uu.nl}.}}

\newcommand{\green}{\mathbf{A}}
\newcommand{\black}{\mathbf{B}}
\newcommand{\sparseexp}{{212}}
\newcommand{\sparseG}{\widehat{G}}
\newcommand{\Oh}{\mathcal{O}}
\newcommand{\cost}{\mathrm{cost}}

\newcommand{\BB}{\mathcal{B}}
\newcommand{\perim}{\mathrm{perim}}
\newcommand{\coref}{f_{\mathrm{core}}}
\newcommand{\pegs}{\mathbf{P}}
\newcommand{\Gclose}{G_{\mathrm{close}}}
\newcommand{\leftpeg}{p_{\leftarrow}}
\newcommand{\rightpeg}{p_{\rightarrow}}

\newcommand{\radial}{\mathcal{R}}
\newcommand{\overlay}{\mathcal{L}}

\usepackage{todonotes}
\usepackage{xspace}

\theoremstyle{plain}
\newtheorem{claim}[theorem]{Claim}
\newtheorem{definition}[theorem]{Definition}
\let\plainqed\qedsymbol
\newcommand{\claimqed}{$\lrcorner$}
\newenvironment{claimproof}{\begin{proof}\renewcommand{\qedsymbol}{\claimqed}}{\end{proof}\renewcommand{\qedsymbol}{\plainqed}}

\newcommand{\defparproblem}[4]{\par
 \vspace{3mm}
\noindent\fbox{
 \begin{minipage}{0.96\textwidth}
 \begin{tabular*}{\textwidth}{@{\extracolsep{\fill}}lr} #1 & {\bf{Parameter:}} #3 \vspace{1mm} \\ \end{tabular*}
 {\textbf{Input:}} #2
	\vspace{1mm}\\%
 {\textbf{Question:}} #4
 \end{minipage}
 }
 \vspace{3mm}
\par
}

\newcommand{\pmwc}{\textsc{Planar Vertex Multiway Cut}\xspace}
\newcommand{\vtxplan}{\textsc{Vertex Planarization}\xspace}
\newcommand{\vtxplandis}{\textsc{Disjoint Vertex Planarization}\xspace}
\newcommand{\poct}{\textsc{Odd Cycle Transversal}}
\newcommand{\ppoct}{\textsc{Plane Odd Cycle Transversal}}

\newcommand{\ptj}{\textsc{Bipartite Steiner $T$-join}}
\newcommand{\pptj}{\textsc{Plane Bipartite Steiner $T$-join}}

\newcommand{\mwc}{\textsc{MwC}\xspace}

\newcommand{\tree}{\mathbb{T}}

\newcommand{\opindex}{\mathrm{\textsc{idx}}}
\renewcommand{\int}{\mathrm{\textsc{int}}}

\newcommand{\longversion}[1]{\onlyfull{#1}}
\newcommand{\shortversion}[1]{\onlylipics{#1}}

\begin{document}

\maketitle

\begin{abstract}
We show that \textsc{Odd Cycle Transversal} and \textsc{Vertex Multiway Cut} admit deterministic 
polynomial kernels when restricted to planar graphs and parameterized 
by the solution size. This answers a question of Saurabh.
On the way to these results, we provide an efficient sparsification routine in the flavor of the sparsification routine used for the \textsc{Steiner Tree} problem in planar graphs (FOCS 2014). It differs from the previous work because it preserves the existence of low-cost subgraphs that are not necessarily Steiner trees in the original plane graph, but structures that turn into (supergraphs of) Steiner trees after adding all edges between pairs of vertices that lie on a common face. We also show connections between \textsc{Vertex Multiway Cut} and the \textsc{Vertex Planarization} problem, where the existence of a polynomial kernel
remains an important open problem.

\end{abstract}

\section{Introduction}

Kernelization provides a rigorous framework within the paradigm of parameterized complexity 
to analyze preprocessing routines for various combinatorial problems. 
A \emph{kernel} of size $g$ for a parameterized problem $\Pi$ and a computable
function $g$ is a polynomial-time algorithm that reduces an input instance $x$ with parameter
$k$ of problem $\Pi$ to an equivalent one with size and parameter value bounded
by $g(k)$. 
Of particular importance are \emph{polynomial kernels}, where the function $g$ is required
to be a polynomial, that are interpreted as theoretical tractability of preprocessing for
the considered problem $\Pi$.
Since a kernel (of any size) for a decidable problem implies fixed-parameter tractability (FPT)
of the problem at hand, the question whether a \emph{polynomial} kernel exists 
became a ``standard'' tractability question one asks about a problem already known to be
FPT, and serves as a further finer-grained distinction criterion 
between FPT problems.

In the recent years, a number of kernelization techniques emerged, including the
bidimensionality framework for sparse graph classes~\cite{FominLST10}
and the use of representative sets for graph separation problems~\cite{KratschW12}. 
On the hardness side, a lower bound framework against polynomial kernels has been
developed and successfully applied to a multitude of problems~\cite{BodlaenderDFH09,DellM14,Drucker15,FortnowS11}.
For more on kernelization, we refer to the survey~\cite{LokshtanovMS12} for background
and to the appropriate chapters of the textbook~\cite{pa-book} for basic definitions
and examples.

For this work, of particular importance are polynomial kernels for graph separation problems.
The framework for such kernels developed by Kratsch and Wahlstr\"{o}m in~\cite{KratschW12,KratschW14},  relies
on the notion of \emph{representative sets} in linear matroids, especially in gammoids.
Among other results, the framework provided a polynomial kernel for 
\textsc{Odd Cycle Transversal} and for \textsc{Multiway Cut} with a constant number of terminals. 
However, all kernels for graph separation problems based on representative sets are randomized,
due to the randomized nature of all known polynomial-time algorithms that obtain a linear
representation of a gammoid. As a corollary, all such kernels have exponentially small probability
of turning an input yes-instance into a no-instance.

The question of \emph{deterministic} polynomial kernels for the cut problems
that have \emph{randomized} kernels due to the representative sets framework remains widely
open. Saket Saurabh, at the open problem session during the Recent Advances
in Parameterized Complexity school (Dec 2017, Tel Aviv)~\cite{rapc}, asked whether a deterministic polynomial kernel for \textsc{Odd Cycle Transversal} exists when
the input graph is planar. In this paper, we answer this question affirmatively, and prove an analogous result for the \textsc{Multiway Cut} problem.

\begin{theorem}\label{thm:main}
\textsc{Odd Cycle Transversal} and \textsc{Vertex Multiway Cut}, when restricted
to planar graphs and parameterized by the solution size, admit deterministic polynomial kernels.
\end{theorem}
Recall that the \textsc{Odd Cycle Transversal} problem, given a graph $G$ and an integer $k$,
asks for a set $X \subseteq V(G)$ of size at most $k$ such that $G \setminus X$ is bipartite.
For the \textsc{Multiway Cut} problem, we consider the \textsc{Vertex Multiway Cut} variant
where, given a graph $G$, a set of terminals $T \subseteq V(G)$, and an integer $k$,
we ask for a set $X \subseteq V(G) \setminus T$ of size at most $k$ such that every connected
component of $G \setminus X$ contains at most one terminal. In other words, we focus on the vertex-deletion
variant of \textsc{Multiway Cut} with undeletable terminals.
In both cases, the allowed deletion budget, $k$, is our parameter. (A deterministic polynomial kernel for \textsc{Edge Multiway Cut} in planar graphs is known~\cite[Theorem 1.4]{pst-kernel}.)

Note that in general graphs, \textsc{Vertex Multiway Cut} admits a randomized polynomial kernel with $\Oh(k^{|T|+1})$ terminals~\cite{KratschW12}, and whether one can remove the dependency on $|T|$ from the exponent is a major open question in the area. Theorem~\ref{thm:main} answers this question affirmatively in the special case of planar graphs.

Our motivation stems not only from the aforementioned question of Saurabh~\cite{rapc}, but 
also from a second, more challenging question of a polynomial kernel for the
\vtxplan{} problem. Here, given a graph $G$ and an integer $k$, one asks for a set
$X \subseteq V(G)$ of size at most $k$ such that $G \setminus X$ is planar. 
For this problem, an involved $2^{\Oh(k \log k)} \cdot n$-time fixed-parameter
algorithm is known~\cite{JansenLS14}, culminating a longer 
line of research~\cite{JansenLS14,Kawarabayashi09,MarxS12}.
The question of a polynomial kernel for the problem has not only been posed by Saurabh
during the same open problem session~\cite{rapc}, but also comes out naturally in another
line of research concerning vertex-deletion problems to minor-closed graph classes.

Consider a minor-closed graph class $\mathcal{G}$. By the celebrated Robertson-Seymour theorem,
the list of minimal forbidden minors $\mathcal{F}$ of $\mathcal{G}$ is finite, i.e., there is 
a finite set $\mathcal{F}$ of graphs such that a graph $G$ belongs to $\mathcal{G}$
if and only if $G$ does not contain any graph from $\mathcal{F}$ as a minor. 
The \textsc{$\mathcal{F}$-Deletion} problem, given a graph $G$ and an integer $k$,
asks to find a set $X \subseteq V(G)$ of size at most $k$ such that $G \setminus X$ has no minor
belonging to $\mathcal{F}$, i.e., $G \setminus X \in \mathcal{G}$. 
If $\mathcal{F}$ contains a planar graph or, equivalently, $\mathcal{G}$ has bounded treewidth,
then the parameterized and kernelization complexity of the \textsc{$\mathcal{F}$-Deletion}
problem is well understood~\cite{FominLMS12}.
However, our knowledge is very partial
in the other case, when $\mathcal{G}$ contains all planar graphs.
The understanding of this general problem has been laid out as one of the future research directions in a monograph
of Downey and Fellows~\cite{DowneyF13}.
The simplest not fully understood case is when $\mathcal{G}$ is exactly the set of planar graphs,
that is, $\mathcal{F} = \{K_{3,3}, K_5\}$, and the \textsc{$\mathcal{F}$-Deletion} becomes
the \vtxplan{} problem.
The question of a polynomial kernel or a $2^{\Oh(k)} \cdot n^{\Oh(1)}$-time FPT algorithm for
\vtxplan{} remains open~\cite{Uniform,rapc}.

\iflipics{In Appendix~\ref{app:lb}}{In Section~\ref{sec:lb}}, we observe that there is a simple
polynomial-time reduction from \pmwc{} to \vtxplan{} that keeps the parameter $k$
unchanged. If \vtxplan{} would admit a polynomial kernel, then our reduction would transfer
the polynomial kernel back to \pmwc{}. 
In the presence of Theorem~\ref{thm:main}, such an implication is trivial, but the reduction itself
serves as a motivation: a polynomial kernel for \pmwc{} should be easier than for \vtxplan{}, and one should begin with the first before proceeding to the latter.
Furthermore, we believe the techniques developed in this work can be of use for
the more general \vtxplan{} case.

\iflipics{%
\begin{wrapfigure}{r}{5cm}
\centering
\includegraphics[width=.9\linewidth]{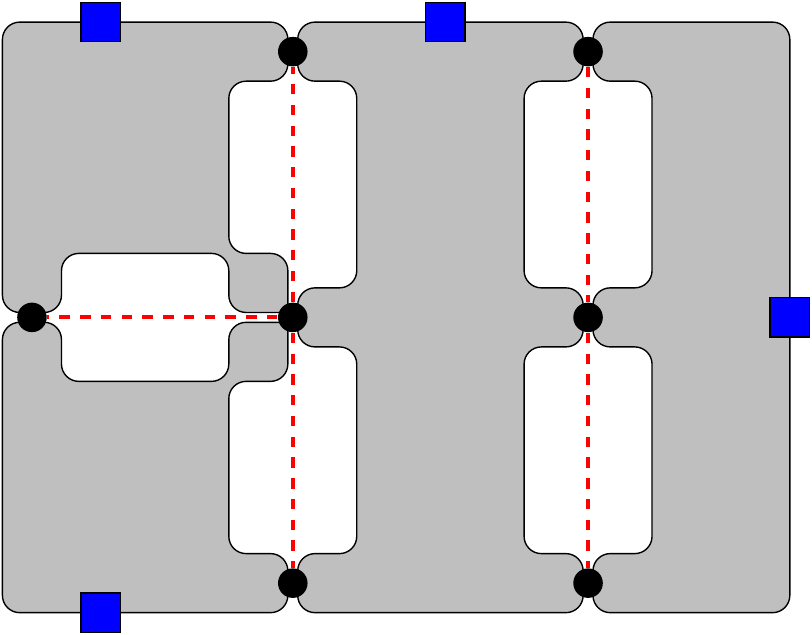}
\caption{When all terminals (blue squares) lie on the infinite face, a solution to \textsc{Vertex Multiway Cut} (black circles) becomes a Steiner forest (red dashed connections) in the overlay graph.}
\label{fig:intro}
\end{wrapfigure}}{%
\begin{figure}[htb]
\centering
\includegraphics[width=.4\linewidth]{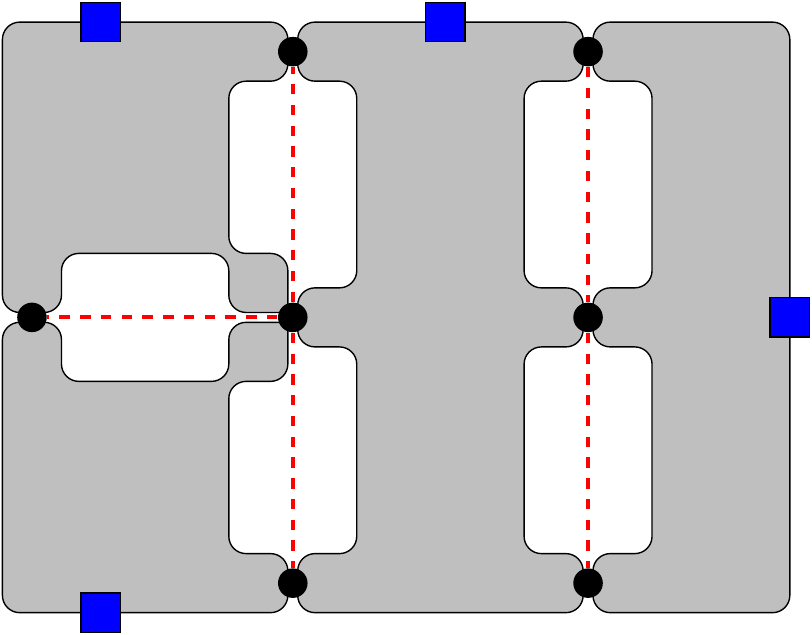}
\caption{When all terminals (blue squares) lie on the infinite face, a solution to \textsc{Vertex Multiway Cut} (black circles) becomes a Steiner forest (red dashed connections) in the overlay graph.}
\label{fig:intro}
\end{figure}}

\subparagraph{Techniques}

On the technical side, our starting point is the toolbox of~\cite{pst-kernel} that provides a polynomial kernel for
\textsc{Steiner Tree} in planar graphs, parameterized by the number of edges of the solution.
The main technical result of~\cite{pst-kernel} is a sparsification routine that, given
a connected plane graph $G$ with infinite face surrounded by a simple cycle $\partial G$,
provides a subgraph of $G$ of size polynomial in the length of $\partial G$ that,
for every $A \subseteq V(\partial G)$, preserves an optimal Steiner tree connecting $A$.

Both \textsc{Odd Cycle Transversal} and \textsc{Vertex Multiway Cut} in a plane graph $G$
translate into Steiner forest-like questions in the \emph{overlay graph} $\overlay(G)$ of $G$: a supergraph
of $G$ that has a vertex $v_f$ for every face of $G$, adjacent to every vertex
of $G$ incident with $f$. 
To see this, consider a special case of \pmwc{} where all terminals lie on the infinite face of the input embedded graph.
Then, an optimal solution is a Steiner forest between some tuples of vertices on the outer face lying between the terminals, cf.~Figure~\ref{fig:intro}.
Following~\cite{pst-kernel}, this suggest the following approach to kernelization of vertex-deletion cut problems in planar graphs:
\begin{enumerate}
\item By problem-specific reductions, reduce to the case of a graph of bounded radial diameter.
\item Using the diameter assumption, find a tree in the overlay graph that has size bounded polynomially in the solution size, and that
spans all ``important'' objects in the graph (e.g., neighbors of the terminals in the case of \textsc{Multiway Cut} or odd faces in the case of \textsc{Odd Cycle Transversal}).
\item Cut the graph open along the tree. Using the Steiner forest-like structure of the problem at hand, argue that an optimal solution 
becomes an optimal Steiner forest for some choice of tuples of terminals on the outer face of the cut-open graph.
\item Sparsify the cut-open graph with a generic sparsification routine that preserves optimal Steiner forests, glue the resulting graph back,
  and return it as a kernel.
\end{enumerate}
However, contrary to the \textsc{Steiner tree} problem~\cite{pst-kernel}, these Steiner forest-like questions optimize a different
cost function than merely number of edges, namely \emph{the number of vertices of $G$},
with the ``face'' vertices $v_f \in V(\overlay(G)) \setminus V(G)$ being for free. 
This cost function is closely related to (half of) the number of edges in case of paths and trees
with constant number of leaves, but may diverge significantly in case of trees with high-degree
vertices. 

For this reason, we need an analog of the main technical sparsification routine of~\cite{pst-kernel}
suited for our cost function.
\iflipics{We provide one in Section~\ref{sec:sparse} (statement and overview of the proof) and Appendix~\ref{app:sparse} (full proof).}{We provide one in Section~\ref{sec:sparse}.}
To this end, we re-use most of the intermediate results of~\cite{pst-kernel}, changing
significantly only the final divide\&conquer argument.

The application of the obtained sparsification routine to the case of \textsc{Odd Cycle Transversal},
    presented in Section~\ref{sec:oct},
follows the phrasing of the problem as a $T$-join-like problem in the overlay graph 
due to Fiorini et al~\cite{FioriniHRV2008}. For the sake of reducing the number of odd faces, 
we adapt the arguments of Such\'{y}~\cite{Suchy2017} for \textsc{Steiner tree}.

\iflipics{The arguments for \textsc{Vertex Multiway Cut} are somewhat more involved
and sketched in Section~\ref{sec:mwc-short}. A full version is postponed to Appendix~\ref{app:mwc}.}{The arguments for \textsc{Vertex Multiway Cut} are somewhat more involved
and presented in Section~\ref{sec:mwc}.}
Here, we first use known LP-based rules~\cite{CyganPPW13,GargVY04,Guillemot11a,Razgon11}
to reduce the number of terminals and neighbors of terminals to $\Oh(k)$ and then
use an argument based on outerplanarity layers to reduce the diameter.

\section{Preliminaries}

A finite undirected graph~$G$ consists of a vertex set~$V(G)$ and edge set~$E(G) \subseteq \binom{V(G)}{2}$. We denote the open neighborhood of a vertex~$v$ in~$G$ by~$N_G(v)$. For a vertex set~$S \subseteq V(G)$ we define its open neighborhood as~$N_G(S) := \bigcup_{v \in S} N_G(v) \setminus S$. 

For vertex subsets~$X, Y$ of a graph~$G$, we define an~$(X,Y)$-cut as a vertex set~$Z \subseteq V(G) \setminus (X \cup Y)$ such that no connected component of~$G \setminus Z$ contains both a vertex of~$X$ and a vertex of~$Y$. An~$(X,Y)$ cut~$Z$ is \emph{minimal} if no proper subset of~$Z$ is an~$(X,Y)$-cut, and \emph{minimum} if it has minimum possible size. 

\subsection{Planar graphs}
In a connected embedded planar (i.e.~plane) graph $G$, the \emph{boundary walk} of a face $f$ is the unique closed walk in $G$ obtained by going along the face in counter-clockwise direction. Note that a single vertex can appear multiple times on the boundary walk of $f$ and an edge can appear twice if it is a bridge. We denote the number of edges of this walk by $|f|$; note that bridges are counted twice in this definition. The \emph{parity} of a face $f$ is the parity of $|f|$. Then a face is \emph{odd (even)} if its parity is odd (even). The boundary walk of the outer face of $G$ is called the \emph{outer face walk} and denoted $\partial G$.

We define the \emph{radial distance} in plane graphs, based on a measure that allows to hop between vertices incident on a common face in a single step. Formally speaking, a \emph{radial path} between vertices~$p$ and~$q$ in a plane graph~$G$ is a sequence of vertices~$(p=v_0, v_1, \ldots, v_\ell=q)$ such that for each~$i \in [\ell]$, the vertices~$v_{i-1}$ and~$v_i$ are incident on a common face. The \emph{length} of the radial path equals~$\ell$, so that a trivial radial path from~$v$ to itself has length~$0$. The \emph{radial distance} in plane graph~$G$ between~$p$ and~$q$, denoted~$d^\radial_{G}(u,v)$, is defined as the minimum length of a radial $pq$-path. 

For a plane graph~$G$, let~$F(G)$ denote the set of faces of~$G$. For a plane (multi)graph~$G$, an \emph{overlay graph}~$G'$ of~$G$ is a graph with vertex set~$V(G) \cup F(G)$ obtained from~$G$ as follows. For each face~$f \in F(G)$, draw a vertex with identity~$f$ in the interior of~$f$. For each connected component~$C$ of edges incident on the face~$f$, traverse the boundary walk of~$C$ starting at an arbitrary vertex. Every time a vertex~$v$ is visited by the boundary walk, draw a new edge between~$v$ and the vertex representing~$f$, without crossing previously drawn edges. Doing this independently for all faces of~$G$ yields an overlay graph~$G'$. Observe that an overlay graph may have multiple edges between some~$f \in F(G)$ and~$v \in V(G)$, which occurs for example when~$v$ is incident on a bridge that lies on~$f$. The resulting plane multigraph~$G'$ is in general not unique, due to different homotopies for how edge bundles may be routed around different connected components inside a face. For our purposes, these distinctions are never important. We therefore write~$\overlay(G)$ to denote an arbitrary fixed overlay graph of~$G$. Observe that~$F(G)$ forms an independent set in~$\overlay(G)$.

Apart from the overlay graph, we will also use the related notion of \emph{radial graph} (also known as face-vertex incidence graph). A \emph{radial graph} of a connected plane graph~$G$ is a plane multigraph~$\radial(G)$ obtained from~$\overlay(G)$ by removing all edges with both endpoints in~$V(G)$. Hence a radial graph of~$G$ is bipartite with vertex set~$V(G) \cup F(G)$, where vertices are connected to the representations of their incident faces. From these definitions it follows that~$\overlay(G)$ is the union of~$G$ and~$\radial(G)$, which explains the terminology.

\iflipics{The proof of the following simple but useful lemma can be found in Appendix~\ref{app:prelims}.}{We need also the following simple but useful lemma.}
\begin{lemma} \label{lem:construct:steinertree}
Let $G$ be a connected graph, let $T \subseteq V(G)$ and assume that for each vertex $v \in V(G)$, there is a terminal $t \in T$ that can reach $v$ by a path of at most $K$ edges. Then $G$ contains a Steiner tree of at most~$(2K+1)(|T|-1)$ edges on terminal set~$T$, which can be computed in linear time.
\end{lemma}
\onlyfull{%
\begin{proof}
Observe that there exists a spanning forest in~$G$ where each tree is rooted at a vertex of~$T$, and each tree has depth at most~$K$. Such a spanning forest can be computed in linear time by a breadth-first search in~$G$, initializing the BFS-queue to contain all vertices of~$T$ with a distance label of~$0$. Consider the graph~$H$ obtained from~$G$ by contracting each tree into the terminal forming its root. Since~$G$ is connected, $H$ is connected as well. An edge~$t_1t_2$ between two terminals in~$H$ implies that in~$G$ there is a vertex in the tree of~$t_1$ adjacent to a vertex of the tree of~$t_2$. So for each edge in~$H$, there is a path between the corresponding terminals in~$G$ consisting of at most~$2K+1$ edges.

Compute an arbitrary spanning tree of the graph~$H$, which has~$|T|-1$ edges since~$H$ has~$|T|$ vertices. As each edge of the tree expands to a path in~$G$ between the corresponding terminals of length at most~$2K+1$, it follows that~$G$ has a connected subgraph~$F$ of at most~$(2K+1)(|T|-1)$ edges that spans all terminals~$T$. To eliminate potential cycles in~$F$, take a spanning subtree of~$F$ as the desired Steiner tree.
\end{proof}}


\begin{lemma}[{\cite[Lemma 1]{Jansen10}}] \label{lemma:bipartite:neighborhood:count}
Let~$G$ be a planar bipartite graph with bipartition~$V(G) = X \uplus Y$ for~$X \neq \emptyset$. If all distinct~$u,v \in Y$ satisfy~$N_G(u) \not \subseteq N_G(v)$, then~$|Y| \leq 5|X|$.
\end{lemma}


\section{Sparsification}\label{sec:sparse}

\subsection{Overview}


A \emph{plane partitioned graph} is an undirected multigraph $G$, together with a fixed embedding in the plane and
a fixed partition $V(G) = \green(G) \uplus \black(G)$ where $\green(G)$ is an independent set. 
Consider a subgraph $H$ of a plane partitioned graph $G$.
The \emph{cost} of $H$ is defined as $\cost(H) := |V(H) \cap \black(G)|$, that is, we pay
for each vertex of $H$ in the part $\black(G)$. 
We say that $H$ \emph{connects} a subset $A \subseteq V(G)$ if
$A \subseteq V(H)$ and $A$ is contained in a single connected component of $H$.

Our main sparsification routine is the following.
\begin{theorem}\label{thm:sparse}
Given a connected plane partitioned graph $G$, one can in time
$|\partial G|^{\Oh(1)} \cdot \Oh(|G|)$ find a subgraph $\sparseG$ in $G$,
  with the following properties:
  \begin{enumerate}
  \item $\sparseG$ contains all edges and vertices of $\partial G$,
  \item $\sparseG$ contains $\Oh(|\partial G|^\sparseexp)$ edges,
  \item for every set $A \subseteq V(\partial G)$
    there exists a subgraph $H$ of $\sparseG$ that connects $A$ 
    and has minimum possible cost among all subgraphs of $G$ that connect $A$.
  \end{enumerate}
\end{theorem}

In the subsequent sections, given a connected plane graph $G$, we will apply Theorem~\ref{thm:sparse} to a graph $G'$ that is either the overlay graph of $G$ without the vertex corresponding to the outer face, or the radial graph of $G$.
In either case, $\green(G') = V(G') \setminus V(G)$ is the set of face vertices and $\black(G') = V(G)$, i.e., we pay for each ``real'' vertex, not a face one.
If the studied vertex-deletion graph separation problem in $G$ turns into some Steiner problem in $G'$, then we may hope to apply the sparsification routine
of Theorem~\ref{thm:sparse}.


After this brief explanation of the motivation of the statement of Theorem~\ref{thm:sparse}, we proceed with an overview of its proof. 
We closely follow the divide\&conquer approach of the polynomial kernel for \textsc{Steiner Tree} in planar graphs~\cite{pst-kernel}. 

We adopt the notation of (strictly) enclosing from~\cite{pst-kernel}.
For a closed curve $\gamma$ on a plane, a point~$p \notin \gamma$ is strictly enclosed by~$\gamma$ if~$\gamma$ is not continuously retractable to a single point in the plane punctured at~$p$.
A point $p$ is \emph{enclosed} by $\gamma$ if it is strictly enclosed or lies on $\gamma$. The notion of (strict) enclosure naturally extends to vertices, edges, and faces of a plane graph $G$
being (strictly) enclosed by $\gamma$; here a face (an edge) is strictly enclosed by $\gamma$ if every interior point of a face (every point on an edge except for the endpoints, respectively) 
is strictly enclosed.
We also extend this notion to (strict) enclosure by a closed walk $W$ in a plane graph $G$ in a natural manner. 
Note that this corresponds to the natural notion of (strict) enclosure if $W$ is a cycle or, more generally, a closed walk without self-intersections. 

We start with restricting the setting to $G$ being bipartite and $\partial G$ being a simple cycle.
Theorem~\ref{thm:sparse} follows from Lemma~\ref{lem:sparse} by simple manipulations\onlylipics{, and
its proof is deferred to Appendix~\ref{app:sparse}}.
\begin{lemma}\label{lem:sparse}
The statement of Theorem~\ref{thm:sparse} is true in the restricted
setting with $G$ being a connected bipartite simple graph with $\partial G$ being a simple cycle 
and $\green(G)$ being one of the bipartite color classes (so that $\black(G)$ is an independent set as well).
\end{lemma}
We now sketch the proof of Lemma~\ref{lem:sparse}. \onlylipics{The full proof can be found in Appendix~\ref{app:sparse}.}

First observe that the statement of Lemma~\ref{lem:sparse} is well suited for a recursive divide\&conquer algorithm.
As long as $|\partial G|$ is large enough, we can identify a subgraph~$S$ of~$G$ such that:
\begin{enumerate}
\item The number of edges of $S$ is $\Oh(|\partial G|)$;\label{sparse-cond:len}
\item For every set $A \subseteq V(\partial G)$ there exists a subgraph $H$ of $G$ that connects $A$, has minimum possible cost among all subgraphs of $G$ that connect $A$,
	and for every finite face $f$ of $S \cup \partial G$, if $G_f$ is the subgraph of $G$ consisting of the edges and vertices embedded within the closure of~$f$, then one of the following holds:\label{sparse-cond:small}
\begin{enumerate}
\item $|\partial G_f| \leq (1-\delta) |\partial G|$ for some universal constant $\delta > 0$;\label{sparse-cond:small:mult}
\item $H$ does not contain any vertex of degree more than $2$ that is strictly inside $f$.\label{sparse-cond:small:paths}
\end{enumerate}
\end{enumerate}
%
%
Similarly as in the case of~\cite{pst-kernel}, we show that such a subgraph $S$ is good for recursion. 
First, we insert $S$ into the constructed sparsifier $\sparseG$. Second, we recurse on $G_f$ for every finite face $f$ of $S \cup \partial G$ that satisfies Point~\ref{sparse-cond:small:mult}.
Third, for every other finite face $f$ (i.e., one satisfying Point~\ref{sparse-cond:small:paths}), we insert into $\sparseG$ a naive shortest-paths sparsifier: for every two vertices $u_1,u_2 \in V(\partial G_f)$, we insert into $\sparseG$ a minimum-cost path between $u_1$ and $u_2$ in $G_f$. 
Property~\ref{sparse-cond:len} together with the multiplicative progress on $|\partial G|$ in Point~\ref{sparse-cond:small:mult} ensure that the final size of $\sparseG$ is polynomial in $|\partial G|$,
with the exponent of the polynomial bound depending on $\delta$ and the constant hidden in the big-$\Oh$ notation in Property~\ref{sparse-cond:len}.

The main steps of constructing $S$ are the same as in~\cite{pst-kernel}. First, we try minimum-size (i.e., with minimum number of edges, as opposed to minimum-cost) Steiner trees
for a constant number of terminals on $\partial G$. If no such trees are found, the main technical result of~\cite{pst-kernel} shows that one can identify a cycle $C$ in $G$ of length $\Oh(|\partial G|)$
with the guarantee that for any choice of $A \subseteq V(\partial G)$, there exists a minimum-size Steiner tree connecting $A$ that does not contain any Steiner point strictly inside $C$. 
In~\cite{pst-kernel} such a cycle is used to construct a desired subgraph $S$ with the inside of $C$ being a face satisfying the Steiner tree analog of Point~\ref{sparse-cond:small:paths}.
In the case of Lemma~\ref{lem:sparse}, we need to perform some extra work here to show that --- by some shortcutting tricks and adding some slack to the constants --- one can construct such a cycle
$C'$ with the guarantee that the face $f$ inside $C'$ satisfies exactly the statement of Point~\ref{sparse-cond:small:paths}: that is, no ``Steiner points'' with regards to minimum-\emph{cost} trees,
not minimum-\emph{size} ones.

In other words, the extra work is needed to at some point switch from ``minimum-size'' subgraphs (treated by~\cite{pst-kernel}) to ``minimum-cost'' ones (being the main focus of Lemma~\ref{lem:sparse}).
 In our proof, we do it as late as possible, trying to re-use as much of the technical details of~\cite{pst-kernel} as possible. 
Observe that for a path $H$ in $G$, the cost of $H$ equals $|E(H)|/2$ up to an additive $\pm \frac{1}{2}$ error.
Similarly, for a tree $H$ with a constant number of leaves, the cost of $H$ is $|E(H)|/2$ up to an additive error bounded by a constant. 
Hence, as long as we focus on paths and trees with bounded number of leaves, the ``size'' and ``cost'' measures are roughly equivalent. 
However, if a tree $H$ in $G$ contains a high-degree vertex $v \in \black(G)$, the cost of $H$ may be much smaller than half of the number of edges of $H$: a star
with a center in $\black(G)$ has cost one and arbitrary number of edges. For this reason, the final argument of the proof of Lemma~\ref{lem:sparse} that constructs the aforementioned
cycle $C'$ using the toolbox of~\cite{pst-kernel} needs to be performed with extra care (and some sacrifice on the constants, as compared to~\cite{pst-kernel}).

\subsection{From Lemma~\ref{lem:sparse} to Theorem~\ref{thm:sparse}}

We start with the formal proof of Theorem~\ref{thm:sparse} from Lemma~\ref{lem:sparse}.

\begin{proof}[Proof of Theorem~\ref{thm:sparse}.]
Let $G$ be a connected plane partitioned graph $G$. First, for every edge $uv \in E(G)$
with $u,v \in \black(G)$, we subdivide it with a new vertex $x \in \green(G)$. 
In this manner, $G$ becomes bipartite with $\green(G)$ and $\black(G)$ being its bipartition
classes, the minimum possible cost of connecting subgraphs does not change (since the cost
does not count vertices in $\green(G)$), and the length of $\partial G$
at most doubled (so the bound $\Oh(|\partial G|^\sparseexp)$ remains the same).

Second, if $\partial G$ is not a simple cycle, then the graph induced by $\partial G$
is a collection $\mathcal{C}_{\geq 3}$ of simple cycles of length at least three,
a collection $\mathcal{C}_2$ of cycles of length two, 
and a set $C_1$ of bridges. For every $C \in \mathcal{C}_{\geq 3} \cup \mathcal{C}_2$,
let $G_C$ be the subgraph of $G$ enclosed by $C$. Note that $\partial G_C = C$.

For every $C \in \mathcal{C}_{\geq 3}$, we turn $G_C$ into a simple graph $G_C'$
by dropping multiple edges, but never dropping an edge of $C$. That is, for every multiple
edge $e$ of $G_C$, we delete all but one of its multiple copies, but we always keep the one
contained in $C$ if it exists. Since $C$ is a simple cycle of length at least three, no two
edges of $C$ connect the same pair of vertices. Consequently, $C = \partial G_C'$
and $G_C'$ satisfies the conditions of Lemma~\ref{lem:sparse}. We apply the algorithm
of Lemma~\ref{lem:sparse}, obtaining a graph $\sparseG_C$ of size $\Oh(|C|^\sparseexp)$.

For $C \in \mathcal{C}_2$, observe that $\sparseG_C := C$ trivially satisfies the conditions
of Theorem~\ref{thm:sparse} for $G_C$: the only nontrivial subset~$A \subseteq V(\partial G_C)$ consists of the two vertices on~$C$, for which an edge of~$C \subseteq \sparseG_C$ forms a minimum-cost connecting subgraph. Consequently, 
$$\sparseG := C_1 \cup \bigcup_{C \in \mathcal{C}_2} C \cup \bigcup_{C \in \mathcal{C}_{\geq 3}} \sparseG_C$$
 satisfies the required properties: the size bound follows from the fact
that $|\partial G| = 2|C_1| + \sum_{C \in \mathcal{C}_2 \cup \mathcal{C}_{\geq 3}} |C|$
while the covering property follows easily from the fact that every minimum-cost
connecting subgraph $H$ splits into minimum-cost connecting subgraphs in each
$G_C$.
\end{proof}

We continue with a formal proof of Lemma~\ref{lem:sparse}.
In Section~\ref{ss:sparse:toolbox} we recall the main technical results of~\cite{pst-kernel} we re-use here. 
In Section~\ref{ss:sparse:C} we show how to find the aforementioned cycle $C'$.
Finally, we wrap up the argument in Section~\ref{ss:sparse:wrapup}.

In this section we implicitly identify every graph $G$ with its set of edges. In particular, $|G|$ is a shorthand for $|E(G)|$
while $H \subseteq G$ means that $H$ is a subgraph of $G$. 
For a tree $H$ and two vertices $x,y \in V(H)$, by $H[x,y]$ we denote the unique path from $x$ to $y$ in $H$.
For a cycle $C$ embedded on a plane and two vertices $x,y \in V(C)$, by $C[x,y]$ we denote the counter-clockwise path along $C$ from $x$ to $y$ (which is a trivial path if $x=y$).

\subsection{Toolbox from the \textsc{Steiner Tree} kernel}\label{ss:sparse:toolbox}

We start with briefly recalling the content of Section~4 of~\cite{pst-kernel}.
A \emph{brick} is a connected plane graph $B$ whose outer face is surrounded by a simple cycle $\partial B$.
A \emph{subbrick} of a plane graph $G$ is a subgraph $B$ that is a brick and consists of all edges of $G$ enclosed by $\partial B$.
A \emph{brick covering} of a plane graph $G$ is a collection $\BB$ of subbricks of $G$ such that every finite face of $G$ is a finite
face of some brick in $\BB$ as well. 
A brick covering $\BB$ is a \emph{brick partition} if every finite face of $G$ is a face of exactly one brick in $\BB$.
The \emph{total perimeter} of a brick covering $\BB$ is $\perim(\BB) := \sum_{B \in \BB} |\partial B|$, and a brick covering $\BB$ is \emph{$c$-short}
if $\perim(\BB) \leq c|\partial G|$.
Furthermore, for a constant $\tau > 0$, a brick covering is \emph{$\tau$-nice} if for every $B \in \BB$ we have $|\partial B| \leq (1-\tau)|\partial G|$.

A connected subgraph $F$ of a plane graph $G$ is called a \emph{connector}, and the vertices of $\partial G$ that are incident to at least one edge of $F$
are \emph{anchors} of $F$. A connector $F$ is \emph{brickable} if the boundary of every finite face of $\partial G \cup F$ is a simple cycle, that is,
these boundaries form subbricks of $B$. 
Thus, a brickable connector $F$ induces a brick partition $\BB_F$ of $G$.
Note that if $F$ is a tree with every leaf lying on $\partial G$, then $F$ is a brickable connector in $G$.
Another important observation is that $\perim(\BB_F) \leq |\partial G| + 2|E(F)|$, so if $|E(F)| \leq c|\partial G|$, then $\BB_F$ is $(2c+1)$-short.
For brevity, we say that a brickable connector $F$ is $c$-short or $\tau$-nice if $\BB_F$ is $c$-short or $\tau$-nice, respectively.

A technical modification of the algorithm of Erickson et al.~\cite{erickson} gives the following.%
\footnote{Technically speaking, Theorem~4.4 of~\cite{pst-kernel} is stated with only $c=3$,
  but a quick inspection of its proof shows that its two ingredients, Lemmata~9.4 and~9.6,
  are already stated for arbitrary $c$.}
\begin{theorem}[Theorem~4.4 of~\cite{pst-kernel}]\label{thm:sparse:find}
Let $\tau > 0, c \geq 3$ be fixed constants. 
Given a brick $G$, in $\Oh(|\partial G|^8 |G|)$ time one can either correctly conclude that no $c$-short $\tau$-nice tree $F$ exist,
or find a $c$-short $\tau$-nice brick covering of $G$.
\end{theorem}

A direct adaptation of the proof of Lemma~4.5 of~\cite{pst-kernel} gives the main recursive step of the algorithm of Lemma~\ref{lem:sparse}.

\begin{lemma}[essentially Lemma~4.5 of \cite{pst-kernel}]\label{lem:sparse:recurse}
Let $c,\tau > 0$ be constants. 
Let $G$ be a brick and let $\BB$ be a $c$-short $\tau$-nice brick covering of $G$.
Assume that the algorithm of Theorem~\ref{thm:sparse} was applied recursively to bricks in $\BB$, yielding a graph $\sparseG_B$ for every $B \in \BB$.
Furthermore, assume that for every $B \in \BB$ we have $|\sparseG_B| \leq C \cdot |\partial B|^\alpha$ for some constants $C > 0$ and $\alpha \geq 1$
such that $(1-\tau)^{\alpha-1} \leq 1/c$. Then $\sparseG := \bigcup_{B \in \BB} \sparseG_B$ satisfies the conditions of Theorem~\ref{thm:sparse} for $G$
and $|\sparseG| \leq C \cdot |\partial G|^\alpha$.
\end{lemma}
\begin{proof}
The condition that $\partial G \subseteq \sparseG$ and the size bound follows exactly the same as in the proof of Lemma~4.5 of~\cite{pst-kernel}.
For the last condition of Theorem~\ref{thm:sparse}, the proof is essentially the same as in~\cite{pst-kernel}, but we repeat it for completeness. 
Consider a set $A \subseteq V(\partial G)$ and a subgraph $H$ that connects it in $G$ with minimum possible cost.
Without loss of generality, assume that, among subgraphs of $G$ connecting $A$ of minimum possible cost, $H$ contains minimum possible number of edges that are not in $\sparseG$.
We claim that there are no such edges; by contrary, let $e$ be such an edge. 
Since $\BB$ is a brick covering, let $B \in \BB$ be a brick containing $e$. Let $H_B$ be the connected component of $H \cap \mathrm{int} B$ that contains $e$, where
$\mathrm{int} B$ is the subgraph $B \setminus \partial B$. Clearly, $H_B$ is a connector in $B$, and let $A_B$ be its anchors.
By the properties of $\sparseG_B$, there exists a subgraph $H_B' \subseteq \sparseG_B$ that connects $A_B$ and is of cost not larger than the cost of $H_B$.
Consequently, $H' := (H \setminus H_B) \cup H_B'$ connects $A$, has cost not larger than the cost of $H$, and has strictly less edges outside $\sparseG$ than $H$.
This is the desired contradiction.
\end{proof}

We now move to carves and mountains (Sections~5 and~6 of~\cite{pst-kernel}). 
For a constant $\delta \in (0, 1/2)$, a \emph{$\delta$-carve} in a brick $G$ is a pair $(P,I)$ such that $P$ is a path in $G$ of length at most $(1/2 - \delta)|\partial G|$
with both endpoints on $\partial G$, and $I$ is a path on $\partial G$ between the endpoints of $P$ with length at most $1/2 \cdot |\partial G|$.
The \emph{interior} of a $\delta$-carve $(P,I)$ is the subgraph of $G$ enclosed by the closed walk $P \cup I$. 
The main result of Section~5 of~\cite{pst-kernel} is the following.
\begin{theorem}[Theorem~5.7 of \cite{pst-kernel}]\label{thm:sparse:core}
For any $\tau \in (0, 1/4)$ and $\delta \in [2\tau, 1/2)$, if a brick $G$ has no $3$-short $\tau$-nice tree, then there exists a finite face of $B$ that is never in the interior
of a $\delta$-carve in $G$. Furthermore, such a face can be found in $\Oh(|G|)$ time.
\end{theorem}
Section~6 of~\cite{pst-kernel} treats \emph{mountains}, a special case of $\delta$-carves. 
The definitions of~\cite{pst-kernel} are general to accommodate the edge-weighted setting as well; here we are content with only the unweighted setting.
Let $G$ be a brick and let $\delta \in (0, 1/2)$.
A $\delta$-carve $M = (P,I)$ with endpoints $l$ and $r$ of $P$ (so that $I$ is the counter-clockwise traverse along $\partial G$ from $l$ to $r$)
is a \emph{$\delta$-mountain} with \emph{summit} $v_M \in V(P)$ if for $P_L := P[l, v_M]$ and $P_R = P[v_M, r]$ we have that $P_L$ is a shortest $l-P_R$ path
in the subgraph of $G$ enclosed by $M$ and $P_R$ is a shortest $r-P_L$ path in the subgraph enclosed by $M$.
The main structural result of Section~6 of~\cite{pst-kernel} is the following.
\begin{theorem}[Theorem~6.3 of \cite{pst-kernel}]\label{thm:sparse:mountains}
Let $\tau \in (0, 1/4)$ and $\delta \in [2\tau, 1/2)$.
Let $G$ be a brick that does not admit a $3$-short $\tau$-nice tree and let $l,r \in V(\partial G)$ 
such that the counter-clockwise walk $I$ along $\partial G$ from $l$ to $r$ has length strictly less than $|\partial G|/2$. 
Then, there exists a closed walk $W_{l,r}$ in $G$ of length at most $3|I|$ that contains $I$ and such that, for each finite face $f$ of $G$, $f$ is enclosed by $W_{l,r}$
if and only if $f$ is enclosed by some $\delta$-mountain with endpoints $l$ and $r$.
Furthermore, the set of faces enclosed by $W_{l,r}$ can be found in $\Oh(|G|)$ time.
\end{theorem}

\subsection{Finding the middle cycle}\label{ss:sparse:C}

Armed with the toolbox from the previous section, we now re-engineer the argument of Section~7 of~\cite{pst-kernel} to our setting.
This is the place where the arguments of this work and~\cite{pst-kernel} mostly diverge, as 
here we build the interface between the cost and size (number of edges) measures.

Let $G$ be a brick.
Consider a set $A \subseteq V(\partial G)$ and a subgraph $H$ of $G$ connecting $A$
of minimum possible cost. Without loss of generality, we can assume that $H$ is a tree and,
furthermore, all its leaves are in $A$. 
Furthermore, assume that $H$ contains a vertex $x \in V(H) \cap V(\partial G)$ that is not a leaf.
Then $H$ can be partitioned into two trees $H_1$ and $H_2$, $E(H) = E(H_1) \uplus E(H_2)$, $V(H_1) \cap V(H_2) = \{x\}$.
Observe that if $H$ is a tree that is a subgraph of $G$ connecting $A$ of minimum possible cost, then 
$H_i$ is a subgraph of $G$ connecting $(A \cap V(H_i)) \cup \{x\}$ of minimum possible cost for every $i=1,2$.
Furthermore, if for $i=1,2$, $H_i'$ is a subgraph of $G$ connecting $(A \cap V(H_i)) \cup \{x\}$ of minimum possible cost,
then $H' := H_1' \cup H_2'$ is a subgraph of $G$ connecting $A \cup \{x\}$ of minimum possible cost. 
Since $x \in V(H)$, we infer that $\cost(H) = \cost(H')$ and $H'$ is also a minimum-cost subgraph of $G$ connecting $A$.

The paragraph above motivates the following definitions. A tree $H \subseteq G$ such that $V(H) \cap V(\partial G)$
is exactly the set of leaves of $H$ is called \emph{boundary-anchored}. 
A boundary-anchored tree $H$ that is a minimum-cost subgraph of $G$ connecting $A:= V(H) \cap V(\partial G)$ is called a \emph{minimum-cost tree} (connecting $A$).

From the discussion above, the following is immediate.
\begin{lemma}\label{lem:sparse:trees}
Let $G$ be a bipartite brick with bipartition classes $\green(G)$ and $\black(G)$.
Let $A \subseteq V(\partial G)$ and let $H$ be a minimum-cost subgraph of $G$ that connects $A$
that, among all minimum-cost subgraphs of $G$ connecting $A$, has minimum possible number of edges.
Then, $H$ is a tree with all its leaves in $A$, and $H$ can be decomposed
as $H = H_0 \uplus H_1 \uplus \ldots H_\ell$, where $H_0 \subseteq \partial G$ and 
for every $1 \leq i \leq \ell$, $H_i$ is a minimum-cost tree connecting $A_i := V(H_i) \cap V(\partial G)$.
Furthermore, for every $1 \leq i \leq \ell$, $H_i$ has minimum possible number of edges among all minimum-cost subgraphs of $G$
connecting $A_i$.
\end{lemma}

We also have the following bound.
\begin{lemma}\label{lem:sparse:opt}
Let $G$ be a bipartite brick with bipartition classes $\green(G)$ and $\black(G)$.
Let $A \subseteq V(\partial G)$ and let $H$ be a minimum-cost tree
connecting $A$. Then the cost of $H$ is at most $|\partial G|/2$ and
$|H| \leq \frac{3}{2} |\partial G|$.
In particular, $H$ is a $4$-short brickable connector.
\end{lemma}
\begin{proof}
For the first claim, it suffices to observe that $\partial G$ is a subgraph connecting
any $A \subseteq V(\partial G)$ of cost $|\partial G|/2$.
For the second claim, root $H$ in an arbitrary vertex, and compute $|H|$ as follows.
First, there are at most $|A \cap \green(G)| \leq |\partial G|/2$ edges of $H$
that are incident to a leaf of $H$ belonging to $\green(G)$. For every edge $e$ that is not as above, if $e$ connects
$v$ with its parent $u$, then charge $e$ to $v$ if $v \in \black(G)$ and otherwise
charge $e$ to any of the children of $v$ (which exist and are in $\black(G)$).
In this manner, every edge of $H$ that is not incident with a leaf of $H$ in $\green(G)$ is charged to some vertex in $\black(G) \cap V(H)$
and every vertex $w \in \black(G) \cap V(H)$ is charged at most twice: once by
the edge from $w$ to its parent, and one possibly from the parent of $w$ to the grandparent
of $w$. Since the number of vertices of $\black(G) \cap V(H)$ is the cost of $H$, which
is at most $|\partial G|/2$, the bound on $|H|$ follows.
Finally, following the observations preceding Theorem~\ref{thm:sparse:find}, $H$ is a brickable connector that is $4$-short.
\end{proof}

We are now ready to find the cycle $C'$ whose existence was promised in Section~\ref{sec:sparse}.
\begin{theorem}[analog of Theorem~7.1 of \cite{pst-kernel}]\label{thm:sparse:taming}
Let $\tau \in (0, \frac{1}{44}]$ be a fixed constant.
Let $G$ be a brick that does not admit a $4$-short $\tau$-nice tree,
is bipartite with bipartition classes $\green(G)$ and $\black(G)$, and 
has perimeter at least $4/\tau$.
Then one can in $\Oh(|G|)$ time compute a simple cycle $C$ in $G$ with the following properties:
\begin{enumerate}
\item the length of $C$ is at most $\frac{64}{\tau^2} |\partial G|$;
\item for each vertex $v \in V(C)$, the distance from $v$ to $V(\partial G)$ in $G$ is at most $(1/4 - 2\tau)|\partial G|$ and, furthermore,
  there exists a shortest path from $v$ to $V(\partial G)$ that does not contain any edge strictly enclosed by $C$;
\item $C$ encloses $\coref$, where $\coref$ is any arbitrarily chosen face of $G$ promised by Theorem~\ref{thm:sparse:core} that is not enclosed by any $2\tau$-carve;
\item for any $A \subseteq V(\partial G)$, there exists a minimum-cost subgraph $H$ of $G$ connecting $A$ such that no vertex of degree at least $3$
is strictly enclosed by $C$.
\end{enumerate}
\end{theorem}
\begin{proof}
For two vertices $x,y \in V(\partial G)$, by $\partial G[x, y]$ we denote the path
that is the counter-clockwise traverse of $\partial G$ from $x$ to $y$.

We start by computing a set of \emph{pegs} $\pegs \subseteq V(\partial G) \cap \green(G)$ in the 
following greedy manner. We start with arbitrary $v_0 \in V(\partial G) \cap \green(G)$
and traverse $\partial G$ starting from $v_0$ twice, once clockwise and once counter-clockwise.
In each pass, we take as a next peg the first vertex in $\green(G)$
that is at distance (along $\partial G$) larger than $\tau |\partial G|/4$ from the previously
placed peg. With two passes, we have $|\pegs| \leq 8/\tau$, which is a constant. Furthermore,
we have that for every $v \in V(\partial G)$ there
exist pegs $\leftpeg(v),\rightpeg(v) \in \pegs$ with $v \in \partial G[\leftpeg(v),\rightpeg(v)]$ and 
$$|\partial G[\leftpeg(v), v]|, |\partial G[v,\rightpeg(v)]| \leq \tau |\partial G|/4 + 1 \leq \tau |\partial G|/2.$$
In the above, the first inequality stems from the way we place pegs (including the requirement than pegs are in $\green(G)$)
and the second inequality is implied by $|\partial G| \geq 4/\tau$.
Note that $\leftpeg(v)$ is the first peg in clockwise direction from $v$ and $\rightpeg(v)$ is the first peg
in the counter-clockwise direction from $v$.

We set $\delta = 4\tau$.
For any $l,r \in \pegs$, $l \neq r$, $|\partial G[l,r]| < |\partial G|/2$, we apply Theorem~\ref{thm:sparse:mountains} 
to $l$, $r$, and $\delta$, obtaining a set of faces $MR_{l,r}$. We define $MR = \bigcup_{l,r} MR_{l,r}$. 
Clearly, $\coref \notin MR$. We define $\widehat{\coref}$ to be the connected component of $G^\ast-MR$ that contains $\coref$, where $G^\ast$ is the dual of $G$ without the infinite face. 
Furthermore, let $C_1$ be the closed walk in $G$ around $\widehat{\coref}$.

Consider a face $f \in MR$. By definition, there exists $l,r$ with $f \in MR_{l,r}$ and by Theorem~\ref{thm:sparse:mountains} there exists a $\delta$-mountain $M = (P, I)$ enclosing $f$.
Since $P$ is a simple path, there exists a path $Q_f$ in $G^\ast$ from $f$ to a face incident to an edge of $\partial G$, with all faces on $Q_f$
enclosed by $M$. This implies that $C_1$ is a simple cycle.

Furthermore, since every edge on $C_1$ is an edge of $\partial G$ or some edge of
the walk $W_{l,r}$ obtained from Theorem~\ref{thm:sparse:mountains}, we have that
\begin{align*}
|C_1| &\leq \left| \partial G \cup \bigcup \left \{ W_{l,r} \setminus \partial G[l,r] ~|~l,r \in \pegs \wedge l \neq r \wedge |\partial G[l,r]| < |\partial G|/2\right\}   \right| \\
    & \leq |\partial G| + |\pegs| \cdot (|\pegs|-1) \cdot 2 \cdot |\partial G|/2 \\
    & \leq \frac{64}{\tau^2} |\partial G|.
\end{align*}
The above estimation is also the sole need for introducing the pegs. If one picks $MR$ to be the union of $MR_{l,r}$ for every $l,r \in V(\partial G)$ with $|\partial G[l,r]| < |\partial G|/2$
then the above estimate will be cubic, not linear in $|\partial G|$. The linear dependency on $|\partial G|$ is essential for the final polynomial bound on the size of the kernel.

So far, all computations can be done in $\Oh(|G|)$ time due to $\pegs$ being of constant
size and the time bounds of Theorems~\ref{thm:sparse:core} and~\ref{thm:sparse:mountains}. 
We now compute the subgraph $\Gclose$ induced by all vertices of $G$
that are within distance at most $(1/4-2\tau)|\partial G|$ from $V(\partial G)$.
This subgraph can be easily computed by breadth-first search in linear time.

\newcommand{\GC}{{G^C}}
Since for the sake of defining $MR_{l,r}$, we used $\delta$-mountains for $\delta=4\tau$,
all vertices of any such $\delta$-mountain are contained in $\Gclose$ and, consequently,
$C_1$ is contained in $\Gclose$.
Let $\GC$ be the subgraph of $\Gclose$ enclosed by $C_1$; note that $\GC$ is a brick.
Furthermore, let $\coref^\GC$ be the face of $\GC$ that contains $\coref$.
We define $C$ to be some shortest cycle in $\GC$ separating $\coref^\GC$ from the infinite face
of $\GC$. Since $C$ corresponds to a minimum cut in a dual of $\GC$, $C$ can be computed
in linear time~\cite{cut-unit-linear}.
We claim that $C$ satisfies the desired conditions.

Clearly, the length of $C$ is at most the length of $C_1$ (as $C_1$ is a good candidate for $C$),
and hence satisfies the desired length bound.
Also, $C$ encloses $\coref$ by definition.
For the second property, the fact that $C$ is a cycle in $\GC$ ensures
that every $v \in V(C)$ is within distance $(1/4-2\tau)$ from $V(\partial G)$. 
Let $P_v$ be a shortest path from $v$ to $V(\partial G)$ that minimizes the number of edges
strictly enclosed by $C$. We claim that there are no such edges; by contradiction, assume
that there exists a subpath $Q$ of $P_v$ with endpoints $x,y \in V(C)$ and all edges
and internal vertices strictly enclosed by $C$. If $|Q| \geq C[x,y]$ or $|Q| \geq C[y,x]$,
    then we reach a contradiction with the choice of $P_v$
    by replacing $Q$ with $C[x,y]$ or $C[y,x]$ respectively on $P_v$.
Otherwise, if $|Q| < |C[x,y]|, |C[y,x]|$, then, as $P_v \subseteq \Gclose$ and thus $Q \subseteq \GC$,
either $Q \cup C[y,x]$ or $Q \cup C[x,y]$ is a strictly better candidate for $C$, a contradiction.

We are left with the last desired property.
Consider a set $A \subseteq V(\partial G)$ and a subgraph $H$ of $G$ connecting $A$
of minimum possible cost. 
Furthermore, we choose $H$ that satisfies the following minimality property:
$H$ has minimum number of edges among all minimum-cost subgraphs connecting $A$ and, subject to that,
has minimum number of edges strictly enclosed by $C$. We claim that this choice of $H$ satisfies the desired properties.

Lemma~\ref{lem:sparse:trees} implies that it suffices to consider the case when $H$ is actually a minimum-cost tree. 

Since every tree in $G$ with all leaves on $\partial G$
is a brickable connector, Lemma~\ref{lem:sparse:opt} implies that every minimum-cost tree in $G$ is $4$-short. 
Since $G$ does not admit a $4$-short $\tau$-nice tree, every minimum-cost tree
in $G$ is \emph{not} $\tau$-nice.

Consequently, we infer that $H$ is not $\tau$-nice, that is,
there exists a brick $B \in \BB_H$ with $|\partial B| > (1-\tau)|\partial G|$.
Since $H$ is a minimum-cost tree,
there exist $a,b \in V(\partial G)$ such that $\partial B \setminus H = \partial G[a, b]$.

Observe that $\partial G[b, a]$ connects $A$; hence
\begin{align*}
&\left\lceil \frac{|\partial G[b, a]|}{2} \right\rceil \geq \cost(\partial G[b, a]) \geq \cost(H) =
\cost(H[a,b]) + \cost(V(H) \setminus V(H[a,b])) \\
&\quad \geq \frac{|H[a,b]|}{2} + \cost(V(H) \setminus V(H[a,b])) \\
&\quad \geq \cost(V(H) \setminus V(H[a,b])) + \frac{|\partial B|}{2} - \frac{|\partial G[a,b]|}{2}.
\end{align*}
Consequently, as $|\partial B| > (1-\tau)|\partial G|$, we have that
\begin{equation}\label{eq:sparse:legs}
\cost(V(H) \setminus V(H[a,b])) < \frac{1}{2} + \frac{\tau}{2}|\partial G|.
\end{equation}
Let $Z$ be the union of $\{a,b\}$ and the
set of vertices of $V(H[a,b])$ that are of degree at least $3$ in $H$.
Let $v_a,v_b$ be the vertices on $H[a,b]$ with
$$|H[a, v_a]|,|H[b,v_b]| \leq \min\left(|H[a,b]|/2, (1/2 - 7\tau)|\partial G|\right)$$
such that $H[a,v_a]$ and $H[b,v_b]$ are as long as possible.
Note that it may happen that $v_a=v_b$, but $a, v_a, v_b, b$ lie on $H[a,b]$ in this order
and $v_a \neq b$, $v_b \neq a$.

Let $w_a$ be the vertex of $Z$ on $H[a,v_a]$ that is closest to $v_a$
and similarly define $w_b$ on $H[b,v_b]$.
Let $H_a$ be the subtree of $H$ being the connected component of $H-e_a$ containing $w_a$,
rooted at $w_a$, where $e_a$ is the edge of $H[a,b]$ incident with $w_a$ but not lying
on $H[a,w_a]$. Traverse $\partial G$ in clockwise direction from $a$, and let $c$
be the last vertex of $H_a$ encountered (before returning back to $a$); note that 
it may happen that $c=a$ if $w_a=a$.

Note that from~\eqref{eq:sparse:legs} we infer that $|H[c, w_a]| < 2 + \tau|\partial G|$
and, consequently, 
  \begin{equation}\label{eq:sparse:Hac}
  |H[a, c]| \leq 2 + (1/2-6\tau)|\partial G| \leq (1/2-5\tau)|\partial G|,
  \end{equation}
where the last inequality follows from the fact that $|\partial G| \geq 4/\tau$.

\begin{claim}\label{cl:sparse:ca}
Either $a = w_a = c$ or $c \neq a$ and $|\partial G[c, a]| < |\partial G|/2$.
\end{claim}
\begin{claimproof}
Assume otherwise.
Following (\ref{eq:sparse:Hac}), $(H[c, a], \partial G[a,c])$ is a $\delta$-carve.
We have $|\partial G[c, a]|+ |H[c,a]| > (1-\tau)|\partial G|$, as otherwise $H[c,a]$ would be a $3$-short $\tau$-nice tree in $G$.
By~\eqref{eq:sparse:Hac}, this implies $|\partial G[c,a]| > (1/2 + 4\tau)|\partial G|$, so 
$|\partial G[a, c]| < (1/2-4\tau)|\partial G|$.
Let $D = H[c, a] \cup \partial G[a,c]$ be a closed walk and let $H'$ be constructed
from $H$ by first deleting all edges enclosed by $D$, and then adding $D \setminus \partial G[a, b]$ instead.
Clearly, $H'$ connects $A$ as well.
Let $K$ be the set of vertices of $\black(G) \cap V(H)$ that are not enclosed by $D$.
Since $B$ is enclosed by $D$, and $|\partial B| > (1-\tau)|\partial G|$, we have that
\begin{align*}
\cost(H') &\leq |K| + \left\lceil \frac{|D|-|\partial G[a, b]|}{2}\right\rceil \leq |K| + \frac{1}{2} + \frac{1-9\tau}{2}|\partial G| - \frac{|\partial G[a,b]|}{2} \\
   & < |K| + \frac{|\partial B|-|\partial G[a,b]|}{2} = |K| + \frac{|H[a,b]|}{2} \leq \cost(H).
   \end{align*}
Here, we again used the fact that $|\partial G| \geq 4/\tau$.
The above inequality contradicts the choice of $H$.
\end{claimproof}

In the second case of Claim~\ref{cl:sparse:ca} (i.e., $c \neq a$ and $|\partial G[c,a] < |\partial G|/2$),
   we have that $(H[c, a], \partial G[c,a])$ is a $\delta$-carve using (\ref{eq:sparse:Hac}).
Similarly as before, this implies that $|\partial G[c,a]| \leq |H[c,a]| + \tau|\partial G|
\leq (1/2-4\tau)|\partial G|$, as otherwise $H[c,a]$ is a $3$-short $\tau$-nice tree.
We now use the pegs $\leftpeg(c)$ and $\rightpeg(a)$.
Let $P = \partial G[\leftpeg(c), c] \cup H[c,a] \cup \partial G[a, \rightpeg(a)]$
and $I = \partial G[\leftpeg(c), \rightpeg(a)]$.
Note that the placement of pegs ensure that 
$|\partial G[\leftpeg(c), c]| + |\partial G[a, \rightpeg(a)]| \leq \tau |\partial G|$.
Consequently, $|I| \leq (1/2 - 3\tau)|\partial G|$, $|P| \leq (1/2-4\tau)|\partial G|$
and $(P, I)$ is a $\delta$-carve. 
Note that the path $H[a,c]$ contains no vertex of $\partial G[\leftpeg(c), c] \cup \partial G[a, \rightpeg(a)]$ (as $H$ is a minimum-cost tree)
and thus $P$ is a simple path.

\begin{figure}[tb]
\centering
\includegraphics[width=.4\linewidth]{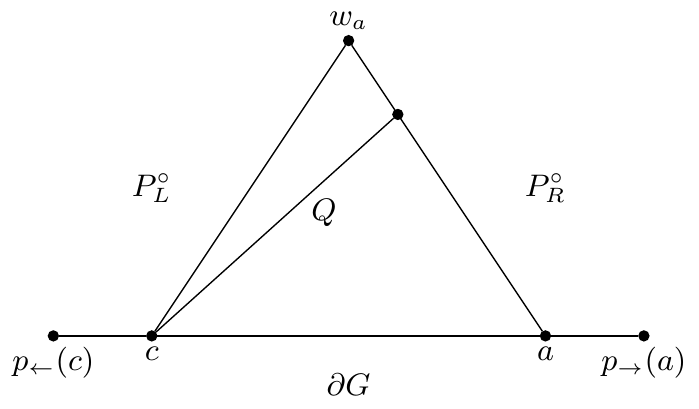}
\caption{Situation in the proof of Claim~\ref{cl:sparse:local-mnt}.}
\label{fig:local-mnt}
\end{figure}

\begin{claim}\label{cl:sparse:local-mnt}
$(P,I)$ is a $\delta$-mountain with $w_a$ as a summit.
\end{claim}

\begin{claimproof}
Assume the contrary; see Figure~\ref{fig:local-mnt} for an illustration.
Let $P_L^\circ = P[\leftpeg(c), w_a]$ and $P_R^\circ = P[\rightpeg(a), w_a]$.
Assume that there exists a path $Q^\circ$ from $\leftpeg(c)$
to $P_R^\circ$, enclosed by $P \cup I$, that is strictly shorter than $P_L^\circ$;
the proof for the symmetric case of a short path from $\rightpeg(a)$ to $P_L^\circ$ is
analogous and thus omitted. 
Since $P$ and $I$ overlap on $\partial G[\leftpeg(c), c] \cup \partial G[a, \rightpeg(a)]$,
this implies that there exists a subpath $Q$ of $Q^\circ$ from $c$
to $P_R := H[a, w_a]$ that is strictly shorter than $P_L := H[c, w_a]$
and is enclosed by $H[a, c] \cup \partial G[c, a]$. 
Let $x \in V(H[a, w_a])$ be the other endpoint of $Q$. Construct a graph $H'$ from $H$ as follows:
replace in $H$ all edges enclosed by $Q \cup H[x, c]$ by $H[w_a, x] \cup Q$. By the choice of $c$,
$H'$ connects $A$. As for the cost, note that since $|Q| < |H[c, w_a]|$ while $Q$ and $H[c,w_a]$ share $c$ as an endpoint in the bipartite graph $G$, we have
$$|\black(G) \cap (V(Q) \setminus \{x\})| \leq |\black(G) \cap (V(H[c, w_a]) \setminus \{w_a\})|.$$
Consequently, since $H' \subseteq (H \setminus H[c, w_a]) \cup Q$, we have that
$\cost(H') \leq \cost (H)$. 
Furthermore, as $|Q| < |H[c,w_a]|$ but $H' \subseteq (H \setminus H[c, w_a]) \cup Q$,
we have $|H'| < |H|$, contradicting the choice of $H$. This finishes the proof of the lemma.
\end{claimproof}

We infer from Theorem~\ref{thm:sparse:core} and Claim~\ref{cl:sparse:local-mnt} that $\coref$ is not enclosed
by $H[a, c] \cup \partial G[c, a]$ and no edge of $H_a$ is strictly enclosed by $C_1$
(and thus not by $C$ either).
A symmetric argument holds for $b$ and $w_b$ and a subtree $H_b$ defined symmetrically.

If $H[v_a,v_b]$ does not contain any internal vertex from $Z$, we are done, as all vertices
of degree at least three in $H$ are then contained in $H_a$ or $H_b$. 
Otherwise, the choice of $w_a$ and $w_b$ implies that $v_a \neq v_b$, $|H[v_a,v_b]| > 1$.
Consequently, the choice of $v_a$ and $v_b$ implies that
$|H[a, b]| > (1-14\tau)|\partial G|$.
Since $\cost(H) \leq \cost(\partial G[b, a])$, we have that $|H[a, b]| \leq |\partial G[b, a]|$ (note that these two paths share endpoints). Hence, $|\partial G[a, b]| < 14\tau|\partial G|$.

Consider two consecutive vertices $w_1$, $w_2$ from $Z$ on $H[a, b]$. 
Note that $H' := (H \setminus H[w_1,w_2]) \cup \partial G[a, b]$ connects $A$, and
$$\cost(H') \leq \cost(H) + 1 - \frac{|H[w_1,w_2]|}{2} + \frac{|\partial G[a, b]|}{2}.$$
By the choice of $H$ and the assumption $|\partial G| \geq 4/\tau$, 
we have that 
\begin{equation}\label{eq:sparse:w1w2}
|H[w_1,w_2]| \leq 15\tau|\partial G|.
\end{equation}
We claim the following.
\begin{claim}\label{cl:sparse:Hclose}
For every $v \in V(H)$ there exists a path from $v$ to $\partial G$ contained in $H$ that is of length
at most $(1/4-2\tau)|\partial G|$. In particular, $H \subseteq \Gclose$.
\end{claim}
\begin{claimproof}
\eqref{eq:sparse:legs} implies the claim for every $v \in V(H) \setminus V(H[a,b])$ 
as a path in $H$ from $v$ to the closest vertex of $H$ on $V(\partial G) \setminus V(H[a,b])$ is of length at most
$\tau|\partial G| + 2 \leq (1/2-2\tau)|\partial G|$ as $|\partial G| > 4/\tau$ and $\tau \leq \frac{1}{44}$.

Let then $v \in V(H[a,b])$ and consider a path that traverses $H[a,b]$ from $v$ to the closest vertex $z \in Z$
and then, if $z \notin \{a,b\}$, follow the path from $z$ to $\partial G$ in $E(H) \setminus E(H[a,b])$. 
By~\eqref{eq:sparse:w1w2} and~\eqref{eq:sparse:legs} this path is of length at most
$2 + \tau|\partial G| + \frac{15}{2} \tau |\partial G| \leq 9\tau|\partial G| \leq
(1/4 - 2\tau)|\partial G|$ as $\tau \leq \frac{1}{44}$. 
This finishes the proof of the claim.
\end{claimproof}

Using the fact that~$H \subseteq \Gclose$ we complete the proof. If $H$ does not contain any edge strictly enclosed by $C$, then we are done, so assume otherwise. Consider the subgraph of~$H$ consisting of edges strictly enclosed by~$C$, and let~$H_C$ be a connected component of this subgraph. Then~$H_C$ is a tree with leaves on~$V(C)$. The tree~$H_C$ partitions the disk bounded by~$C$ into faces, of which exactly one face~$\coref^C$ contains face~$\coref$ (since~$C$ encloses~$\coref$). Let~$x,y \in V(C) \cap V(H_C)$ be the unique two leaves of~$H_C$ that lie on~$\coref^C$, and choose them in such an order that~$H[x,y] \cup C[x,y]$ does not enclose~$\coref$, that is,~$C' := H[x,y] \cup C[y,x]$ encloses $\coref$. Since $H \subseteq \Gclose$, the cycle~$C' \subseteq \Gclose$ was a candidate for the separating cycle~$C$. By the minimality of $C$, we have~$|C'| \geq |C|$ and therefore~$|H[x,y]| \geq |C[x,y]|$.
Consider a graph $H'$ constructed from $H$ by replacing every edge of $H$ enclosed by $H[x,y] \cup C[x,y]$ with $C[x,y]$. 
Clearly, $H'$ connects $A$. Since $|H[x,y]| \geq |C[x,y]|$ and these paths have the same endpoints, we have that $\cost(H') \leq \cost(H)$
and $|H'| \leq |H|$. Furthermore, $H'$ contains strictly fewer edges enclosed by $C$ than $H$, which is a contradiction with the choice of $H$.

Consequently, no edge of $H$ is strictly enclosed by $C$, and we are done. This finishes the proof of Theorem~\ref{thm:sparse:taming}.
\end{proof}

We conclude this section with a remark about the possibility of minimum-cost trees using edges strictly enclosed by~$C$. Although the above proof replaced a graph $H$ with vertices of degree at least $3$ strictly enclosed by $C$
with a subgraph $H'$ that does not contain any edge strictly enclosed by $C$ (not even using degree-2 vertices strictly enclosed by~$C$), this does not imply that any subset~$A \subseteq V(\partial G)$ can be connected by a minimum-cost tree $H$ that uses no edges strictly enclosed by $C$. 
We used the assumption of a vertex of degree at least $3$ on $H$ to show that $H \subseteq \Gclose$. This may not be the case in the following example, leading to situations where all minimum-cost connecting subgraphs use edges strictly enclosed by~$C$.

Let $A = A_1 \uplus A_2$ with $a_1 \in A_1$, $a_2 \in A_2$ such that $|\partial G[a_1,a_2]| = |\partial G|/2$ and for every $i=1,2$ and $a \in A_i$ we have that
$\mathrm{min}(|\partial G[a,a_i]|, |\partial G[a_i,a]|) < 0.001|\partial G|$. That is, $A$ consists of two antipodal ``clouds'' of vertices.
Then, there may exist a way of connecting $A$ that is slightly cheaper than taking an appropriate subpath of $\partial G$, namely connecting each $A_i$
with a small tree $H_i$ and then connecting roots of $H_i$ via a path of length almost $|\partial G| / 2$ through the middle of $G$ and~$C$. Effectively, the existence of a shorter path through the middle of~$G$ shows that a cycle even shorter than~$C$ exists in~$G$ that separates the infinite face of~$G$ from~$\coref$. But this shorter cycle contains vertices that do not belong to~$\Gclose$ and therefore we cannot use this cycle in our recursive step, as the lengths of paths to and from the cycle would grow prohibitively large.

\subsection{Wrapping up the argument}\label{ss:sparse:wrapup}

With Theorem~\ref{thm:sparse:taming} and Lemma~\ref{lem:sparse:recurse} in hand,
we now wrap-up the proof of Lemma~\ref{lem:sparse}
essentially in the same way as it is done in Section~8 of~\cite{pst-kernel}.

We fix $\tau = 1/44$ and choose $\alpha$ such that
$$(1-\tau)^{\alpha-1} < \frac{1}{4},$$
and
$$(1-3\tau)^{\alpha-1} < 2\,857\,536^{-1}.$$
We show a recursive algorithm for Lemma~\ref{lem:sparse} that runs in time
$\Oh(|\partial G|^\alpha |G|)$ and returns a subgraph of size
at most $\beta \cdot |\partial G|^\alpha$ for a constant $\beta > 0$ satisfying
$$2\,857\,536(1-3\tau)^{\alpha-1} + 123\,904^2/\beta < 1,$$
and
$$\beta (1-3\tau)^{\alpha-1} \geq 1.$$
In particular, $\alpha = \sparseexp$ and $\beta = 10^{13}$ suffices.

First, consider the base case $|\partial G| \leq 4/\tau = 176$. For each $A \subseteq V(\partial G)$
we compute a minimum-cost subgraph $H_A$ connecting $A$ in $\Oh(|G|)$ time using
a straightforward adaptation of the algorithm of Erickson et al.~\cite{erickson} and insert $H_A$
into $\sparseG$.\footnote{The algorithm of Erickson et al.~\cite{erickson} is the classic Dreyfus-Wagner dynamic programming algorithm for \textsc{Steiner tree} that, for every
  subset $B \subseteq A$ and $r \in V(G)$ computes a minimum-cost tree connecting $B$ and $r$. The essence of the result of~\cite{erickson} lies in an observation 
  that if $A$ lies on a single face of a planar graph, it suffices only to consider sets $B$ that correspond to a set of consecutive vertices of $A$ on the said face.
  This limits the number of states of the dynamic programming algorithm from $n\cdot 2^{|A|}$ to roughly $|A|^2 n$.
  The same observation holds in our setting and it is straightforward to adapt the Dreyfus-Wagner dynamic programming algorithm to our cost function.}
Furthermore, we insert $\partial G$ into $\sparseG$.
Lemma~\ref{lem:sparse:opt} implies that $|H_A| \leq 264$, and thus $|\sparseG| \leq 
176 + 2^{176} \cdot 264$ which is at most $\beta |\partial G|^\alpha$ for
any $\alpha \geq 176$, $|\partial G| \geq 2$ and $\beta \geq 265$.

In the recursive case, if $|\partial G| > 4/\tau$, we apply
Theorem~\ref{thm:sparse:find} to look for a $4$-short $\tau$-nice brick covering.
If a brick covering $\BB$ is found, then we recurse on each brick $B \in \BB$ separately,
obtaining a graph $\sparseG_B$, and return $\sparseG := \bigcup_{B \in \BB} \sparseG_B$.
Lemma~\ref{lem:sparse:recurse} ensures that the size bound on $\sparseG$ holds.
For the time bound, assume that the application of Theorem~\ref{thm:sparse:find} runs
in time bounded by $c_1 \cdot |\partial G|^8 \cdot |G|$.
Then, the total time spent is bounded by
$$\left(c_1 |\partial G|^8 + \sum_{B \in \BB}c |\partial B|^\alpha\right) |G| \leq |\partial G|^\alpha \cdot |G| \cdot \left(c_1 + 4c(1-\tau)^{\alpha-1}\right).$$
This is at most $c \cdot |\partial G|^\alpha \cdot |G|$ for sufficiently large $c$ by the choice
of $\alpha$.

\newcommand{\bulka}{\psi}

We are left with the case that the application of Theorem~\ref{thm:sparse:find} returned
that no $4$-short $\tau$-nice tree exists.
Then, we invoke Theorem~\ref{thm:sparse:core} to find a face $\coref$
and Theorem~\ref{thm:sparse:taming} to find the cycle $C$ enclosing $\coref$.
We have $|C| \leq \frac{64}{\tau^2} |\partial G| \leq 123\,904 |\partial G|$.
Mark a small set $X \subseteq V(C)$ such that the distance between two consecutive vertices of $X$ on $C$
is at most $2\tau|\partial G| = |\partial G|/22$.
As $|\partial G| > 176$, we may greedily mark such a set $X$ of size at most $\frac{|C|}{2\tau|\partial G|} \cdot \frac{9}{8} \leq 3\,066\,624$.
For every $x \in X$, compute a shortest path $P_x$ from $x$ to $V(\partial G)$ that 
does not use any edge strictly enclosed by $C$. This can be done by a single breadth-first search
from $V(\partial G)$ in $G$ with the edges strictly enclosed by $C$ removed. 
In this manner, we obtain also the property that the intersection of two paths $P_x$ and $P_y$ for any $x,y \in X$ is a (possibly empty) suffix.
For $x \in X$, let $\bulka(x)$ 
be the second endpoint of $P_x$.
By the properties of $C$, we have that $|P_x| \leq (1/4-2\tau)|\partial G| = \frac{9}{44}|\partial G|$.

Consider two consecutive (in the counter-clockwise order) vertices $x,y \in X$ on $C$
and consider a walk $P = P_x \cup C[x,y] \cup P_y$ from $\bulka(x)$ to $\bulka(y)$. The length of this
walk is at most $(2 \cdot \frac{9}{44} + \frac{2}{44})|\partial G| = \frac{5}{11}|\partial G|$.
We claim that
\begin{equation}\label{eq:sparse:ananas}
|\partial G[\bulka(x), \bulka(y)] \cup P| \leq (1-3\tau)|\partial G|.
\end{equation}
If $\bulka(x) = \bulka(y)$, then the claim is trivial. 
Otherwise, $P_x$ and $P_y$ do not intersect. 
Let $x'$ be the vertex on $P_x \cap V(C[x,y])$ that is closest to $\bulka(x)$ on $P_x$
and similarly define $y'$. As $P_x$ and $P_y$ do not intersect, $x, x', y', y$ lie
on $C[x,y]$ in this order; it may happen that $x=x'$ or $y=y'$ but $x' \neq y'$.
Then, $P' = P_x[\bulka(x), x'] \cup C[x',y'] \cup P_y[y', \bulka(y)]$ is a simple
path of length at most $\frac{5}{11} |\partial G| \leq (1/2 - 2\tau)|\partial G|$.
Hence, either $(P', \partial G[\bulka(x), \bulka(y)])$ or $(P', \partial G[\bulka(y), \bulka(x)])$
is a $(2\tau)$-carve. In the latter case, such a carve would enclose $\coref$, which 
is impossible. Consequently, $|\partial G[\bulka(x), \bulka(y)]| \leq |\partial G|/2$.
We then have $|\partial G[\bulka(x), \bulka(y)]| \leq (1/2 - \tau)|\partial G|$ as otherwise
$|\partial G[\bulka(y),\bulka(x)]| < (1/2 + \tau)|\partial G|$ which would make 
$P'$ a $4$-short $\tau$-nice brickable connector. 
Together with $|P'| \leq |P| \leq \frac{5}{11} |\partial G| \leq (1/2 - 2\tau)|\partial G|$
this finishes the proof of~\eqref{eq:sparse:ananas}.

Let $W_x = \partial G[\bulka(x), \bulka(y)] \cup P$ be a closed walk.
Let $H_x$
be the subgraph of $W_x$ consisting of all edges of $W_x$ that are incident with the infinite
face of $W_x$ treated as a plane graph.\footnote{It can be shown that in fact $E(H_x) = E(W_x)$ here, but it is not necessary for the argument.}
Each doubly-connected component of $H_x$ is a simple cycle (of length at least $3$, as $G$ is simple) or a bridge. For each simple cycle, we construct a brick $B$ consisting of all edges
enclosed by this cycle, and denote by $\BB_x$ the family of all bricks obtained in this manner.
Let $D_x$ be the set of bridges. 
We have
$$2|D_x| + \sum_{B \in \BB_x} |\partial B| \leq |W_x| \leq (1-3\tau)|\partial G|.$$
Hence,
$$\sum_{x \in X} \left(|D_x| + \sum_{B \in \BB_x} |\partial B| \right)\leq |X| \frac{41}{44} |\partial G| \leq 2\,857\,536 |\partial G|.$$
We recurse on each $B \in \bigcup_{x \in X}\BB_x$, obtaining a graph $\sparseG_B$.
Furthermore, for each $x, y \in V(C)$, we mark one shortest path $Q_{x,y}$ between $x$ and $y$
in $G$, if its length is at most $|\partial G|$. We define
$$\sparseG = \left(\bigcup_{x,y \in V(C)} Q_{x,y} \right) \cup \left(\bigcup_{x \in X} D_x \cup \bigcup_{B \in \BB_x} \sparseG_B\right).$$
Note that every finite face of $G$ that is not enclosed by $C$ is enclosed by exactly one walk $W_x$.
Hence, Theorem~\ref{thm:sparse:taming} ensures that for every $A \subseteq V(\partial G)$, there
exists a minimum-cost subgraph $H$ connecting $A$ that is contained in $\sparseG$.
In particular, every connection in $H$ strictly enclosed by $C$ may be realized
by one of the chosen shortest paths $Q_{x,y}$.

For the bound on the size of $\sparseG$, we have that
\begin{align*}
|\sparseG| &\leq \sum_{x \in X} \left(|D_x| + \sum_{B \in \BB_x} \beta \cdot |\partial B|^\alpha\right)
 + \binom{|C|}{2} |\partial G| \\
   &\leq \beta |\partial G|^{\alpha - 1} (1-3\tau)^{\alpha-1} \cdot \sum_{x \in X} \left(|D_x| + \sum_{B \in \BB_x} |\partial B|\right) + \binom{|C|}{2} |\partial G| \\
   &\leq \beta \cdot |\partial G|^\alpha \cdot 2\,857\,536 \cdot (1-3\tau)^{\alpha-1} + 123\,904^2 |\partial G|^3 \\
   &\leq \beta |\partial G|^\alpha.
 \end{align*}
Here, the last inequality follows from the choice of $\alpha$ and $\beta$
while the second inequality uses that $\beta (1-3\tau)^{\alpha-1} \geq 1$.

Regarding the time bound, note that all computations, except for the recursive calls, can be done
in time $c_2 \cdot |\partial G|^3 \cdot |G|$ for some constant $c_2$.
Therefore, the total time spent is bounded by
$$\left(c_2 |\partial G|^3 + c \sum_{x \in X} \sum_{B \in \BB_x} |\partial B|^\alpha\right)|G|
\leq |\partial G|^\alpha |G| \left(c_2 + 2\,857\,536c(1-3\tau)^{\alpha-1}\right).$$
This is smaller than $c |\partial G|^\alpha |G|$ for sufficiently large $c$.

\section{Odd Cycle Transversal}\label{sec:oct}
To understand the {\poct} problem, we rely on the correspondence between odd cycle transversals and $T$-joins. This correspondence was originally developed by Hadlock~\cite{Hadlock} for the edge version of {\poct} on planar graphs; for the vertex version discussed here, we build on the work of Fiorini et al.~\cite{FioriniHRV2008}. Given a graph $H$ and set $T \subseteq V(H)$, a \emph{$T$-join} in $H$ is a set $J \subseteq E(H)$ such that $T$ equals the set of odd-degree vertices in the subgraph of $H$ induced by $J$. It is known that a connected graph contains a $T$-join if and only if $|T|$ is even.

\begin{lemma}[{\cite[Lemma 1.1]{FioriniHRV2008}}] \label{lem:oct:radial}
Let $T$ be the set of odd faces of a connected plane graph $G$. Then $C \subseteq V(G)$ is an odd cycle transversal of $G$ if and only if $\radial(G)[C \cup F(G)]$ contains a $T$-join, that is, each connected component of $\radial(G)[C \cup F(G)]$ contains an even number of vertices of $T$.
\end{lemma}
This leads to the following problem\longversion{\footnote{Lokshtanov et al. \cite{LokshtanovSW12} call this the \textsc{$L$-join} problem. We prefer this naming, given the prevalence of the use of the letter $T$ and the Steiner-like nature of the problem.}}:
\defparproblem{\ptj}{A connected bipartite graph $G$, a fixed partition $V(G) = \green(G) \uplus \black(G)$, $T \subseteq \green(G)$, and an integer $k$.}{k}{Does there exist a set $C \subseteq \black(G)$ of size at most $k$ such that $G[C \cup \green(G)]$ contains a $T$-join, that is, each connected component of $G[C \cup \green(G)]$ contains an even number of vertices of $T$?}
In particular, we are interested in the problem when $G$ is a plane graph, which we call {\pptj}. We call $T$ the set of \emph{terminals} of the instance; $\green(G) \setminus T$ is the set of \emph{non-terminals}. We call $C \subseteq \black(G)$ a \emph{solution} to an instance of {\ptj} if $|C| \leq k$ and $G[C \cup \green(G)]$ contains a $T$-join.

\begin{lemma} \label{lem:oct:equiv}
If {\pptj} has a polynomial kernel, then {\ppoct} has a polynomial kernel.
\end{lemma}
\begin{proof}
By Lemma~\ref{lem:oct:radial}, the answer to a plane instance~$(G,T,k)$ of {\poct} is equivalent to the answer of the {\pptj} instance on the graph~$\radial(G)$, with the face vertices~$F(G)$ taking the role of~$\green$,~$V(G)$ taking the role of~$\black$, and~$T \subseteq F(G)$ being the odd faces. So if~{\pptj} has a polynomial kernel, then an instance of {\ppoct} can be compressed to size polynomial in~$k$ by transforming it into an instance of {\pptj} and applying the kernel to it. Since {\pptj} is in NP and {\ppoct} is NP-hard, by standard arguments (cf.~\cite{BodlaenderTY11}) the $T$-join instance can be reduced back to an instance of the original problem of size polynomial in~$k$, which forms the kernel.
\end{proof}
Below, we will give a polynomial kernel for {\pptj}. Combined with Lemma~\ref{lem:oct:equiv}, this implies a polynomial kernel for {\ppoct}.
\longversion{For most reduction rules, we will be able to give a direct analogue for {\ppoct}, which yields slightly lower constants.}

\subsection{Reducing the number of terminals} \label{sec:oct:degree}
Let $(G, \green(G), \black(G), T, k)$ be an instance of {\pptj}. As a first step, we show that the graph can be reduced so that there remain at most $6k^2$ terminals. To this end, we adapt the rules that Such\'y~\cite{Suchy2017} developed for \textsc{Plane Steiner Tree} parameterized by the number of Steiner vertices of the solution tree. Each of the rules is applied exhaustively before a next rule will be applied.

\begin{observation} \label{obs:oct:base}
Let $C$ be a solution for the instance. Then each vertex of $T$ has a neighbor in $C$.
\end{observation}
This is the analogue of~\cite[Lemma 2]{Suchy2017} and is immediate from the bipartiteness of $G$. \longversion{For {\ppoct}, the equivalent idea is the trivial observation that each odd face is incident to a vertex of the solution.}

\begin{observation} \label{obs:oct:simple}
If $k < 0$ or there is a connected component containing exactly one terminal~$t \in T$, then we can safely answer NO.
\end{observation}
\longversion{The equivalent rule for {\ppoct} is to answer NO when $k < 0$. The second part is irrelevant, because any plane graph has an even number of odd faces and we assumed that the graph is connected.}

\longversion{We now deal with sets of (false) twins.}

\begin{lemma} \label{lem:oct:twins}
Let $X \subseteq T$ be a maximal set such that $N_G(x) = N_G(y)$ for all $x,y \in X$. Remove all but $2-(|X| \bmod 2)$ vertices of $X$ from the graph and $T$. The resulting instance $(G', \green(G'), \black(G'), T', k)$ has a solution if and only if $(G, \green(G), \black(G), T, k)$ has a solution.
\end{lemma}
\begin{proof}
Let $Y \subseteq X$ be the set of remaining vertices of $X$. Observe that $|X| \equiv |Y| \pmod{2}$ and that $|Y| \geq 1$. The equivalence is now immediate.
\end{proof}
\longversion{The equivalent rule for {\ppoct} is to check whether $G$ itself is an odd cycle and then answer YES if $k \geq 1$ and NO otherwise. Indeed, twin face vertices only appear in the radial graph of $G$ if $G$ is just a single cycle.}

\begin{lemma} \label{lem:oct:empty-comp}
Let $u,v \in \black(G)$ and let $L = N_G(u) \cap N_G(v) \cap T$ with $L \not= \emptyset$. If a connected component $X$ of $G \setminus (L \cup \{u,v\})$ exists that contains no terminals, then remove $X$ from the graph. The resulting instance $(G', \green(G'), \black(G'), T, k)$ has a solution if and only if $(G, \green(G), \black(G), T, k)$ has a solution.
\end{lemma}
\begin{proof}
If $(G', \green(G'), \black(G'), T, k)$ has a solution $C$, then $C$ is a solution for the instance $(G, \green(G), \black(G), T, k)$, as $G'$ is an induced subgraph of $G$ with the same terminal set.

Suppose that $C$ is a minimal solution for $(G, \green(G), \black(G), T, k)$. We construct a solution for $(G', \green(G'), \black(G'), T, k)$ such that $C' \cap X = \emptyset$. Suppose that $C \cap X \not= \emptyset$; otherwise, let $C'=C$. Let $C' = (C \setminus X) \cup \{v\}$. In either case, $|C'| \leq |C| \leq k$. We claim that $C'$ is still a solution for $(G, \green(G), \black(G), T, k)$. 
To this end, first consider $C \cup \{v\}$. All connected components of $G[C \cup \green(G)]$ that neighbor $v$ will then be unified into a single connected component $Z$ of $G[C \cup \{v\} \cup \green(G)]$. 
The parity of $|Z \cap T|$ is equal to the sum (mod~$2$) of the parities of $|Z' \cap T|$ of the connected components $Z'$ of $G[C \cup \green(G)]$ neighboring~$v$. Since these parities are $0$, their sum is $0$, and $|Z \cap T|$ is even. Now consider the connected component $ZZ$ of $G[C' \cup \green(G)]$ that contains $v$. Clearly, $ZZ \cap T = Z \cap T$, because $X \cap T = \emptyset$ and any path in $G[C \cup \{v\} \cup \green(G)]$ that intersects $X$ can be re-routed through $v$ and the vertices of $L \subseteq \green(G)$. The claim follows, and thus the lemma as well.
\end{proof}
\longversion{The equivalent rule for {\ppoct} is to consider any two odd faces $f,f'$ that share at least two vertices. Now consider any simple cycle $Y$ in $G[V(f) \cup V(f')]$ such that the closure $R$ of one of the two regions of the plane obtained after the removal of $Y$ contains no odd face (and in particular does not contain $f$ nor $f'$). Since $R$ contains no odd faces, a parity argument shows that $Y$ is an even cycle. Moreover, $Y$ contains exactly two (non-adjacent) vertices $u,v$ of $V(f) \cup V(f')$. If such faces and cycle exist, then remove all vertices in 
the interior of $R$ and add the edge $uv$. In case the two paths between $u$ and $v$ on $Y$ are both odd, we additionally subdivide the newly added edge $uv$ to ensure that $f$ and $f'$ remain odd.}

We now present the final two reduction rules. Each relies on the following operation.

\begin{lemma} \label{lem:oct:op}
Let $v \in \black(G)$. Let $G'$ be obtained from $G$ by contracting all edges between $v$ and its neighbors in $G$. Let $v'$ be the resulting vertex, and let $\green(G')$ and $\black(G')$ be the resulting color classes, where $v' \in \green(G')$. Let $T'$ be obtained from $T$ by removing $N_G(v) \cap T$, and adding $v'$ to $T'$ if and only if $|N_G(v) \cap T| \equiv 1 \pmod{2}$.
\begin{itemize}
\item If $(G, \green(G), \black(G), T, k)$ has a solution $C$ with $v \in C$, then $(G', \green(G'), \black(G'), T', k-1)$ has a solution;
\item if $(G', \green(G'), \black(G'), T', k-1)$ has a solution, then $(G, \green(G), \black(G), T, k)$ has a solution.
\end{itemize}
\end{lemma}
\begin{proof}
Suppose there is a solution $C$ to $(G, \green(G), \black(G), T, k)$ such that $v \in C$. Then the vertices of $T \cap N_G(v)$ are in the same connected component $Z$ of $G[C \cup \green(G)]$. Let $C' = C \setminus\{v\}$ and let $Z'$ be obtained from $Z$ by contracting all edges between $v$ and $N_G(v)$. Then $Z'$ is a connected component of $G'[C' \cup \green(G')]$. By the construction of $T'$, $Z'$ contains an even number of vertices of $T'$. Moreover, $|C'| = |C|-1 \leq k-1$. Hence, $C'$ is a solution to $(G', \green(G'), \black(G'), T', k-1)$.

Suppose there is a solution $C'$ to $(G', \green(G'), \black(G'), T', k-1)$. Let $C = C' \cup \{v\}$. Let $Z'$ be the connected component of $G'[C' \cup \green(G')]$ that contains $v'$, and let $Z$ be obtained from $Z'$ by adding $N_G[v]$ and removing $v'$. Then $Z$ is a connected component of $G[C \cup \green(G)]$. Moreover, by the construction of $T'$, $Z$ contains an even number of vertices of $T$. Finally, $|C| = |C'|+1 \leq k$. Hence, $C$ is a solution to $(G, \green(G), \black(G), T, k)$.
\end{proof}
\longversion{The equivalent operation for {\ppoct} is simply the deletion of the vertex $v$ and reducing $k$ by $1$.}

\begin{lemma} \label{lem:oct:comp}
Let $u,v \in \black(G)$ and let $L = N_G(u) \cap N_G(v) \cap T$ with $L \not= \emptyset$. If a connected component $X$ of $G \setminus (L \cup \{u,v\})$ 
exists for which all terminals in $X \cap T$ neighbor $v$ and there is a solution $C$ to the instance $(G, \green(G), \black(G), T, k)$, then there is a solution that contains $v$.
\end{lemma}
\begin{proof}
Assume that $v \not\in C$, or the lemma would already follow. Since the rule of Lemma~\ref{lem:oct:empty-comp} is inapplicable, there is a terminal in $X$. Moreover, no terminal in $X \cap T$ neighbors $u$, because any such terminal would be in $L$ and thus not in $X$. Since every terminal has to have a neighbor in $C$, it follows that $C \cap X \not= \emptyset$. Therefore, $C' = (C \setminus X) \cup \{v\}$ is not larger than~$C$. We claim that $C'$ is still a solution. To this end, first consider $C \cup \{v\}$. All connected components of $G[C \cup \green(G)]$ that neighbor $v$ will then be unified into a single connected component $Z$ of $G[C \cup \{v\} \cup \green(G)]$. In particular, $Z$ contains $X \cap T$. The parity of $|Z \cap T|$ is equal to the sum (mod~$2$) of the parities of $|Z' \cap T|$ of the connected components $Z'$ of $G[C \cup \green(G)]$ that neighbor~$v$. Since these parities are $0$, their sum is $0$, and $|Z \cap T|$ is even. Now consider the connected component $ZZ$ of $G[C' \cup \green(G)]$ that contains $v$. Clearly, $ZZ \cap T = Z \cap T$, because any path in $G[C \cup \{v\} \cup \green(G)]$ that intersects $X$ can be re-routed through $v$ and the vertices of $L$.
The claim follows, and thus the lemma as well.
\end{proof}

\begin{lemma} \label{lem:oct:degree}
If there is a vertex $v \in \black(G)$ adjacent to more than $6k$ terminals and there is a solution $C$ to the instance $(G, \green(G), \black(G), T, k)$, then there is a solution that contains $v$.
\end{lemma}
\begin{proof}
The proof is completely analogous to the proof of~\cite[Lemma 11]{Suchy2017}. If $v \in C$, then we are done. So assume that $v \not\in C$. Let $B \subseteq C$ be the set of vertices in $C$ adjacent to at least two terminals in $N_G(v)$. Given $b \in B$, let $x,y$ be any two terminals in $N_G(b) \cap N_G(v)$ and consider the region $R$ that is enclosed by the cycle $x,b,y,v$ and that does not contain the outer face. If $R$ does not contain any other terminal of $N_G(b) \cap N_G(v)$, then $R$ is called the \emph{(internal) eye} of $x,b,y,v$. The \emph{support} of $b \in B$, denoted $\mbox{supp}(b)$, is the set of vertices $a \in B$ such that $a$ is contained inside an eye $R$ of $b$, but not inside an eye of any $b' \in B \setminus \{b\}$ for which $b'$ is inside $R$. The bound of $6k$ (instead of $5k$) ensures that the proof of~\cite[Lemma 16]{Suchy2017} can be modified (straightforwardly) to yield a vertex $b \in \black(G)$ adjacent to more than $2\,|\mbox{supp}(b)|+4$ vertices of $T$. The further arguments then imply the existence of a twin set in $T$ of size at least~$3$, thus contradicting the exhaustive execution of the rule of Lemma~\ref{lem:oct:twins}.
\end{proof}
\longversion{For {\ppoct}, the rule is already true for a bound of $5k$, because $T$ contains no twins.}

Lemma~\ref{lem:oct:comp} and~\ref{lem:oct:degree}, when combined with Lemma~\ref{lem:oct:op}, naturally lead to two reduction rules. After exhaustively applying all the reduction rules in this section, each vertex of $\black(G)$ neighbors at most $6k$ terminals.

\begin{observation} \label{obs:oct:size}
If $|T| > 6k^2$, then we can safely answer NO.
\end{observation}
This rule is immediate from Observation~\ref{obs:oct:base} and the fact that any solution contains at most $k$ vertices that are each adjacent to at most $6k$ terminals by Lemma~\ref{lem:oct:degree}. \longversion{For {\ppoct}, the rule is already true for a bound of $5k^2$.}

\subsection{Reducing the diameter and obtaining the kernel}
We now reduce the diameter of the graph. Our arguments here are a generalization of the arguments of Fiorini et al.~\cite{FioriniHRV2008} in their FPT-algorithm for {\ppoct}.

\begin{lemma} \label{lem:oct:dist}
Suppose there is a solution for $(G, \green(G), \black(G), T, k)$. Let $C$ be a minimal solution. Then each vertex $v \in C$ has distance at most $k+1$ in $G[C \cup \green(G)]$ to a vertex of~$T$.
\end{lemma}
\begin{proof}
Suppose for sake of contradiction that $v \in C$ has distance at least $k+1$ to each vertex of $T$ in $G[C \cup \green(G)]$. Since $C$ is minimal, there are two connected components $X$ and $Y$ of $G[(C \setminus\{v\}) \cup \green(G)]$ with an odd number of terminals. Let $x \in X \cap T$ and $y \in Y \cap T$. Consider a shortest path in $G[C \cup \green(G)]$ from $x$ to $v$. This path $P$ is fully contained in $G[V(X) \cup\{v\}]$ and has length at least $k+1$. As~$P$ connects vertices on opposite sides of the bipartite graph,~$|V(P) \cap C| \geq 1 + k/2$. Hence, $|V(X) \cap C| \geq k/2$. Similarly, $|V(Y) \cap C| \geq k/2$. Since $X$ and $Y$ are vertex disjoint, it follows that $|C| \geq 2 k/2 + 1 > k$, a contradiction.
\end{proof}

\begin{corollary} \label{cor:oct:dist}
Suppose there is a solution for $(G, \green(G), \black(G), T, k)$. Let $C$ be a minimal solution. Then 
every vertex of~$C \cup T$ has distance at most~$k+2$ to~$T$ in $G[C \cup \green(G)]$.
\end{corollary}

\begin{lemma} \label{lem:oct:split}
We can safely answer NO, or we can compute, in polynomial time, disjoint subgraphs $G_1,\ldots,G_\ell$ of $G$ for some $\ell \leq k$ such that:
\begin{enumerate}
\item the graphs $G_{i}$ jointly contain all terminals;
\item for each $i$ and for each vertex $v \in V(G_i)$, there is a terminal $t \in T \cap V(G_i)$ that can reach $v$ by a path of at most $k+2$ edges; \label{pty:split:shortpath}
\item for any solution $C$ for $(G, \green(G), \black(G), T, k)$, $C \cap V(G_i)$ is a solution for $(G_i, \green(G_i), \black(G_i), T \cap V(G_i), k_i)$ for each $i$, where $k_i = |C \cap V(G_i)|$;
\item if $(G_i, \green(G_i), \black(G_i), T \cap V(G_i), k_i)$ has a solution for each $i$ for some $k_1,\ldots,k_\ell \geq 0$ that sum up to at most $k$, then $(G, \green(G), \black(G), T, k)$ has a solution.
\end{enumerate}
\end{lemma}
\begin{proof}
For each terminal $t \in T$, let $B(t)$ be the set of all vertices within distance $k+2$ of $t$. Let $G_1,\ldots,G_\ell$ be the connected components of $G[\bigcup_{t \in T} B(t)]$. If $\ell > k$, then $G$ has more than $k$ terminals with disjoint neighborhoods in $\black(G)$, and we can safely answer NO. We now consider the properties set forth in the lemma statement:
\begin{enumerate}
\item True by construction and the definition of the function $B$.
\item True by construction and the definition of the function $B$.
\item True by construction, the definition of the function $B$, and Corollary~\ref{cor:oct:dist}.
\item We take the union $C$ of the solutions $C_i$ of the sub-instances. Note that the subgraphs $G_i$ are disjoint and thus contain disjoint sets of terminals. Hence, any connected component of $G[C \cup \green(G)]$ that contains connected components of $G[C_i \cup \green(G)]$ for multiple $i$, still contains an even number of terminals.
\end{enumerate}
This finishes the proof.
\end{proof}
\longversion{For {\ppoct}, we make sure to define $B(t)$ as the set of all faces and vertices within distance $k+2$ or $k+3$ in $\radial(G)$ of an odd face, depending on whether $k$ is odd or even respectively. This ensures that the graphs $G_i$ each correspond to a union of faces of $G$ and their incident vertices.}

Property~\ref{pty:split:shortpath} of Lemma~\ref{lem:oct:split} implies that each constructed subgraph $G_i$ has diameter~$\Oh(k \cdot |T \cap V(G_i)|)$, which is~$\Oh(k^3)$ using Observation~\ref{obs:oct:size}. The proof of Theorem~\ref{thm:pptj:kernel} employs an additional argument to obtain a quadratic-size Steiner tree to cut open. 

\longversion{We are now ready to present the kernel.}

\begin{theorem} \label{thm:pptj:kernel}
{\pptj} has a kernel of size $\Oh(k^{425})$.
\end{theorem}
\begin{proof}
We first exhaustively apply the reduction rules of Subsection~\ref{sec:oct:degree} until each vertex of $\black(G)$ neighbors at most $6k$ terminals. The rules can clearly be executed exhaustively in polynomial time. As per Observation~\ref{obs:oct:size}, we may assume that $|T| \leq 6k^2$. Then we apply Lemma~\ref{lem:oct:split} and consider each of the $\ell$ subgraphs $G_i$ separately. Let $T_i = T \cap V(G_i)$; note that $|T_i| \leq 6k^2$. Moreover, we can assume that $|T_i|$ is even, or we can safely answer NO. 

We construct a small set $A_i \subseteq V(G_i)$ such that $G[A_i]$ is connected and contains $T_i$. Start by adding $T_i$ to $A_i$. Then, we find a subset $T_i'$ of $T_i$ such that the sets $N_{G_i}(t)$ are pairwise disjoint for $t \in T_i'$ by the following iterative marking procedure: add any unmarked $t \in T_i$ to $T_i'$ and then mark all terminals in $N_{G_{i}}(N_{G_{i}}(t))$. It follows from Observation~\ref{obs:oct:base} that $|T_i'| \leq k$, or we can safely answer NO. 
Now apply Lemma~\ref{lem:construct:steinertree} to find a Steiner tree of at most $(2(k+2)+1)\, (|T_i'|-1)$ edges (and vertices) on $T_i'$. Add these vertices to $A_i$. Finally, for each $t \in T_i$, let $t' \in T_i'$ be a terminal such that $t \in N_{G_{i}}(N_{G_{i}}(t')) \cup \{t'\}$ and add an arbitrary vertex of $N(t) \cap N(t')$ to $A_i$. Then $|A_i| \leq 6k^2 + 6k^2 + (2k+5)\,|T_i'| = \Oh(k^2)$. Moreover, by construction, $G_i[A_i]$ is connected and contains $T_i$.


Let $S_i$ be a spanning tree of $G_i[A_i]$. Note that $S_i$ has size $\Oh(k^2)$ by the construction of $A_i$ and contains $T_i$. We cut the plane open along $S_i$ and make the resulting face the outer face. Let $\hat{G_i}$ denote the resulting plane graph. That is, we create a walk $W_i$ on the edges of $S_i$ that visits each edge of $S_i$ exactly twice. This walk has $\Oh(k^2)$ edges. Then we duplicate the edges of $S_i$ and duplicate each vertex $v$ of $S_i$ exactly $d_{S_i}(v)-1$ times, where $d_{S_i}(v)$ is the degree of $v$ in $S_i$. We can then create a face in the embedding that has $W_i$ as boundary. Then we obtain $\hat{G_i}$ by creating an embedding in which this new face is the outer face. See Figure~\ref{fig:cutopen}. 
This also yields a natural mapping $\pi$ from $E(\hat{G_i})$ to $E(G_i)$ and from $V(\hat{G_i})$ to $V(G_i)$. Finally, we observe that the terminals $T_i$ are all on the outer face of $\hat{G_i}$ and that $\hat{G_i}$ is a connected plane partitioned graph.

\begin{figure}
\centering
\iflipics{%
\includegraphics[width=.6\linewidth]{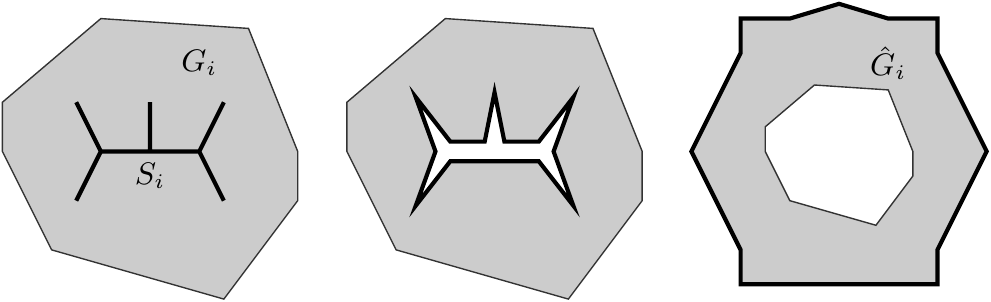}%
}{%
\includegraphics[width=.6\linewidth]{fig-cutopen}%
}
\caption{The process of cutting open the graph $G_i$ along the tree $S_i$. Adapted from~\cite{pst-kernel} with permission.}
\label{fig:cutopen}
\end{figure}

Now apply Theorem~\ref{thm:sparse} to $\hat{G_i}$ and let $\tilde{G_i}$ be the resulting graph. Let $F_i = \pi(\tilde{G_i})$. Note that $\tilde{G_i}$ has $\Oh(|\partial \hat{G_i}|^{212}) = \Oh(|W_i|^{212}) = \Oh(k^{424})$ edges, and thus so has $F_i$. Finally, let $F = \bigcup_{i=1}^{\ell} F_i$. Clearly, $|F| = \Oh(k^{425})$, as $\ell < k$. Also note that each of the reduction rules, the above marking procedures, and $F$ itself can be computed in polynomial time.

We claim that $(F, \green(F), \black(F), T, k)$ is a kernel. Since $F$ is a subgraph of $G$, it follows that if $(F, \green(F), \black(F), T, k)$ has a solution, then so does $(G, \green(G), \black(G), T, k)$. Now let $C$ be a minimum solution for $(G, \green(G), \black(G), T, k)$. Then $C_i = C \cap V(G_i)$ is a solution for $(G_i, \green(G_i), \black(G_i), T \cap V(G_i), k_i)$ for each $i$, where $k_i = |C \cap V(G_i)|$. Consider some $i$ and let $J_i$ be a $T$-join of $G_i[C_i \cup \green(G_i)]$. 

Let $Z$ be a connected component of $G_i[J_i]$. We show how to find a connected subgraph $Z'$ of $F_i$ (and thus of $G_i$) such that $V(Z') \cap T_i \supseteq V(Z) \cap T_i$ and $|V(Z') \cap \black(G_i)| \leq |V(Z) \cap \black(G_i)|$. Consider the subgraph $\hat{Z}$ of $\hat{G_i}$ formed by $\pi^{-1}(V(Z) \cup E(Z))$. Note that any connected component $Y$ of $\hat{Z}$ connects $A = V(Y) \cap \partial \hat{G_i}$. 
Then by Theorem~\ref{thm:sparse}, there is a subgraph $H(Y)$ of $\tilde{G_i}$ that connects $A$ and has minimum possible cost among all subgraphs of $G_i$ that connect $A$. Hence, $|V(H(Y)) \cap \black(G_i)| \leq |V(Y) \cap \black(G_i)|$. Now let $H$ be the union of $H(Y)$ over all connected components $Y$ of $\hat{Z}$. Observe that $H$ is a subgraph of $\tilde{G_i}$. Let $Z' = (\pi(V(H)), \pi(E(H)))$. Observe that, by construction, $Z'$ is a subgraph of $F_i$ with the claimed properties.
In particular, observe that although $\hat{Z}$ can be much larger than $Z$ due to the duplication of vertices of $Z \cap S_i$ when $G_i$ was cut open along $S_i$, we de-duplicate these vertices when using $\pi(V(H))$. 

Consider the union $J_i'$ of all these connected subgraphs $Z'$ over all connected components $Z$ of $G_i[J_i]$. Then $|V(G_i[J_i']) \cap \black (G_i)| \leq |V(G_i[J_i]) \cap \black (G_i)| = |C_i|$. Moreover, by construction, for each connected component $ZZ$ of $G_i[J_i']$ there exists a set $\mathcal{Z}(ZZ)$ of connected components $Z$ of $G_i[J_i]$ such that $ZZ \cap T$ is the union of $Z \cap T$ over all these connected components $Z$. We note that the sets $\mathcal{Z}(ZZ)$ induce a partition of the connected components of $G_i[J_i]$. Observe that the parity of $|ZZ \cap T|$ is equal to the sum (mod~$2$) of the parities of the corresponding connected components of $G_i[J_i]$, and thus, equal to $0$. It follows that $G_i[J_i']$ contains a $T_i$-join. Hence, $V(G_i[J_i']) \cap \black(G_i)$ is a solution for $(G_i, \green(G_i), \black(G_i), T_i, k_i)$. By repeating this procedure for all $i$, it follows from the proof of Lemma~\ref{lem:oct:split} that the union of these solutions is a solution for $(G,\green(G),\black(G),T,k)$. Moreover, any $T$-join that is contained in this solution is fully contained in $F$. Hence, $(F, \green(F), \black(F), T, k)$ has a solution\shortversion{.}\longversion{, and the claim follows.}
\end{proof}
\longversion{Combining Lemma~\ref{lem:oct:equiv} and Theorem~\ref{thm:pptj:kernel}, we obtain the following:}

\begin{corollary} \label{cor:ppoct:kernel}
{\ppoct} has a polynomial kernel.
\end{corollary}

\section{Vertex Multiway Cut}\label{sec:mwc}\label{app:mwc}

In this section we develop a polynomial kernel for \pmwc, which is formally defined as follows.

\defparproblem{\pmwc (\mwc)}
{A planar graph~$G$, a set of terminals~$T \subseteq V(G)$, and an integer~$k$.}
{$k$.}
{Is there a vertex set~$X \subseteq V(G) \setminus T$ of size at most~$k$ such that each connected component of~$G \setminus X$ contains at most one terminal?}

We refer to a vertex multiway cut for terminal set~$T$ as a \emph{$T$-multiway cut}. Before presenting the technical details, we describe the main idea in applying the Steiner tree sparsification to \pmwc. Let~$G$ be a plane graph whose outer face walk~$\partial G$ is a simple cycle and let~$X \subseteq V(G)$. Obtain~$\hat{G}$ from~$G$ by inserting a vertex~$v_f$ inside every finite face~$f$, making~$v_f$ adjacent to all vertices that lie on~$f$; that is, $\hat{G}$ is an overlay graph after removal of the vertex corresponding to the outer face. Let~$F$ denote the set of inserted face-vertices. Suppose that set~$X' \subseteq V(G)$ \emph{mimics} the set~$X$ in the following way: for each~$Y \subseteq X \cap V(\partial \hat{G})$ which is contained in a single connected component of~$\hat{G}[X \cup F]$, the set~$Y$ is contained in a single connected component of~$\hat{G}[X' \cup F]$. Then for every pair of vertices~$u,v \in V(\partial G)$, if~$X$ is a~$(u,v)$-cut in~$G$ then~$X'$ is also a~$(u,v)$-cut in~$G$. Hence by preserving minimum connectors for subsets of boundary vertices of~$\hat{G}$, we can preserve minimum solutions to \pmwc in the setting that all terminals lie on the outer face. In combination with a cutting-open step to achieve such a setting, this yields the kernel.

\subsection{Preparing the graph to be cut open}

To allow the graph to be ``cut open'' to obtain a face of size~$k^{\Oh(1)}$ that contains all neighbors of the terminals, we employ a preprocessing step based on \emph{outerplanarity layers} of the graph.

\begin{definition}
Let~$G$ be a plane graph. The outerplanarity layers of~$G$ form a partition of~$V(G)$ into~$L_1, \ldots, L_m$ defined recursively. The vertices incident with the outer face of~$G$ belong to layer~$L_1$. If~$v \in V(G)$ lies on the outer face of the plane subgraph obtained from~$G$ by removing~$L_1, \ldots, L_i$, then~$v$ belongs to layer~$L_{i+1}$. For a vertex~$v \in V(G)$, the unique index~$i \in [m]$ for which~$v \in L_i$ is the \emph{outerplanarity index} of~$v$ and denoted~$\opindex_G(v)$.
\end{definition}

\begin{definition}
For a plane graph~$G$, let~$\tree(G)$ be the simple graph obtained by simultaneously contracting all edges~$uv$ whose endpoints belong to the same outerplanarity layer, discarding loops and parallel edges. For a node~$u \in V(\tree(G))$, let~$\kappa(u) \subseteq V(G)$ denote the vertex set of~$G$ whose contraction resulted in~$u$. Let~$\opindex_G(u)$ denote the outerplanarity index that is shared by all nodes in~$\kappa(u)$. For a subtree~$\tree'$ of~$\tree(G)$, define~$\kappa(\tree') := \bigcup _{v \in V(\tree')} \kappa(v)$. 

For nonadjacent nodes~$x,y$ of a tree~$\tree'$, let~$\int_{\tree'}(x,y)$ denote the set of internal nodes on the unique simple~$xy$-path in~$\tree'$. We define the following:
\begin{itemize}
	\item $\tree'[xy,x]$ is the tree in the forest~$\tree' \setminus \int_{\tree'}(x,y)$ that contains~$x$.
	\item $\tree'[xy,\int]$ is the tree in the forest~$\tree' \setminus \{x,y\}$ that contains~$\int_{\tree'}(x,y)$.
\end{itemize}
\end{definition}

We omit the argument~$G$ when it is clear from the context. We utilize several properties of this definition.

\begin{lemma} \label{lemma:layertree}
For a connected plane graph~$G$ and the associated contraction~$\tree(G)$ based on the outerplanarity layers~$L_1, \ldots, L_m$, the following holds:
\begin{enumerate}
	\item $\tree(G)$ is a tree.\label{prop:tree}
	\item For any node~$u \in V(\tree(G))$, the graph~$G[\kappa(u)]$ is connected.\label{prop:connected} 
	\item If~$x,y$ are distinct vertices of~$\tree(G)$, then for any internal node~$z$ of the unique $xy$-path in~$\tree(G)$, the set~$\kappa(z)$ is a $(\kappa(x), \kappa(y))$-cut in~$G$.\label{prop:cut}
	\item There is a unique node~$v_1 \in V(\tree(G))$ such that~$\kappa(v_1) = L_1$. \label{prop:root:layer}
	\item Root~$\tree(G)$ at vertex~$v_1$. If~$u \in \tree(G)$ is a child of~$p \in \tree(G)$, then~$\opindex_G(u) = 1+\opindex_G(p)$.\label{prop:idx:child}
\end{enumerate}
\end{lemma}
\begin{proof}
\eqref{prop:tree} Proof by induction on the number~$m$ of outerplanarity layers of~$G$. If~$m=1$ then all vertices are in the same outerplanarity layer. The contraction process therefore contracts \emph{all} edges, which results in the $1$-vertex tree since~$G$ is connected.

For~$m > 1$, let~$L_1, \ldots, L_m$ be the outerplanarity layers of~$G$ and let~$G' := G \setminus V(L_1)$. Let~$G_1, \ldots, G_t$ be the connected components of~$G'$. Since~$G$ is connected, the graph~$G[L_1]$ is connected and will be contracted to a single vertex by the process that builds~$\tree(G)$. Again using connectivity, each connected component~$G_i$ of~$G'$ is adjacent to a vertex on~$L_1$. By planarity of the embedding, vertices of~$L_1$ are only adjacent to vertices of~$L_2$. Hence~$\tree(G)$ can be obtained by taking the disjoint union of~$\tree(G_1), \ldots, \tree(G_t)$, adding a new root node~$r$ to represent the contraction of~$L_1$, and for each~$i \in [t]$ making~$r$ adjacent to the unique node of~$\tree(G_i)$ that represents the contraction of~$L_2 \cap V(G_i)$. Since each~$\tree(G_i)$ is a tree by induction, the resulting graph~$\tree(G)$ is a tree.

\eqref{prop:connected} By definition, for a node~$u \in V(\tree(G))$ the set~$\kappa(u)$ consists of the vertices in~$G$ whose contraction resulted in~$u$. So any two vertices of~$\kappa(u)$ are connected in~$G$ by a path of contracted edges.

\eqref{prop:cut} By the nature of contractions we have for any edge~$uv$ in~$G$ that~$u$ and~$v$ are contracted to the same node of~$\tree(G)$, or are contracted to adjacent nodes of~$\tree(G)$. This means that any~$uv$-path in~$G$ projects to a walk in~$\tree(G)$ connecting the contraction of~$u$ to the contraction of~$v$. So if~$x,y$ are distinct vertices of~$\tree(G)$ and~$z$ is an internal node of the $xy$-path in~$\tree(G)$, then any $\kappa(x)\kappa(y)$-path~$P$ in~$G$ projects to an $xy$-walk in~$\tree(G)$, which therefore visits~$z$. Consequently,~$P$ contains a vertex of~$\kappa(z)$. Hence~$\kappa(z)$ intersects every path between~$\kappa(x)$ and~$\kappa(y)$ in~$G$ and forms a~$(\kappa(x), \kappa(y))$-cut.

\eqref{prop:root:layer} Since~$G$ is connected, the graph~$G[L_1]$ is connected. Hence all vertices of~$L_1$ are contracted into a single node~$v_1$ when forming~$\tree(G)$, while by definition no vertices of~$\bigcup_{i \geq 2} L_i$ are contracted into~$v_1$.

\eqref{prop:idx:child} Proof by induction on~$\opindex_G(p)$. If~$\opindex_G(p) = 1$, then~$\kappa(p)$ is the set of vertices on the outer face of~$G$. Since~$up$ is an edge of~$\tree(G)$ that was obtained by contractions, there is a vertex~$v_u \in \kappa(u)$ that is adjacent in~$G$ to a vertex~$v_p \in \kappa(p)$. When removing~$L_1$ and all its incident edges from~$G$, the vertex~$v_p$ disappears from the outer face, and also the drawing of the edge~$v_u v_p$ disappears from the outer face, thereby placing~$v_u$ on the outer face of the resulting graph. It follows that~$\opindex_G(u) = \opindex_G(v_u) = 2 = 1 + \opindex_G(v_p) = 1 + \opindex_G(p)$.

For~$\opindex_G(p) > 1$, consider the graph~$G'$ obtained from~$G$ by removing the vertices on the outer face. Since~\eqref{prop:root:layer} shows that the root~$v_1$ is the unique vertex of~$\tree(G)$ with~$\opindex_G(v_1) = 1$, while~$\opindex_G(p) > 1$, it follows that~$\kappa(p)$ is disjoint from~$L_1$. Since~$u$ is a child of~$p$ and therefore not the root, we also know that~$\kappa(u)$ is disjoint from~$L_1$. From the given definitions, it follows that for all vertices of~$G'$, their outerplanarity index in~$G'$ is one smaller than in~$G$. Moreover, from the construction of~$\tree(G)$ and the proof of \eqref{prop:tree} it follows that~$\tree(G)$ can be obtained from the trees~$\tree(G'_i)$ of the connected components~$G'_i$ of~$G$ by inserting a root vertex and attaching it to the single vertex in each of the trees~$\tree(G'_i)$ that represents vertices of~$L_2$. Hence there is a component~$G'_i$ of~$G'$ that contains~$\kappa(u)$ and~$\kappa(p)$, and~$\tree(G'_i)$ corresponds to a subtree of~$\tree(G)$ rooted at a child of the root. Since the outerplanarity index of vertices in~$G'$ is exactly one lower than in~$G$, while~$u$ is a child of~$p$ in~$\tree(G'_i)$, \eqref{prop:idx:child} follows by induction.
\end{proof}

At several points in the remainder, we will consider~$\tree(G)$ to be a \emph{rooted tree} rooted at the node identified in~\eqref{prop:root:layer}. We also utilize the following property of the contraction~$\tree(G)$.

\begin{observation} \label{observation:child:commonface}
If~$G$ is a connected plane graph and node~$u$ is adjacent to node~$p$ in~$\tree(G)$, then for every vertex~$v_u \in \kappa(u) \subseteq V(G)$ there is a vertex~$v_p \in \kappa(p) \subseteq V(G)$ such that~$u$ and~$v$ are incident on a common face of~$G$.
\end{observation}

\begin{definition} \label{def:tree:important}
Let~$G$ be a connected plane graph, let~$T \subseteq V(G)$ be a set of terminals, and consider the contraction~$\tree(G)$. We call a node~$u \in V(\tree(G))$ a \emph{terminal} node if~$\kappa(u) \cap T \neq \emptyset$. We call a node~$u \in \tree(G)$ \emph{important} if it is a terminal node, the root of~$\tree(G)$, or if at least two subtrees rooted at distinct children of~$u$ contain a terminal node.
\end{definition}

The above definition implies that the set of important nodes is closed under taking lowest common ancestors. Standard counting arguments on trees show that the number of non-root non-terminal important nodes is less than the number of terminal nodes. This yields the following:

\begin{observation} \label{obs:num:important}
For any connected plane graph~$G$ and terminal set~$T$, there are at most~$2|T|$ important nodes in~$\tree(G)$.
\end{observation}

The following lemma shows how this tree structure helps to identify vertices that can be avoided by an optimal solution.

\begin{lemma} \label{lemma:move:solution}
Let~$(G,T,k)$ be a connected plane instance of \mwc. Let~$u,v$ be distinct important nodes of~$\tree := \tree(G)$ and~$w_1, \ldots, w_r$ be the internal nodes of the~$uv$-path in~$\tree$ in their natural order, such that~$r \geq 2(k+1)$ and none of the nodes in~$\tree[uv,\int]$ are important. For each pair of nodes~$x \in \{w_1, \ldots, w_{k+1}\}$ and~$y \in \{w_{r-k}, \ldots, w_r\}$, let~$K_{xy}$ be a minimum~$(\kappa(x),\kappa(y))$-cut in~$G$ if such a cut has size at most~$k$, and let~$K_{xy} := \emptyset$ otherwise.

Then for any $T$-multiway cut~$X \subseteq V(G) \setminus T$ in~$G$ of size at most~$k$, there is a $T$-multiway cut~$X'$ such that:
\begin{enumerate}
	\item $|X'| \leq |X|$.
	\item $X \setminus \kappa(\tree[uv,\int]) = X' \setminus \kappa(\tree[uv,\int])$.\label{replacesolution:onlyinside}
	\item $X' \cap \kappa(\tree[w_{k+1} w_{r-k}, \int]) \subseteq \bigcup_{x,y} K_{xy}$ where~$x,y$ ranges over combinations as above.\label{controlled:middle}
\end{enumerate}
\end{lemma}

The lemma essentially says that in the subgraph~$H$ of~$G$ corresponding to a long path of~$\tree(G)$ without important nodes, any multiway cut~$X$ of size at most~$k$ in~$G$ can be rearranged so that from the vertices represented in the middle of the path (strictly between~$w_{k+1}$ and~$w_{r-k}$), it only removes vertices from one of the size-at-most-$k$ cuts~$K_{xy}$.

\begin{proof}[Proof of Lemma \ref{lemma:move:solution}]
Assume the stated conditions hold and let~$X$ be a multiway cut of size at most~$k$. The sets~$\kappa(w_1), \ldots, \kappa(w_r)$ are vertex-disjoint, implying that~$X$ is disjoint from one of the first~$k+1$ and one of the last~$k+1$ sets. Choose~$x$ from the first~$k+1$ sets and~$y$ from the last~$k+1$ sets such that~$X \cap \kappa(x) = X \cap \kappa(y) = \emptyset$. We distinguish two cases, depending on whether~$X$ forms a~$(\kappa(x), \kappa(y))$-cut or not.

\begin{claim} \label{claim:nocut:fixsolution}
If~$X$ is not a~$(\kappa(x), \kappa(y))$-cut in~$G$, then~$X' := X \setminus \kappa(\tree[xy, \int])$ is a multiway cut in~$G$.
\end{claim}
\begin{claimproof}
Suppose that~$X$ is not a~$(\kappa(x), \kappa(y))$-cut in~$G$, implying that there is a path from a vertex in~$\kappa(x)$ to a vertex in~$\kappa(y)$ in~$G\setminus X$. Since the sets~$G[\kappa(x)]$ and~$G[\kappa(y)]$ are connected by Lemma~\ref{lemma:layertree}, every vertex of~$\kappa(x)$ can reach every vertex of~$\kappa(y)$ in~$G\setminus X$.

Assume for a contradiction that~$X'$ is not a multiway cut in~$G$ and let~$t,t'$ be distinct terminals that are connected by a path~$P$ in~$G \setminus X'$. Since~$X$ is a multiway cut while~$X'$ is not, we know~$X \cap \kappa(\tree[xy, \int])$ contains a vertex of~$P$. We make a further distinction on where~$t$ and~$t'$ are represented in~$\tree$. Let~$u_t, u_{t'} \in V(\tree)$ such that~$t \in \kappa(u_t)$ and~$t' \in \kappa(u_{t'})$. Observe that~$\tree[xy, \int] \subseteq \tree[uv, \int]$ contains no important nodes; in particular, $\kappa(\tree[xy, \int])$ contains no terminals. Moreover, $\tree[xy,x], \tree[xy,y], \tree[xy,\int]$ partition the node set of~$\tree$. It follows that each of~$\{u_t,u_{t'}\}$ belongs to~$\tree[xy,x]$ or~$\tree[xy,y]$. We make a distinction depending on whether these important nodes belong to the same tree, or to different trees.

First, suppose~$u_t$ and~$u_{t'}$ belong to the same tree and assume without loss of generality (by symmetry) that~$u_t, u_{t'} \in V(\tree[xy,x])$. This implies that in~$\tree$, node~$x$ intersects every path from~$u_t$ to~$\int(x,y)$, and also intersects every path from~$u_{t'}$ to~$\int(x,y)$. Since~$P$ contains a vertex from~$X \cap (\kappa(\tree[xy,\int]))$, by Lemma~\ref{lemma:layertree} this implies that starting from~$t$, path~$P$ visits a vertex from~$\kappa(x)$ before any vertex of~$\kappa(\tree[xy,\int])$. Similarly, traversing~$P$ backwards from~$t'$ we have that~$P$ visits a vertex from~$\kappa(x)$ before any vertex of~$\kappa(\tree[xy,\int])$. Let~$P_t$ denote the subpath from~$t$ to the first vertex in~$\kappa(x)$, and~$P_{t'}$ the subpath from the last vertex of~$\kappa(x)$ to~$t'$. Then paths~$P_t$ and~$P_{t'}$ are disjoint from~$X$, since they are disjoint from~$\kappa(\int(x,y))$ while~$X \setminus \kappa(\int(x,y))$ equals~$X' \setminus \kappa(\int(x,y))$. But then we can find a path from~$t$ to~$t'$ in~$G\setminus X$ by first traversing~$P_t$, moving from the~$\kappa(x)$-endpoint of~$P_t$ to the~$\kappa(x)$-endpoint of~$P_{t'}$ through~$G[\kappa(x)]$ which is connected and disjoint from~$X$, and ending with~$P_{t'}$. This contradicts that~$X$ is a solution in~$G$.

Now suppose that~$u_t$ and~$u_{t'}$ belong to different trees, say~$u_t \in V(\tree[xy,x])$ and~$u_{t'} \in V(\tree[xy,y])$. Since~$P$ visits a vertex of~$\tree[xy,\int]$, while~$x$ separates~$u_t$ from~$\tree[xy,\int]$ in~$\tree$, we know~$P$ passes through~$\kappa(x)$; let~$P_t$ be the subpath from~$P$ starting at~$t$ to its first occurrence of~$\kappa(x)$, which is disjoint from~$\kappa(\tree[xy,\int])$. Similarly,~$y$ separates~$\tree[xy,\int]$ from~$u_{t'}$ in~$\tree$, and therefore the subpath~$P_{t'}$ from the last occurrence of~$\kappa(y)$ to~$t'$ is disjoint from~$\tree[xy,\int]$. Let~$P_{xy}$ be a path in~$G\setminus X$ from the~$\kappa(x)$-endpoint of~$P_t$ to the~$\kappa(y)$-endpoint of~$P_{t'}$, which exists by the observation in the beginning of the proof. But then the concatenation of~$P_t, P_{xy}$, and~$P_{t'}$ forms an~$xy$-path in~$G\setminus X$, contradicting that~$X$ is a multiway cut.
\end{claimproof}

\begin{claim}
If~$X$ is a~$(\kappa(x), \kappa(y))$-cut in~$G$, then~$X' := (X \setminus \kappa(\tree[xy, \int])) \cup K_{xy}$ is a $T$-multiway cut in~$G$ with~$|X'| \leq |X|$.
\end{claim}
\begin{claimproof}
We first show that~$X \cap \kappa(\tree[xy,\int])$ is also a~$(\kappa(x), \kappa(y))$-cut in~$G$. This follows from the fact that all minimal paths from a vertex of~$\kappa(x)$ to a vertex of~$\kappa(y)$ have their interior vertices in~$\kappa(\tree[xy,\int])$ as a consequence of Lemma~\ref{lemma:layertree}. This also implies that~$K_{xy} \subseteq \kappa(\tree[xy,\int])$. Since~$|X| \leq k$ and~$X$ is a~$(\kappa(x), \kappa(y))$-cut in~$G$, a minimum~$(\kappa(x), \kappa(y))$-cut in~$G$ also has size at most~$k$ and therefore~$K_{xy}$ is such a cut, rather than~$\emptyset$.

Since~$X \cap \kappa(\tree[xy,\int])$ is a~$(\kappa(x), \kappa(y))$-cut in~$G$, which is replaced by the \emph{minimum}~$(\kappa(x), \kappa(y))$-cut~$K_{xy}$ to obtain~$X'$, it follows that~$|X'| \leq |X|$. It remains to prove that~$X'$ is a multiway cut in~$G$.

Assume for a contradiction that~$P$ is a path between distinct terminals~$t,t'$ in~$G \setminus X'$, and let~$u_t, u_{t'} \in V(\tree)$ such that~$t \in \kappa(u_t)$ and~$t' \in \kappa(u_{t'})$. If~$u_t$ and~$u_{t'}$ both belong to the same tree of~$\{\tree[xy,x], \tree[xy,y]\}$, then the same argument as in the proof of Claim~\ref{claim:nocut:fixsolution} yields a contradiction, since we have not changed the solution outside the set~$\kappa(\tree[xy,\int])$. 
So each of the two trees contains one of the important nodes~$u_t, u_{t'}$, as important nodes cannot be in~$\tree[xy,\int]$. This implies that both~$x$ and~$y$ lie on the path in~$\tree$ from~$u_t$ to~$u_{t'}$. By Lemma~\ref{lemma:layertree}, both~$\kappa(x)$ and~$\kappa(y)$ are~$(t,t')$-cuts in~$G$. So~$P$ intersects both~$\kappa(x)$ and~$\kappa(y)$, and provides a path from~$\kappa(x)$ to~$\kappa(y)$ in~$G \setminus K_{xy}$. But this contradicts that~$K_{xy}$ is a~$(\kappa(x), \kappa(y))$-cut in~$G$.
\end{claimproof}

Using these two claims, we complete the proof. In both cases, we have shown that~$X'$ is a multiway cut not larger than~$X$. We only remove or replace a part of~$X$ within~$\kappa(\tree[uv, \int])$, while the replacement cut~$K_{xy}$ is contained in~$\kappa(\tree[xy,\int]) \subseteq \kappa(\tree[uv,\int])$ as noted in the second claim. Hence we have~\eqref{replacesolution:onlyinside}. Finally, we built~$X'$ from~$X$ by removing all vertices from~$\kappa(\tree[xy,\int])$ and inserting some cut~$K_{xy}$. Since~$x$ is one of the first~$k+1$ nodes on the~$uv$-path in~$\tree$, and~$y$ is one of the last~$k+1$ nodes, it follows that the intersection of~$X'$ with the nodes in the middle is a subset of a cut~$K_{xy}$, which yields~\eqref{controlled:middle}. This concludes the proof of Lemma~\ref{lemma:move:solution}.
\end{proof}

Using Lemma~\ref{lemma:move:solution} we can identify a set of nodes~$R \subseteq V(\tree(G))$ in a connected instance~$(G,T,k)$ of \mwc, with the guarantee that if there is a solution of size at most~$k$, there is one contained in~$\kappa(R)$. If we treat vertices outside~$\kappa(R)$ as `forbidden to be deleted', then we can safely contract edges between such vertices without changing the answer. As we show later, these contractions reduce the diameter of the plane graph and allow us to find a bounded-size tree in an overlay graph that can be used to cut the graph open.

\begin{lemma} \label{lemma:relevant:treenodes}
There is a polynomial-time algorithm that, given a connected plane instance~$(G,T,k)$ of \mwc, outputs~$R \subseteq V(\tree(G))$ containing all important nodes such that:
\begin{enumerate}
	\item for any $T$-multiway cut~$X \subseteq V(G) \setminus T$ of size at most~$k$ in~$G$, there is a $T$-multiway cut~$X' \subseteq V(G) \setminus T$ of size at most~$|X|$ with~$X' \subseteq \kappa(R)$.\label{pty:preserves:opt}
	\item for each important node~$u$ of~$\tree(G)$, the path in~$\tree(G)$ from~$u$ to the root contains~$\Oh(|T| k^3)$ nodes from~$R$.\label{pty:short:path}
\end{enumerate}
\end{lemma}
\begin{proof}
Consider~$\tree := \tree(G)$. From~$\tree$ we derive a related tree~$\tree'$ whose vertex set consists solely of the \emph{important} nodes of~$\tree$ (Definition~\ref{def:tree:important}) by repeatedly contracting edges that have one important endpoint and one non-important endpoint. We think of contracting such edges into the important endpoint, whose identity is preserved. The resulting tree~$\tree'$ has at most~$2|T|$ nodes by Observation~\ref{obs:num:important}. Since the set of important nodes is closed under taking lowest common ancestors, for each edge~$uv$ of~$\tree'$, either~$u$ and~$v$ are adjacent in~$\tree$ or~$\tree[uv,\int]$ contains only non-important nodes. Define~$R \subseteq V(\tree)$ as follows:
\begin{itemize}
	\item Initialize~$R \subseteq V(\tree)$ as the important nodes of~$\tree$.
	\item Add every node of~$\tree$ that does not lie on a simple path between important nodes to~$R$.
	\item For each edge~$uv$ of~$\tree'$ between important nodes~$u$ and~$v$, do the following.
	\begin{itemize}
		\item If~$\int_\tree(u,v)$ contains at most~$2(k+1)$ nodes, then add all nodes of~$\tree[uv,\int]$ to~$R$.
		\item If~$\int_\tree(u,v)$ contains more than~$2(k+1)$ nodes~$w_1, \ldots, w_r$, then add all nodes of~$\tree[uv,\int] \setminus \tree[w_{k+1} w_{r-k}, \int]$ to~$R$. (The added nodes are the first and last~$k+1$ internal nodes of the~$uv$-path, together with all nodes that do not lie on the simple~$uv$-path but belong to pendant trees that attach to the first or last~$k+1$ internal nodes of the~$uv$-path.) Then consider the set of at most~$(k+1)^2$ cuts~$K_{xy} \subseteq \kappa(\tree[uv,\int])$ of size at most~$k$ described in Lemma~\ref{lemma:move:solution}. Add to~$R$ each node~$z$ of~$\tree[uv,\int]$ for which~$\kappa(z) \cap (\bigcup _{x,y} K_{xy}) \neq \emptyset$.
	\end{itemize}
\end{itemize}

It is straightforward to implement this procedure in polynomial time. We prove that the set~$R$ defined in this way has the desired properties.

First consider~\eqref{pty:preserves:opt}. The only nodes of~$T$ that do \emph{not} belong to~$R$, are those nodes~$z$ that belong to~$\int_\tree(u,v)$ for an edge~$uv$ of~$\tree'$, but which are neither among the first nor last~$(k+1)$ nodes of~$\int_\tree(u,v)$, are not contained in trees that attach to one of the first or last~$k+1$ nodes of~$\int_\tree(u,v)$, and for which~$\kappa(z)$ contains no vertex of a selected cut~$K_{xy}$. If there is a multiway cut~$X$ in~$G$ of size at most~$k$, then we can apply Lemma~\ref{lemma:move:solution} once for each edge~$uv$ of~$\tree'$ to transform this into another multiway cut~$X'$ with~$|X'| \leq |X|$. Applying the lemma to an edge~$uv$ with~$\int_\tree(u,v) = \{w_1, \ldots, w_r\}$ ensures that~$X' \cap \kappa(\tree[w_{k+1} w_{r-k}, \int]) \subseteq \bigcup_{x,y} K_{xy} \subseteq \kappa(R)$. Since the second property of Lemma~\ref{lemma:move:solution} ensures that the solution stays the same outside~$\kappa(\tree[uv,\int])$, by applying the lemma for each edge of~$\tree'$ we find a multiway cut~$X'$ with~$|X'| \leq |X|$ and~$X' \subseteq \kappa(R)$.

To prove~\eqref{pty:short:path}, consider an important node~$u$ of~$\tree$. If~$u$ is the root then the desired property is trivial. Otherwise, let~$u = v_1, \ldots, v_s$ be the important nodes that lie on the path from~$u$ to the root and recall that the root itself is important. Then~$\tree'$ contains edges~$v_i v_{i+1}$ for~$i \in [s-1]$. The nodes visited by path~$P$ are exactly~$\{v_1, \ldots, v_s\} \cup (\bigcup_{i \in [s-1]} \int_\tree(v_i, v_{i+1}))$. By Observation~\ref{obs:num:important} we have~$s \leq 2|T|$. For each~$i \in [s-1]$, we have~$|R \cap \int_\tree(v_i, v_{i+1})| \leq \Oh(k^3)$: the first and last~$k+1$ nodes on the~$v_i v_{i+1}$-path belong to~$R$, along with each node~$z$ of~$\tree$ for which~$\kappa(z) \cap (\bigcup_{x,y} K_{xy}) \neq \emptyset$. Since we considered~$\Oh(k^2)$ cuts~$K_{x,y}$, each of size at most~$k$, the claimed bound of~$\Oh(k^3)$ nodes for each~$i \in [s-1]$ follows. As~$s \leq 2|T|$, this proves~\eqref{pty:short:path}.
\end{proof}

\begin{definition}
Let~$(G,T,k)$ be a connected plane instance of \mwc and~$\tree := \tree(G)$. For a set~$R \subseteq V(\tree)$, define~$G/\overline{R}$ as the simple plane graph~$G'$ obtained by simultaneously contracting each edge for which neither endpoint belongs to~$\kappa(R)$ (and removing loops and parallel edges).
\end{definition}

Recall that~$d^\radial_G(u,v)$ denotes the radial distance between~$u$ and~$v$ in a plane graph~$G$. 

\begin{lemma} \label{lemma:contraction:reduces:radialdist}
Let~$(G,T,k)$ be a connected plane instance of \mwc. Let~$R \subseteq V(\tree(G))$ be a superset of the important nodes such that each path in~$\tree(G)$ from an important node to the root contains at most~$\ell$ vertices from~$R$, and consider the graph~$G / \overline{R}$. Then for each terminal~$t \in T$, for each vertex~$x$ on the outer face of~$G / \overline{R}$, we have~$d^\radial_{G / \overline{R}}(t,x) \leq 2 \ell$.
\end{lemma}
\begin{proof}
Consider an arbitrary terminal~$t \in T$ and vertex~$x$ on the outer face of~$G / \overline{R}$. Note that~$G$ and~$G/\overline{R}$ have exactly the same vertices on their outer face, since the outer face corresponds to the root of~$\tree(G)$, which is an important node by definition and whose vertices are therefore not contracted.

We first construct a radial path~$P_G$ from~$t$ to~$x$ in~$G$ and then argue it can be transformed to a bounded-length radial path in~$G / \overline{R}$. Let~$u_1$ be the node of~$\tree(G)$ for which~$t \in \kappa(u_1)$, and let~$u_1, u_2, \ldots, u_m$ with~$u_m$ the root of~$\tree(G)$ denote the simple path from~$u_1$ to the root.

By Observation~\ref{observation:child:commonface}, vertex~$t$ is incident on a common face with a vertex of~$\kappa(u_2)$, where~$u_2$ is the parent of~$u_1$ in~$\tree(G)$. Hence a radial path in~$G$ can move from~$t$ to a vertex of~$\kappa(u_2)$ in a single step. Repeated application of Observation~\ref{observation:child:commonface} yields a radial path consisting of vertices~$t = v_1 \in \kappa(u_1), v_2 \in \kappa(u_2), \ldots, v_m \in \kappa(u_m)$ with~$v_m$ on the outer face of~$G$. As~$v_m$ and~$x$ are both on the outer face, a radial path can connect them in a single step. Hence~$P_G := v_1, \ldots, v_m, x$ is a radial path in~$G$. Next we show how~$P_G$ transforms into a bounded-length radial path between~$t = v_1$ and~$x$ in~$G / \overline{R}$. 

Consider~$i < j \in [m]$ such that no nodes of~$u_i, u_{i+1} \ldots, u_{j}$ belong to~$R$. Then the construction of~$G / \overline{R}$ contracts all edges between vertices in~$\kappa(\{u_i, \ldots, u_j\})$. Each set~$\kappa(u_\ell)$ for~$i \leq \ell \leq j$ induces a connected subgraph of~$G$ by Lemma~\ref{lemma:layertree}. Since~$u_i, \ldots, u_j$ is a path in~$\tree(G)$, which is a contraction of~$G$, there is an edge between~$\kappa(u_\ell)$ and~$\kappa(u_{\ell+1})$ for each~$i \leq \ell < j$. Hence~$\kappa(\{u_i, \ldots, u_j\})$ induces a connected subgraph of~$G$ and is contracted to a \emph{single} vertex~$v_U$ when building~$G / \overline{R}$. This implies that every vertex of~$\kappa(R)$ that shares a face of~$G$ with a vertex of~$\kappa(\{u_i, \ldots, u_j\})$, will share a face with~$v_U$ in~$G / \overline{R}$.

The above argument shows that, in the radial path~$P_G$, we can replace any maximal sequence~$v_i, \ldots, v_j$ of consecutive vertices that do not belong to~$\kappa(R)$, by the single vertex~$v_U$ into which~$v_i, \ldots, v_j$ are all contracted, while ensuring that~$v_U$ shares a face of~$G / \overline{R}$ with~$v_{i-1} \in \kappa(R)$ and~$v_{j+1} \in \kappa(R)$. Since~$P_G \setminus x$ contains at most~$\ell$ vertices from~$R$ by assumption, we have that~$P_G$ itself uses at most~$\ell+1$ vertices of~$R$ and hence there are at most~$\ell$ maximal consecutive subsequences of vertices not in~$\kappa(R)$. Replacing each of these subsequences by a single vertex yields a radial path of~$(\ell + 1) + \ell$ vertices in~$G / \overline{R}$ connecting~$t = v_1$ to vertex~$x$ on the outer face. Hence~$d^\radial_{G / \overline{R}}(t,x) \leq 2\ell$.
\end{proof}

We are almost ready to state the main result of this subsection. It will show how a planar instance of \mwc can be preprocessed so that it can be cut open by a small tree in an overlay graph that contains all the neighbors of the terminals, thereby ensuring that they appear on a single face. To allow a tree containing all of~$N_G(T)$ to have size polynomial in~$k$, we need existing linear-programming based reduction rules to decrease the number of terminals and their degrees in terms of~$k$. These were presented in earlier work~\cite{CyganPPW13,GargVY04,Guillemot11a,Razgon11}, and are summarized in the following lemma.

\begin{lemma} \label{lemma:mwc:lp:preprocessing}
There is a polynomial-time algorithm that, given an instance~$(G,T,k)$ of \textsc{Planar Vertex Multiway Cut}, either correctly determines that the answer is \textsc{no} or outputs an equivalent planar instance~$(G',T',k')$ such that~$k' \leq k$,~$|T'| \leq 2k'$, and each vertex of~$T'$ has at most~$k'$ neighbors in~$G'$.
\end{lemma}
\begin{proof}
Exhaustively apply Reductions 1, 2, 3, and 5 by Cygan et al.~\cite{CyganPPW13} to reduce an instance. 

Since these reductions involve vertex deletions and edge contractions and never increase~$k$, they preserve planarity of the graph and ensure~$k' \leq k$. If~$k' \leq 0$ then the answer to the instance is \textsc{no} as all rules they employ are proven to be safe. The bound on the number of terminals in~$T'$ follows from their Lemma 2.7. Their Lemma 2.4 proves that after exhaustive reduction, for any terminal~$t \in T'$ the set~$N_{G'}(t)$ is the unique minimum~$(t, T' \setminus \{t\})$ cut in~$G'$. Hence if any terminal has more than~$k'$ neighbors, it cannot be separated from the remaining terminals by~$k$ deletions and the answer is \textsc{no}.
\end{proof}

Using Lemma~\ref{lemma:mwc:lp:preprocessing} and the contraction reductions described above, we prove the main result of this subsection. It shows how to prepare an instance to be cut open. Recall that~$\overlay(G)$ denotes an \emph{overlay} graph of~$G$.


\begin{lemma} \label{lemma:mwc:findtree}
There is a polynomial-time algorithm that, given an instance~$(G,T,k)$ of \textsc{Planar Vertex Multiway Cut}, either correctly decides its answer is \textsc{no}, or constructs a plane graph~$G'$, a terminal set~$T' \subseteq V(G')$ of size at most~$|T|$, an integer~$k' \leq k$, a vertex set~$Z \subseteq V(G') \setminus T'$, and a subgraph~$H$ of~$\overlay(G' \setminus T')$ with the following properties.
\begin{enumerate}
	\item Graph~$G$ has a $T$-multiway cut~$X \subseteq V(G) \setminus T$ of size at most~$k$ if and only if~$G'$ has a $T'$-multiway cut~$X' \subseteq V(G') \setminus (T' \cup Z)$ of size at most~$k'$.\label{prop:equivalence:forbidden}
	\item $H$ is a tree of~$\Oh((k')^5)$ edges in~$\overlay(G'\setminus T')$ of maximum degree~$\Oh(k')$ that spans all vertices of~$N_{G'}(T')$.
\end{enumerate}
\end{lemma}
\begin{proof}
Let~$(G,T,k)$ be an instance of \textsc{Planar Vertex Multiway Cut}. Apply Lemma~\ref{lemma:mwc:lp:preprocessing} to obtain the answer \textsc{no} or an equivalent instance~$(\hat{G}, \hat{T}, \hat{k})$ with~$|\hat{T}| \leq 2\hat{k} \leq 2k$ in which each terminal has degree at most~$\hat{k}$. Compute a plane embedding of~$\hat{G}$ using a polynomial-time algorithm for planar embedding (e.g.~\cite{HopcroftT74}). If~$\hat{G}$ consists of multiple connected components~$\hat{G}_1, \ldots, \hat{G}_m$, then compute a plane embedding of each component individually and draw it in the infinite (outer) face, thereby ensuring that no connected component is drawn within a face bounded by the drawing of another component.

Apply Lemma~\ref{lemma:relevant:treenodes} with parameter~$\hat{k}$ to each component~$\hat{G}_i$ of~$\hat{G}$, using~$\hat{T} \cap V(\hat{G}_i)$ as the terminal set. Each application results in a set~$R_i \subseteq \tree(\hat{G}_i)$ containing all important nodes. Let~$G'$ be the disjoint union of~$\hat{G}_i / \overline{R_i}$ for all~$i \in [m]$, with the embedding of the contracted graphs naturally inherited from the embedding of~$\hat{G}$. Let~$Z$ be the vertices of~$G'$ that were involved in edge contractions, so that~$V(G') = Z \cup (\bigcup _{i \in [m]} \kappa(R_i))$. The graph~$G'$ with vertex set~$Z$, terminal set~$T' := \hat{T}$, and budget~$k' := \hat{k}$ are output by the procedure. Since all vertices involved in edge contractions belong to~$Z$, any~$T'$-multiway cut~$X' \subseteq V(G') \setminus (T' \cup Z)$ in~$G'$ also forms a $T'$-multiway cut in~$\hat{G}$. In the other direction, if~$\hat{G}$ has a $\hat{T}$-multiway cut of size at most~$k$, then by~\eqref{pty:preserves:opt} of Lemma~\ref{lemma:relevant:treenodes} there is a $T'$-multiway cut of~$G'$ disjoint from~$Z$. By the equivalence of~$(G,T,k)$ and~$(\hat{G},\hat{T},\hat{k})$ this establishes~\eqref{prop:equivalence:forbidden}. 

To complete the proof, we describe the construction of the tree~$H$ in~$\overlay(G'\setminus T')$. The connected components of~$G'$ are~$\hat{G}_i / \overline{R_i}$ for~$i \in [m]$. The combination of Lemmas~\ref{lemma:relevant:treenodes} and~\ref{lemma:contraction:reduces:radialdist} ensures that each terminal of~$T' \cap V(G'_i)$ has radial distance~$\Oh(|T' \cap V(G'_i)| \hat{k}^3)$ to a vertex on the outer face of~$G'_i$, which is also the outer face of~$G'$. Since the overlay graph of~$G'$ has a vertex inside each face of~$G'_i$ that connects to the vertices incident to that face, a radial $uv$-path in~$G'_i$ of length~$\ell$ yields a normal $uv$-path in~$\overlay(G')$ with at most~$2\ell$ edges. By Lemma~\ref{lemma:contraction:reduces:radialdist} and the guarantee on~$R_i$ provided by Lemma~\ref{lemma:relevant:treenodes}, each terminal of~$G'_i$ has radial distance~$\Oh(|T' \cap V(G'_i)| \hat{k}^3) \in \Oh(\hat{k}^4)$ to a vertex on the outer face of~$G'_i$. Consequently, for each terminal of~$T'$ there is a path of~$\Oh(\hat{k}^4)$ edges in~$\overlay(G')$ to a vertex on the outer face of~$G'$, which can be found in polynomial time by breadth-first search in the radial graph. Since all vertices on the outer face of~$G'$ are connected by a path of two edges in~$\overlay(G')$, the union of these paths for all~$|T'| \leq 2\hat{k}$, together with connecting edges via the outer-face representative, yield a subgraph~$H'$ of~$\overlay(G')$ with~$\Oh(\hat{k}^5)$ edges that spans all vertices of~$T'$. Since~$H'$ is the union of~$\Oh(\hat{k})$ simple paths, its maximum degree is~$\Oh(\hat{k})$.

We turn~$H'$ into the desired tree~$H$ as follows. Observe that the removal of a terminal~$t \in T'$ from~$G'$ creates a face~$f_t$ incident with all vertices of~$N_{G'}(t)$. Since an overlay graph contains a vertex~$v_{f_t}$ inside this face, for each~$t \in T'$ there is a star subgraph~$S_t$ in~$\overlay(G'\setminus T')$ centered at~$v_{f_t}$ containing all vertices of~$N_{G'}(t)$, which has~$|N_{G'}(t)| \leq \hat{k}$ edges. For each~$t \in T'$, replace the edge(s) incident on~$t$ in~$H'$ by the star subgraph~$S_t$. This results in a connected subgraph of~$\overlay(G' \setminus T')$ on~$\Oh(\hat{k}^5) + \sum _{t \in T'} |N_{G'}(t)| \leq \Oh(\hat{k}^5) + \Oh(\hat{k}^2) \leq \Oh(\hat{k}^5)$ edges that spans~$N_{G'}(T')$. Since each of the star centers has degree~$\Oh(\hat{k})$ in this connected subgraph, while the degree of the other vertices increases by at most one for each terminal compared to~$H$, this connected subgraph still has maximum degree~$\Oh(\hat{k})$. We obtain the desired tree~$H$ by taking a spanning tree of this connected subgraph.
\end{proof}

\subsection{Obtaining a kernel}

\begin{theorem} \label{thm:pmwc}
\pmwc has a polynomial kernel.
\end{theorem}
\begin{proof}
Let~$(G_0, T_0, k_0)$ be an instance of \pmwc. Apply Lemma~\ref{lemma:mwc:findtree} to obtain an instance~$(G,T,k)$ with~$k \leq k_0$, along with a vertex set~$Z \subseteq V(G) \setminus T$ and a tree~$H$ on~$\Oh(k^5)$ edges of maximum degree~$\Oh(k)$ in the multigraph~$\overlay(G\setminus T)$ containing all vertices of~$N_G(T)$. Then~$G_0$ has a $T_0$-multiway cut of size at most~$k_0$ if and only if~$G$ has a $T$-multiway cut of size at most~$k$ that is \emph{disjoint} from~$Z$. We assume that~$T$ is an independent set in~$G$, as otherwise we may output a constant-size no-instance. We cut the tree~$H$ in~$\overlay(G\setminus T)$ open to produce a graph~$\hat{G}$, as follows. Create an Euler tour~$W$ of~$H$ that traverses each edge twice, in different directions, and respects the plane embedding of~$H$. Duplicate each edge of~$H$, replace each vertex~$v$ of~$H$ with~$d_H(v)$ copies of~$v$, where~$d_H(v)$ is the degree of~$v$ in~$H$, and distribute the copies of edges and vertices in the plane embedding to obtain a new face~$B$ whose boundary is a simple cycle corresponding to the Euler tour~$W$. Let~$\hat{G}$ denote the resulting graph, and embed it so that~$B$ is the outer face~$\partial \hat{G}$ of~$G$. There is a natural surjective mapping~$\pi$ from~$E(\hat{G}) \cup V(\hat{G})$ to~$E(\overlay(G\setminus T)) \cup V(\overlay(G\setminus T))$, which maps edges and vertices of~$\hat{G}$ to the edges and vertices of~$\overlay(G\setminus T)$ from which they were obtained. We denote by~$\pi^{-1}$ the inverse mapping, that associates to every edge or vertex of~$\overlay(G\setminus T)$ the set of its copies in~$\hat{G}$, and extend this notation to subsets of~$V(\overlay(G\setminus T))$ and~$E(\overlay(G\setminus T))$. Since~$H$ spans all vertices of~$N_{G}(T)$, all vertices of~$\pi^{-1}(N_G(T))$ appear on the outer face of~$\hat{G}$. 

In the following, we refer to the vertices of a path~$P$ that are not endpoints as the \emph{interior} vertices of the path. Path~$P$ \emph{avoids} a vertex set~$Y$ in its interior if the interior vertices are disjoint from~$Y$. The construction of~$\hat{G}$ gives the following property.

\begin{observation} \label{observation:connect:mwc}
Let~$X \subseteq V(G) \setminus T$ and let~$u,v \in V(G) \cap V(H)$ be distinct. Then~$u$ and~$v$ are connected by a path~$P_{uv}$ in~$(G\setminus T)\setminus X$ that avoids~$V(H)$ in its interior if and only if there exist~$\hat{u} \in \pi^{-1}(u), \hat{v} \in \pi^{-1}(v)$ that are connected by a path~$P_{\hat{u}\hat{v}}$ in~$\hat{G}[\pi^{-1}(V(G) \setminus (T \cup X))]$ that avoids~$V(\partial \hat{G}) = V(\pi^{-1}(H))$ in its interior.
\end{observation}

The following claim shows that vertices incident on a common finite face of~$\hat{G}$ are adjacent, or have a common neighbor that is a representative of a face of~$G\setminus T$.

\begin{claim}
If~$\hat{p},\hat{q}$ are distinct vertices of~$\hat{G}$ that are incident on a common finite face of~$\hat{G}$, then~$\hat{p}\hat{q} \in E(\hat{G})$ or~$N_{\hat{G}}(\hat{p}) \cap N_{\hat{G}}(\hat{q}) \cap \pi^{-1}(F(G\setminus T)) \neq \emptyset$. \label{claim:connectivity:facesreps}
\end{claim}
\begin{claimproof}
By definition of an overlay graph, if a face-representative~$f \in V(\overlay(G\setminus T)) \cap F(G\setminus T)$ shares a face of~$\overlay(G \setminus T)$ with a vertex-representative~$v \in V(\overlay(G\setminus T)) \cap V(G\setminus T)$, then~$f$ is adjacent to~$v$ in~$\overlay(G\setminus T)$. If two vertex-representatives share a common face in~$\overlay(G\setminus T)$, then they share a common face in~$G\setminus T$ and have the representative of that face as common neighbor in~$\overlay(G\setminus T)$. The claim follows from the fact that the composition of faces of~$\hat{G}$ only differs from that of~$\overlay(G\setminus T)$ by the newly created infinite face, while the claim only makes a statement about vertices sharing finite faces. 
(The correctness of the statement relies on the fact that~$\overlay(G\setminus T)$ is a multigraph if there are bridges lying on a face. When splitting the tree open to create a new face, the multi-edges to a single vertex may be distributed among several copies to provide each copy with an edge to a face-representative.)
\end{claimproof}

In the following claim we show how a separation of vertices on~$\partial \hat{G}$ corresponds to a connectivity property of the separator in the cut-open overlay graph. Later, we will use this to characterize multiway cuts in~$G$ by Steiner trees connecting terminals from~$\partial \hat{G}$. Recall that~$F(G\setminus T)$ is the set of faces of~$G\setminus T$, and that~$\overlay(G\setminus T)$ has a vertex for each such face.

\begin{claim} \label{claim:separation:condition}
For any~$X \subseteq V(G) \setminus T$ and distinct vertices~$\hat{u}, \hat{v} \in V(\partial \hat{G})$ with~$\hat{u}, \hat{v} \in \pi^{-1}(V(G) \setminus (T \cup X))$, the following are equivalent:
\begin{enumerate}
	\item There is no path between~$\hat{u}$ and~$\hat{v}$ in the graph~$\hat{G}[\pi^{-1}(V(G) \setminus (T \cup X))]$, i.e., there is no path between~$\hat{u}$ and~$\hat{v}$ in the cut-open graph~$\hat{G}$ that traverses only copies of vertices of~$G \setminus (T \cup X)$ and avoids the copies of face-representatives from~$\overlay(G\setminus T)$.\label{intersect:no-path}
	\item There is a connected component~$C$ of the graph~$\hat{G}[\pi^{-1}(X \cup F(G\setminus T))]$ containing vertices~$\hat{x}, \hat{y} \in V(\partial \hat{G}) \cap V(C)$ such that~$\hat{x}, \hat{y}, \hat{u}, \hat{v}$ are all distinct and their relative order on $\partial \hat{G}$ is~$(\hat{u}, \hat{x}, \hat{v}, \hat{y})$.\label{intersect:interior}
\end{enumerate}
\end{claim}
\begin{claimproof}\eqref{intersect:interior}$\Rightarrow$\eqref{intersect:no-path} Suppose there is such a connected component~$C$, consisting solely of vertices of the cut-open graph~$\hat{G}$ that are copies of the cut~$X$ or copies of face-vertices of the overlay graph. Let~$P$ be a path from~$\hat{x}$ to~$\hat{y}$ within~$C$. Then we can obtain a closed curve~$\gamma$ that traces along path~$P$ from~$\hat{x}$ to~$\hat{y}$, and connects~$\hat{y}$ back to~$\hat{x}$ around the outside of the disk bounded by~$\partial \hat{G}$. The relative order of~$(\hat{u}, \hat{x}, \hat{v}, \hat{y})$ ensures that the curve~$\gamma$ separates the plane into one region containing~$\hat{u}$ and another containing~$\hat{v}$. Since~$\gamma$ does not intersect the drawing of~$\hat{G}[\pi^{-1}(V(G) \setminus (T \cup X))]$, this shows there is no path between~$\hat{u}$ and~$\hat{v}$ in that graph.

\begin{figure}[tb]
\centering
\includegraphics[width=.4\linewidth]{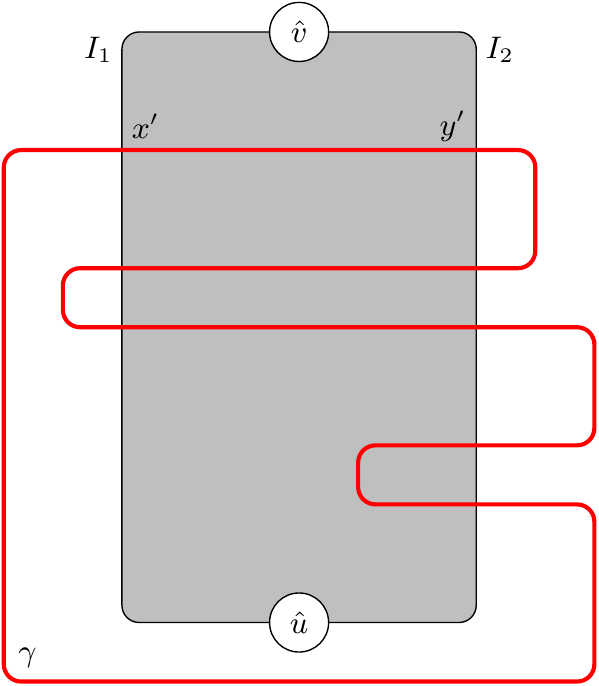}
\caption{Curve $\gamma$ separating $\hat{u}$ from $\hat{v}$ in the proof of Claim~\ref{claim:separation:condition}.}
\label{fig:II}
\end{figure}

\eqref{intersect:no-path}$\Rightarrow$\eqref{intersect:interior} Suppose there is no path between~$\hat{u}$ and~$\hat{v}$ with the prescribed properties. This means that in the subgraph of~$\hat{G}$ obtained by removing (copies of) face-representatives and removing copies of~$X$, vertices~$\hat{u}$ and~$\hat{v}$ belong to different connected components~$C_{\hat{u}}$ and~$C_{\hat{v}}$. Hence there exists a closed curve~$\gamma$ in the plane that does not intersect the drawing of~$\hat{G}[\pi^{-1}(V(G) \setminus (T \cup X))]$, such that the drawing of~$C_{\hat{u}}$ is contained in the interior of~$\gamma$ and~$C_{\hat{v}}$ is contained in the exterior. Recall that~$\partial \hat{G}$ is a simple cycle containing points~$\hat{u}$ and~$\hat{v}$. Hence the drawing of~$\partial \hat{G} \setminus \{\hat{u}, \hat{v}\}$ consists of two open intervals~$I_1, I_2$. To separate~$\hat{u}$ from~$\hat{v}$, the separating curve~$\gamma$ must contain a segment~$\gamma'$ that connects a point~$x'$ on~$I_1$ to a point~$y'$ on~$I_2$ through the interior of~$\partial \hat{G}$; see Figure~\ref{fig:II}. By locally adjusting the curve~$\gamma'$, we can obtain a curve~$\gamma^*$ through the interior of the disk bounded by~$\partial \hat{G}$ that connects a vertex~$\hat{x}$ whose drawing lies on~$I_1$, to a vertex~$\hat{y}$ whose drawing lies on~$I_2$, such that~$\gamma^*$ does not intersect the drawing of~$\hat{G}$ at edges, and does not intersect the drawing at vertices of~$\pi^{-1}(V(G) \setminus (T \cup X))$. Hence~$\gamma^*$ only intersects the drawing of~$\hat{G}$ at vertices of the form~$\pi^{-1}(X \cup F(G\setminus T))$, representing copies of~$X$ and copies of representatives of faces of~$G\setminus T$. Let~$C' \subseteq \pi^{-1}(X \cup F(G\setminus T))$ denote the vertices intersected by~$\gamma^*$. We claim that all vertices of~$C'$ belong to a common connected component of~$\hat{G}[\pi^{-1}(X \cup F(G\setminus T))]$, which will establish~\eqref{intersect:interior}. To prove the desired connectivity of~$C'$, it suffices to argue that any two vertices~$p,q$ of~$C'$ that are visited successively by~$\gamma^*$ are either adjacent in~$\hat{G}$, or have a common neighbor in~$\pi^{-1}(F(G\setminus T))$. But this follows directly from Claim~\ref{claim:connectivity:facesreps}, since successive vertices on~$C'$ share a finite face of~$\hat{G}$.
\end{claimproof}

The following claim shows how $T$-multiway cuts of~$G$ relate to Steiner trees in~$\hat{G}$.

\begin{claim} \label{claim:mwc:characterization}
For any~$X \subseteq V(G) \setminus T$, the following are equivalent:
\begin{itemize}
	\item $X$ is a $T$-multiway cut in~$G$.
	\item For any pair of distinct terminals~$t, t' \in T$, for any path~$P$ in~$G$ from a vertex in~$N_G(t)$ to a vertex in~$N_G(t')$, at least one of the following holds:
	\begin{enumerate}[i.]
		\item Set~$X$ contains a vertex of~$V(P) \cap V(H)$.\label{intersect:at:boundary}
		\item Path~$P$ contains distinct vertices~$u, v \in V(H)$, such that the subpath of~$P$ between~$u$ and~$v$ contains no other vertices of~$V(H)$, and for any~$\hat{u} \in \pi^{-1}(u), \hat{v} \in \pi^{-1}(v)$ there is a connected component~$C$ of the graph~$\hat{G}[\pi^{-1}(X \cup F(G\setminus T))]$ containing vertices~$\hat{x}, \hat{y} \in V(\partial \hat{G}) \cap V(C)$ such that~$\hat{x}, \hat{y}, \hat{u}, \hat{v}$ are all distinct and their relative order on $\partial \hat{G}$ is~$(\hat{u}, \hat{x}, \hat{v}, \hat{y})$.\label{intersect:by:steinertree}
	\end{enumerate}
\end{itemize}
\end{claim}
\begin{claimproof}
Observe that if~$X$ fails to be a $T$-multiway cut in~$G$, then there are distinct terminals~$t,t' \in T$ in the same connected component of~$G\setminus X$, implying there is a path~$P$ from~$N_G(t) \subseteq V(H)$ to~$N_G(t') \subseteq V(H)$ in~$G\setminus T$ that is not intersected by~$X$. Hence~$V(P) \cap V(H) \cap X = \emptyset$, and if we enumerate the vertices of~$V(P) \cap V(H)$ by~$x_1 \in N_G(t), \ldots, x_n \in N_G(t')$ in their natural order along~$P$, then each successive pair~$x_i, x_{i+1}$ is connected by a subpath of~$P$ in~$(G\setminus T)\setminus X$ whose internal vertices are disjoint from~$V(H)$.
With this in mind, the following series of equivalences proves the claim.
\begin{description}
	\item[$\quad$] The set~$X \subseteq V(G) \setminus T$ is a $T$-multiway cut of~$G$.
	\item[$\Leftrightarrow$] For any distinct~$t, t' \in T$, set~$X$ hits each path~$P_{tt'}$ in~$G\setminus T$ connecting~$N_G(t)$ to~$N_G(t')$.
	\item[$\Leftrightarrow$] For each such path~$P_{tt'}$, set~$X$ contains a vertex of~$V(P_{tt'}) \cap V(H)$ or there is a subpath~$P_{uv}$ of~$P_{tt'}$ whose interior avoids~$V(H)$ that starts and ends at distinct vertices~$u,v \in V(H)$, such that~$(G\setminus T)\setminus X$ contains no path from~$u$ to~$v$ that avoids~$V(H)$ in its interior.
	\item[$\Leftrightarrow$] For each such path~$P_{tt'}$, set~$X$ contains a vertex of~$V(P_{tt'}) \cap V(H)$ or there is a subpath~$P_{uv}$ as above, such that for all~$\hat{u} \in \pi^{-1}(u)$ and~$\hat{v} \in \pi^{-1}(v)$, there is no path between~$\hat{u}$ and~$\hat{v}$ in the graph~$\hat{G}[\pi^{-1}(V(G) \setminus (T \cup X))]$ that avoids~$V(\partial \hat{G}) = \pi^{-1}(V(H))$ in its interior. (We use Observation~\ref{observation:connect:mwc}.)
\end{description}
By the equivalence of Claim~\ref{claim:separation:condition}, this concludes the proof.
\end{claimproof}

We will use Claim~\ref{claim:mwc:characterization} to argue for the correctness of the kernelization later. We now prepare for the sparsification step. Rather than invoking Theorem~\ref{thm:sparse} directly, we need to perform one additional step to encode the requirement that in the instance~$(G,T,k)$, we are looking for a multiway cut \emph{that avoids vertex set~$Z$}. So we have to ensure that the connecting subgraphs preserved by Theorem~\ref{thm:sparse} do not use vertices of~$Z$. To achieve that, we will subdivide the edges incident on vertices of~$Z$ so that using connections over~$Z$ would have prohibitively large cost, implying that low-cost Steiner trees avoid~$Z$.

Formally, obtain a plane graph~$G^*$ from~$\hat{G}$ as follows. Fix a constant~$\alpha$ such that the maximum degree of the tree~$H$ we cut open is strictly less than~$\alpha \cdot k$, which exists by the degree-bound~$\Oh(k)$ for~$H$. Since the number of copies into which a vertex is split when cutting open~$H$ is bounded by its degree in~$H$, this implies~$|\pi^{-1}(v)| < \alpha \cdot k$ for all~$v \in V(\overlay(G \setminus T))$.

For each edge~$uv \in E(\hat{G})$ for which at least one endpoint belongs to~$Z$, replace the direct edge~$uv$ by a path of~$\alpha \cdot k^2$ new vertices. Refer to the vertices inserted in this step as~$Z'$. To apply Theorem~\ref{thm:sparse}, we need to interpret~$G^*$ as a plane partitioned graph. Since the outer face of~$\hat{G}$ was a simple cycle, so is the outer face of~$G^*$. Since~$\overlay(G\setminus T)$ is connected, so are~$\hat{G}$ and~$G^*$. The face-representatives~$F(G\setminus T)$ formed an independent set in~$\overlay(G\setminus T)$, and therefore~$\pi^{-1}(F(G\setminus T))$ forms an independent set in~$\hat{G}$ and therefore~$G^*$. Hence~$G^*$ is a plane partitioned graph with~$\black(G^*) := Z' \cup \pi^{-1}(V(G\setminus T))$ and~$\green(G^*) := \pi^{-1}(F(G\setminus T))$. Since the tree~$H$ we cut open to produce~$\hat{G}$ has~$\Oh(k^5)$ edges, we have~$|\partial \hat{G}| \leq \Oh(k^5)$. The transformation to~$G^*$ increases this by at most a factor~$k^2$, so~$|\partial G^*| \leq \Oh(k^7)$.

Apply Theorem~\ref{thm:sparse} to~$G^*$ to obtain a sparsifier subgraph~$\widetilde{G}$ on~$\Oh(|\partial G^*|^{212}) \leq \Oh(k^{1484})$ edges. Let~$D := V(G \setminus (T \cup Z)) \cap \pi(V(\widetilde{G}) \setminus Z')$ be the non-terminal non-$Z$ vertices of the original graph~$G$ for which a copy was selected in the sparsifier~$\widetilde{G}$, so that~$|D| \leq \Oh(k^{1484})$. The key to the correctness of the kernelization will be the following.

\begin{claim} \label{claim:sparsifier:contains:solution}
If~$G$ has a $T$-multiway cut~$X \subseteq V(G) \setminus (T \cup Z)$ of size at most~$k$, then~$G$ has a $T$-multiway cut~$X' \subseteq D \subseteq V(G) \setminus (T \cup Z)$ of size at most~$k$.
\end{claim}
\begin{claimproof}
Suppose~$X$ is a~$T$-multiway cut~$X \subseteq V(G) \setminus (T \cup Z)$ of size at most~$k$. We construct the desired cut~$X'$. Let~$C_1, \ldots, C_m$ be the connected components of~$\hat{G}[\pi^{-1}(X \cup F(G\setminus T))]$ that contain at least one vertex of~$\pi^{-1}(X \cup F(G \setminus T)) \cap V(\partial \hat{G})$. Note that~$C_1, \ldots, C_m$ could equivalently have been defined by replacing~$\hat{G}$ with~$G^*$: this makes no difference since~$X \cap Z = \emptyset$. Hence~$C_i \cap (\pi^{-1}(Z) \cup Z') = \emptyset$ for~$1\leq i\leq m$.

For each~$i$, consider component~$C_i$ and define~$S_i := V(C_i) \cap V(\partial \hat{G}) = V(C_i) \cap V(\partial G^*)$ as the vertices from the outer face contained in the component. Vertices in~$S_i$ correspond to vertices and face-representatives of~$G \setminus T$. Recall that the \emph{cost} of a subgraph~$C_i$ of the plane partitioned graph~$G^*$ was defined as~$|V(C_i) \cap \black(G^*)|$. Since~$C_i$ contains no vertices of~$Z'$, while each vertex~$x \in X \cap V(G\setminus T)$ contributes at most~$|\pi^{-1}(x)| < \alpha \cdot k$ vertices to the cost, the cost of~$C_i$ is strictly less than~$\alpha \cdot k^2$. Since~$C_i$ connects~$S_i$ in~$G^*$, Theorem~\ref{thm:main} guarantees there is a connected subgraph~$A_i$ of~$\widetilde{G}$ with~$\cost(A_i) \leq \cost(C_i)$ that connects~$S_i$. We claim that~$A_i \setminus (Z' \cup \pi^{-1}(Z))$ still connects~$S_i$ in~$G^*$ (and therefore in~$\hat{G}$). This follows from the fact that vertices of~$Z'$ each contribute one to the cost measure, and form degree-two paths of~$\alpha \cdot k^2$ vertices in~$G^*$. Therefore a subgraph of cost less than~$\alpha \cdot k^2$ cannot contain an entire chain of such vertices, and since the vertices of~$Z$ are buffered from the remainder of the graph by such chains, no vertices of~$Z$ appear in~$A_i$.

Using these sets~$A_i$, we define~$X' := \bigcup _{i \in [m]} \pi(A_i \setminus (Z' \cup \pi^{-1}(Z))) \cap V(G \setminus T)$: it consists of the vertices in~$\overlay(G \setminus T)$ that represent vertices of~$G \setminus T$ not belonging to~$Z$, for which at least one copy in the cut-open graph was selected in a connector~$A_i$. By construction we have~$X' \subseteq D \subseteq V(G) \setminus (T \cup Z)$. It remains to prove that~$|X'| \leq |X|$ and that~$X'$ is a $T$-multiway cut.

For the size bound, let us consider how~$|X'|$ relates to~$\sum _{i \in [m]} \cost(A_i)$. For a vertex~$x \in V(G \setminus T)$, let~$f_C(x) := \sum _{i \in [m]} |\pi^{-1}(x) \cap V(C_i)|$ denote the number of occurrences of a copy of~$x$ in a subgraph~$C_i$, and similarly let~$f_A(x) := \sum _{i \in [m]} |\pi^{-1}(x) \cap V(A_i)|$ denote the number of occurrences of a copy of~$x$ in a subgraph~$A_i$.

Since~$x \in X'$ contributes~$f_A(x)$ to~$\sum _{i \in [m]} \cost(A_i)$ but only one to~$|X'|$, we have
$$|X'| = \left (\sum _{i \in [m]} \cost(A_i) \right) - \left (\sum _{x \in X'} f_A(x) - 1 \right).$$
Similarly,
$$|X| = \left (\sum _{i \in [m]} \cost(C_i) \right) - \left (\sum _{x \in X} f_C(x) - 1 \right).$$ 
Now, since the~$A_i$ are minimum-cost connecting subgraphs, we have $\sum _{i \in [m]} \cost(A_i) \leq \sum _{i \in [m]} \cost(C_i)$. For the second term, observe that the only vertices~$x \in X$ for which~$f_C(x) - 1 > 0$ belong to the tree~$H$ that was cut open. Indeed, if~$x$ does not belong to~$V(H)$, then~$|\pi^{-1}(x)| = 1$, and since at most one connected component~$C_i$ contains~$x$ (for~$x$ would merge two such components together) we have~$f_C(x) - 1 = 0$. But for~$x \in X \cap V(H)$, all vertices of~$\pi^{-1}(x)$ are on~$\partial \hat{G}$ and therefore on~$\partial G^*$. Hence for each occurrence of~$x' \in \pi^{-1}(x)$ in a component~$C_i$, we have~$x' \in S_i$ and the replacement component~$A_i$ also contains~$x'$. So we have
$$\sum _{x \in X'} f_A(x) - 1 \geq \sum _{x \in X} f_C(x) - 1,$$
which together with the above proves that~$|X'| \leq |X|$.

As the last step of the proof of Claim~\ref{claim:sparsifier:contains:solution}, we prove that~$X'$ is a $T$-multiway cut in~$G$. By Claim~\ref{claim:mwc:characterization} it suffices to prove that for each pair of distinct terminals~$t, t' \in T$, for any path~$P$ in~$G \setminus T$ from a vertex in~$N_G(t)$ to a vertex in~$N_G(t')$, one of conditions \eqref{intersect:at:boundary}--\eqref{intersect:by:steinertree} from Claim~\ref{claim:mwc:characterization} holds. Consider such a path~$P$. Since~$X$ is a multiway cut, one of \eqref{intersect:at:boundary}--\eqref{intersect:by:steinertree} holds for~$X$. We conclude by a case distinction.

\begin{itemize}
	\item Suppose~\eqref{intersect:at:boundary} holds for~$X$ since~$x \in X \cap V(P) \cap V(H)$. Then~$\pi^{-1}(x)$ consists of one or more vertices on~$\partial \hat{G}$ and therefore of~$\partial G^*$, so that there is at least one connected component~$C_i$ containing a vertex of~$\pi^{-1}(x)$, implying~$\pi^{-1}(x) \cap S_i \neq \emptyset$. Since the replacement subgraph~$A_i$ contains~$S_i$, it follows that~$A_i$ contains a vertex of~$\pi^{-1}(x)$ and therefore~$x \in X'$. Hence~\eqref{intersect:at:boundary} also holds for~$X'$.
	\item Suppose~\eqref{intersect:by:steinertree} holds for~$X$: there is a subpath~$P'$ of~$P$ between distinct vertices~$u,v \in V(H)$ that contains no other vertices of~$V(H)$, such that for any~$u' \in \pi^{-1}(u)$ and~$v' \in \pi^{-1}(v)$ there is a connected component $C$ of the graph~$\hat{G}[\pi^{-1}(X \cup F(G\setminus T))]$ containing~$x'$ and~$y'$ satisfying the given ordering condition. But for any choice of~$u'$ and~$v'$, the witness vertices~$x'$ and~$y'$ belong to a common connected component~$C_i$ considered in our construction above. As~$x'$ and~$y'$ appear on the outer face of~$\hat{G}$ and therefore~$G^*$, they appear in a common set~$S_i$. Hence the replacement subgraph~$A_i$ connects both~$x'$ and~$y'$, showing that~$X'$ satisfies \eqref{intersect:by:steinertree}.
\end{itemize}

This completes the proof of Claim~\ref{claim:sparsifier:contains:solution}.
\end{claimproof}

Using Claim~\ref{claim:sparsifier:contains:solution} we can finally shrink the graph to size polynomial in~$k$. Obtain a graph~$G_1$ from~$G$ by contracting all edges between vertices of~$V(G) \setminus (T \cup D)$. Partition~$V(G_1)$ into~$T \uplus D \uplus U_1$, where~$U_1$ are the vertices resulting from the contractions of non-terminal non-$D$ vertices. 

\begin{observation}
Graph~$G$ has a $T$-multiway cut of size at most~$k$ disjoint from~$Z$ if and only if~$G_1$ has a $T$-multiway cut of size at most~$k$ that is disjoint from~$U_1$.
\end{observation}

Initialize~$U_2$ as a copy of~$U_1$, and shrink it by exhaustively applying the following reduction rule: while there are distinct~$u,v \in U_2$ with~$N_{G_1}(u) \subseteq N_{G_1}(v)$, then remove vertex~$u$ from~$U_2$. Let~$G_2 := G_1[T \cup D \cup U_2]$. Since the removed vertices~$u$ can always be bypassed using the preserved vertices~$v$, we have the following.

\begin{observation}
Graph~$G_1$ has a $T$-multiway cut of size~$k$ disjoint from~$U_1$ if and only if~$G_2$ has one that is disjoint from~$U_2$.
\end{observation}

\begin{claim} \label{claim:utwo:bound}
$|U_2| \leq \Oh(k^{1484})$.
\end{claim}
\begin{claimproof}
Derive a planar bipartite graph~$Q$ from~$G_2$ with bipartition into~$X := (T \uplus D)$ and~$Y := U_2$. Since~$U_1$ was an independent set in~$G_1$, while vertices of~$U_2$ have pairwise incomparable neighborhoods in~$G_2$, it follows that~$N_Q(u) \not \subseteq N_Q(v)$ for distinct~$u,v \in U_2$. Hence by Lemma~\ref{lemma:bipartite:neighborhood:count} we have~$|U_2| = |Y| \leq 5|X| \leq \Oh(k^{1484})$.
\end{claimproof}

We find that~$G_2$ on vertex set~$T \uplus D \uplus U_2$ has~$\Oh(k^{1484})$ vertices. To obtain the final kernelized instance of the original \pmwc problem, we replace vertices~$U_2$ which a solution is not allowed to delete, by grid-like substructures whose interconnections simply cannot be broken by a budget of~$k$ deletions. Observe that the vertex set~$U_2$ is an independent set in~$G_2$. Obtain graph~$G_3$ from~$G_2$ as follows: for each~$u \in U_2$, let~$v_u^1, \ldots, v_u^{d_u}$ be the neighbors of~$u$ in their cyclic order around~$u$ in the embedding. Replace~$u$ by a grid with~$k+1$ rows and~$(k+1)d_u$ columns. Let~$x_{u,1}, \ldots, x_{u, (k+1)d_u}$ be the vertices of the bottom row of this grid, and insert edges between~$v_u^i$ and~$x_{u,(k+1)(i-1) + \ell}$ for all~$1 \leq i \leq d_u$ and~$1 \leq \ell \leq k+1$. Due to our choice of ordering, these edges can be drawn planarly.

\begin{claim} \label{claim:equivalence:g3}
Graph~$G_2$ has a $T$-multiway cut of size at most~$k$ that is disjoint from~$U_2$, if and only if~$G_3$ has a $T$-multiway cut of size at most~$k$.
\end{claim}
\begin{claimproof}
In the forward direction, any $T$-multiway cut~$X \subseteq V(G_2) \setminus (T \cup U_2)$ is also a $T$-multiway cut in~$G_3$, since any path over an inserted grid in~$G_3$ can be replaced by a vertex of~$U_2$ to provide an equivalent path in~$G_2$. The reverse direction is more interesting.

Consider a vertex~$u \in U_2$. For any two distinct neighbors~$v_u^i, v_u^j \in N_{G_2}(u)$, the grid that was inserted into~$G_3$ to replace vertex~$u$ contains~$k+1$ pairwise internally vertex-disjoint paths~$P_1, \ldots, P_{k+1}$ between~$v_u^i$ and~$v_u^j$. Each row of the grid supports one such path; a path starts at~$v_u^i$, moves to a neighbor at the bottom row of the grid, moves up to the appropriate row, moves horizontally through the grid, moves down to the neighbor of~$v_u^j$, and ends in~$v_u^j$. The leftmost neighbor of~$v_u^i$ in the grid connects to the rightmost neighbor of~$v_u^j$ over the top row, so that the paths form a nested structure in the grid. 

Using these~$k+1$ pairwise internally vertex-disjoint paths in~$G_3$ between any pair~$v_u^i, v_u^j \in N_{G_2}(u)$, we complete the proof. Suppose that~$X \subseteq V(G_3) \setminus T$ is a $T$-multiway cut in~$G_3$ of size at most~$k$. We show that~$X' := (X \cap V(G_2)) \setminus U_2$ is a $T$-multiway cut in~$G_2$ of size at most~$k$. Assume for a contradiction that~$G_2 \setminus X'$ contains a path~$P$ between distinct terminals~$t,t'$. For every occurrence of a vertex~$u \in U_2$ on~$P$, the predecessor~$v_u^i$ and successor~$v_u^j$ of~$u$ on~$P$ are connected by~$k+1$ internally vertex-disjoint paths in~$G_3$. Hence~$X$ avoids at least one of these paths, showing that~$v_u^i$ and~$v_u^j$ are also connected in~$G_3 \setminus X$. But by replacing each occurrence of a vertex from~$U_2$ by a path through a replacement grid that is disjoint from~$X$, we obtain a~$tt'$-path in~$G_3$ that is not intersected by~$X$; a contradiction.
\end{claimproof}

Claim~\ref{claim:equivalence:g3} is the last link in a chain of equivalences, which shows that the answer to the original input~$(G_0, T_0, k_0)$ is identical to the answer to~$(G_3, T, k)$. Each step of the transformation can be carried out in polynomial time. It remains to bound the size of~$G_3$. Its vertex set consists of~$T$ (size at most~$2k$), of~$D$ (size~$\Oh(k^{1484})$), and the vertices of the grids inserted to replace members of~$U_2$. Observe that the grid to replace a vertex~$u \in U_2$ consists of~$(k+1)^2 d_u$ vertices, where~$d_u$ is the degree of~$u$ in~$G_2$. The total number of vertices in replacement grids is therefore~$\sum _{u \in U_2} (k+1)^2 d_u = (k+1)^2 \sum _{u \in U_2} d_u$. Note that~$\sum _{u \in U_2} d_u$ is exactly the number of edges in the planar bipartite graph~$Q$ defined in the proof of Claim~\ref{claim:utwo:bound}. Since~$Q$ has~$\Oh(k^{1484})$ vertices, while the number of edges in a bipartite planar graph is at most twice the number of vertices, we have~$\sum _{u \in U_2} d_u = |E(Q)| \leq \Oh(k^{1484})$. Hence the number of vertices in replacement grids is~$\Oh(k^{1486})$, giving a total bound of~$\Oh(k^{1486})$ on the number of vertices in~$G_3$. This completes the proof of Theorem~\ref{thm:pmwc}.
\end{proof}

\section{Reductions to \vtxplan{}}\label{sec:lb}
In this section we show two reductions from \pmwc{} (defined in Section~\ref{app:mwc}): one to the disjoint version of \vtxplan, and one to the regular one.
%
We start with recalling formal problem definitions.

\defparproblem{\vtxplan}{A graph $G$ and an integer $k$.}{k}{Does there exist
  a set $X \subseteq V(G)$ such that $G\setminus X$ is planar?}

\defparproblem{\vtxplandis}{A graph $G$, a set $S \subseteq V(G)$ such that $G\setminus S$ is planar,
  and an integer $k$.}{k+|S|}{Does there exist a set $X \subseteq V(G) \setminus S$ of size
    at most $k$ such that $G\setminus X$ is planar?}

In the next two subsections,
we show polynomial-parameter transformations from \pmwc{} to \vtxplandis{} 
and \vtxplan{}.

Both reductions rely on the same idea: if a \vtxplan{} instance contains a large grid,
the budget of $k$ deletions is not able to effectively break it, and 
there is essentially only one way to embed it in the plane. If some parts of the graph are attached
to vertices of the grid incident to faces far away from each other, a solution to \vtxplan{}
needs to separate such parts from each other. This allows to embed a \pmwc{} instance.

Formally, we rely on the following observation (see Figure~\ref{fig:k5}).
\begin{observation}\label{obs:k5}
Consider the following graph $H_0$: we start with $H_0$ being a $4 \times 4$ grid
with vertices $x_{a,b}$, $1 \leq a,b \leq 4$ (i.e., the vertex $x_{a,b}$ lies in $a$-th
    row and $b$-th column of the grid)
and then add an edge $x_{2,2}x_{2,4}$ but delete edges $x_{1,2}x_{2,2}$ and $x_{1,4}x_{2,4}$.
Then $H_0$ contains a $K_5$ minor and is therefore not planar.
\end{observation}

\begin{figure}[tb]
\begin{center}
\includegraphics{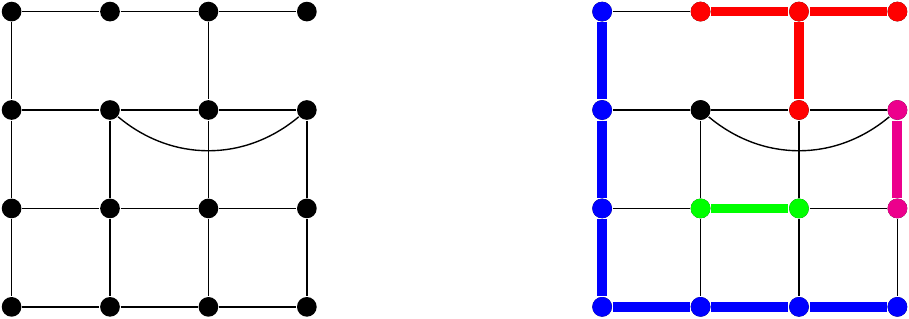}
\caption{The graph $H_0$ of Observation~\ref{obs:k5} with the $K_5$ minor model on the right.}
\label{fig:k5}
\end{center}
\end{figure}

\subsection{From \pmwc{} to \vtxplandis{}}

\begin{lemma}\label{lem:pmwc2vtxplandis}
Given a \pmwc{} instance $(G,T,k)$, one can in $\Oh(|V(G)|+|E(G)|+k)$ time compute
an equivalent \vtxplandis{} instance $(G',S,k)$ with $|S| \leq 8|T|$.
\end{lemma}

\begin{figure}[tb]
\begin{center}
\includegraphics[width=\linewidth]{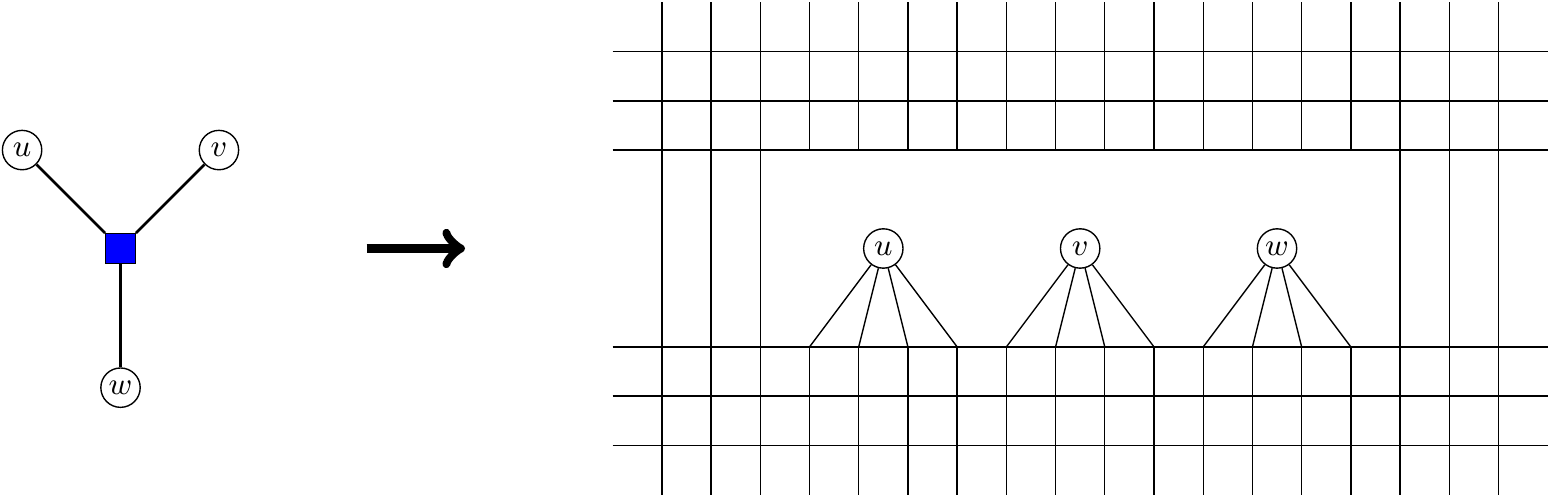}
\caption{Embedding neighbors of a terminal (blue square) into a hole cut out in a large grid. Every neighbor of a terminal
  is connected to $k+1$ vertices of the grid ($k+1=4$ in the figure).}
\label{fig:lb}
\end{center}
\end{figure}

\begin{proof}
If $|T| \leq 1$, then the input instance is trivial, and
we can output $G' = S = \emptyset$.
Otherwise, let $T = \{t_1,t_2, \ldots, t_{|T|}\}$.
We start by constructing a $4 \times 2|T|$ grid $H$. Denote $S = V(H)$; note
that $|S| = 8|T|$ as promised.
For $1 \leq i \leq |T|$, let $x_i$ be the $(2i)$-th vertex in the second row of $H$.
We construct the graph $G'$ from $G \uplus H$ by identifying $t_i$ with $x_i$
for every $1 \leq i \leq |T|$.
We claim that the resulting \vtxplandis{} instance $(G',S,k)$
is equivalent to the input \pmwc{} instance $(G,T,k)$.
Note that $V(G) \setminus T = V(G') \setminus S$.

In one direction, let $X \subseteq V(G) \setminus T$
be a solution to \pmwc{} on $(G,T,k)$. We show that $X$ is also a solution
to \vtxplandis{} on $(G',S,k)$ by showing a planar embedding of $G'\setminus X$.
First, embed $H$ in the natural way. 
Second, for every connected component $C$ of $G\setminus X$, proceed as follows.
If $C$ contains a terminal $t_i$, then fix a planar embedding of $C$ that keeps 
$t_i$ incident to the infinite face, and embed $C$ in one of the faces of $H$
incident with $x_i$.
Otherwise, if $C$ does not contain any terminal, embed $C$ in the infinite face of $H$.
Since every connected component $C$ contains at most one terminal, this is a valid
planar embedding of $G'\setminus X$.

In the other direction, let $X \subseteq V(G') \setminus S$ be a solution
to \vtxplandis{} on $(G',S,k)$. We claim that $X$ is also a solution to \pmwc{}
on $(G,T,k)$. Assume the contrary; since $|X| \leq k$ and~$X \subseteq V(G') \setminus S = V(G) \setminus T$ by assumption, 
we have two terminals $t_i,t_j \in T$ and a $t_i-t_j$ path $P$ in $G\setminus X$.
Consider the subgraph $H \cup P$ of $G'\setminus X$ and contract $P$ to a single edge $t_it_j$.
Then, this minor of $G'\setminus X$ contains $H_0$ from Observation~\ref{obs:k5}
as a minor.
By Observation~\ref{obs:k5}, $G'\setminus X$ contains $K_5$ as a minor, contradicting its planarity.
\end{proof}

\subsection{From \pmwc{} to \vtxplan{}}

\begin{lemma}\label{lem:pmwc2vtxplan}
Given a \pmwc{} instance $(G,T,k)$, one can in polynomial time compute
an equivalent \vtxplan{} instance $(G',k)$ with
$|E(G')|+|V(G')| \leq \Oh(k(|E(G)| + |V(G)|))$.
\end{lemma}

\begin{proof}
We proceed as in the proof of Lemma~\ref{lem:pmwc2vtxplandis}, but
we need to make $H$ thicker in order not to allow any tampering.

If $|T| \leq 1$, then the input instance is trivial, and
we can output $G' = \emptyset$.
Similarly, we output a trivial no-instance if two terminals of $T$ are adjacent.
Otherwise, fix a planar embedding $\phi$ of $G$ and let $T = \{t_1,t_2, \ldots, t_{|T|}\}$.
For every $1 \leq i \leq |T|$, let $d_i$ be the degree of $t_i$ in $G$
and let $v_i^1, \ldots, v_i^{d_i}$ be the neighbors of $t_i$ in $G$
in clockwise order around~$t_i$ in~$\phi$.
Let $D = \sum_{i=1}^{|T|} d_i$.

We define a graph $H$ as follows. We start with $H$ being a $4(k+1) \times (D+|T|)(k+1)$-grid
with vertices $x_{a,b}$, $1 \leq a \leq 4(k+1)$, $1 \leq b \leq (D+|T|)(k+1)$
(i.e., the vertex $x_{a,b}$ lies in $a$-th row and $b$-th column).
For every $1 \leq i \leq |T|$, let $b_i^\leftarrow = (i + \sum_{j < i} d_j)(k+1)$
and $b_i^\rightarrow = b_i^\leftarrow + d_i(k+1)$; additionally, let $b_0^\rightarrow = 0$.
For every $1 \leq i \leq |T|$ 
and every $b_i^\leftarrow < b \leq b_i^\rightarrow$,
we delete from $H$ the edge $x_{k+1,b}x_{k+2,b}$; see Figure~\ref{fig:lb}.

We now define the graph $G'$ as follows. We start with $G' = H \uplus (G\setminus T)$.
Then, for every $1 \leq i \leq |T|$ and every $1 \leq j \leq d_i$,
we make $v_i^j$ adjacent to $x_{k+2,b}$ for every
$b_i^\leftarrow + (j-1)(k+1) < b \leq b_i^\leftarrow + j(k+1)$. 
This finishes the construction of the \vtxplan{} instance $(G',k)$.
We now show that it is equivalent to \pmwc{} on $(G,T,k)$.

In one direction, let $X$ be a solution to \pmwc{} on $(G,T,k)$. We show
that $X$ is also a solution to \vtxplan{} on $(G',k)$ by constructing a planar
embedding of $G'\setminus X$. First, we embed $H$ naturally
and for every $1 \leq i \leq |T|$
let $f_i$ be the face of the embedding that is incident with vertices
$x_{k+2,b}$ for every $b_i^\leftarrow < b \leq b_i^\rightarrow$.
Then, for every connected component $C$ of $G\setminus X$ we proceed as follows.
If $C$ does not contain a terminal, then since~$X \cap T = \emptyset$, component~$C$ contains no neighbors of terminals either; hence the vertices of~$C$ are not adjacent to~$H$ in~$G'$. We embed~$C$ in the infinite face of $H$.
Otherwise, assume that the only terminal of $C$ is $t_i$. 
We take the embedding of $C$ induced by~$\phi$, change the infinite face so that
$t_i$ is incident with the infinite face, and embed $C\setminus t_i$ with the induced embedding
into $f_i$. The fact that $v_i^1,\ldots,v_i^{d_i}$ are embedded around $t_i$ in $\phi$
in this order allows us now to draw all edges between vertices of $N_G(t_i)$
and $\{x_{k+2,b} | b_i^\leftarrow < b \leq b_i^\rightarrow\}$ in a planar fashion.

In the other direction, let $X'$ be a solution to \vtxplan{} on $(G',k)$.
We claim that $X := X' \cap (V(G) \setminus T)$ is a solution to \pmwc{} on $(G,T,k)$.
If this is not the case, then there exist
two terminals $t_{i_1},t_{i_2}$, $1 \leq i_1 < i_2 \leq |T|$ and a path $P$
from $t_{i_1}$ to $t_{i_2}$ in $G\setminus X$. Let $v_{i_1}^{j_1}$ be the neighbor of $t_{i_1}$
on $P$ and $v_{i_2}^{j_2}$ be the neighbor of $t_{i_2}$ on $P$.
Since $|X'| \leq k$, there exist:
\begin{itemize}
\item indices $1 \leq a_1 \leq k+1$, $k+2 \leq a_2 \leq 2k+2$, $2k+3 \leq a_3 \leq 3k+3$, 
  $3k+4 \leq a_4 \leq 4k+4$ such that no vertex of $X'$ is in rows numbered $a_1$, $a_2$, $a_3$, nor $a_4$ of $H$;
\item for every $1 \leq i \leq |T|$, an index $b_{i-1}^\rightarrow < b_i \leq b_i^\leftarrow$
with no vertex of $X'$ in the $b_i$-th column of $H$; and
\item for every $1 \leq i \leq |T|$ and every $1 \leq j \leq d_i$ an index 
$b_i^\leftarrow + (j-1)(k+1) < b_i^j \leq b_i^\leftarrow + j(k+1)$ with no vertex of $X'$
in the $b_i^j$-th column of $H$.
\end{itemize}
We conclude by observing that the graph $H_0$ from Observation~\ref{obs:k5} is a minor of a subgraph
of $G'\setminus X$ induced by $P$, the $a_1$-th, $a_2$-th, $a_3$-th, and $a_4$-th rows of $H$,
   and columns of $H$ with numbers $b_{i_1}$, $b_{i_1}^{j_1}$, $b_{i_2}$, $b_{i_2}^{j_2}$.
\end{proof}

\section{Conclusions}
We conclude with several open problems.
First, the exponents in the polynomial bounds of our kernel sizes are enormous, similarly as for planar \textsc{Steiner tree}~\cite{pst-kernel}.
Thus, we reiterate the question of reducing the bound of the main sparsification routine of~\cite{pst-kernel} to quadratic. 
Second, we hope that our tools can pave the way to a polynomial kernel for \textsc{Vertex Planarization}, which remains an important open problem.
Third, after showing a polynomial kernel for \pmwc{}, it is natural to ask
about the kernelization status of its generalization, \textsc{Plane Multicut},
      both in the edge- and vertex-deletion variants.
Fourth, nothing is known about the kernelization of \textsc{Multiway Cut} parameterized above the LP lower bound~\cite{CyganPPW13},
  even in the case of planar graphs and edge deletions.

\bibliography{refs}

\end{document}